\documentclass[10pt]{article}
\usepackage[left=1in,right=1in,top=1in,bottom=1in]{geometry}

\usepackage[utf8]{inputenc} 
\usepackage{xr} 
\usepackage{tipa}
\usepackage{comment}
\usepackage{lipsum}
\usepackage[dvipsnames]{xcolor}
\usepackage{bm} 
\usepackage{eufrak}

\usepackage{tipa,tikz}
\usetikzlibrary{positioning, arrows.meta, fit, backgrounds}

\definecolor{MyBlue}{RGB}{0,51,102}   
\definecolor{FocusColor}{RGB}{180,70,20}
\definecolor{MyBrown}{RGB}{120,70,15}
\definecolor{ForestGreen}{RGB}{34, 139, 34}

\usepackage{algorithmicx}
\usepackage{algorithm}
\usepackage{algpseudocode}
\algnewcommand\algorithmicinput{\textbf{Input:}}
\algnewcommand\Input{\item[\algorithmicinput]}
\algnewcommand\algorithmicoutput{\textbf{Output:}}
\algnewcommand\Output{\item[\algorithmicoutput]}
\usepackage{upgreek}
\usepackage{enumitem}
\newcommand{\intmeasure}{em}
\newcommand{\interval}[2][\relax]{%
  \underbrace{\ifx#1\relax
      \rule{2\intmeasure}{0pt}
    \else
      \rule{#1\intmeasure}{0pt}
    \fi}_{\text{#2}}
}

\usepackage{amsthm, amsfonts, amssymb, amsmath,enumitem, bbm, bm, mathabx,mathrsfs,bigints}
\usepackage{accents}
\usepackage{mathtools}
\mathtoolsset{showonlyrefs,showmanualtags}
\numberwithin{equation}{section}

\newcommand{\dnorm}{\mathsf{N}}

\newcommand{\fancyname}{DAIF}
\newcommand{\ols}{\mathrm{ols}}
\newcommand{\train}{\mathsf{trn}}
\newcommand{\test}{\mathsf{tst}}
\newcommand{\pred}{\mathsf{pred}}
\newcommand{\ksf}{\bm{\mathsf{K}}}
\newcommand{\ck}{\bm{\mathcal{K}}}
\newcommand{\rsf}{\bm{\mathsf{R}}}

\newtheorem{thm}{Theorem}[section]
\newtheorem{lem}{Lemma}[section]
\newtheorem{assumption}{Assumption}[section]
\newtheorem{cor}{Corollary}[section]
\newtheorem{prop}{Proposition}[section]
\theoremstyle{remark}
 
\newcommand{\E}{\mathbb{E}}
\newcommand{\Cov}{\mathrm{Cov}}
\newcommand{\R}{\mathbb{R}}

\newcommand{\wt}{\widetilde}
\newcommand{\wh}{\widehat}
\newcommand{\wb}{\widebar}
\newcommand{\tr}{\mbox{Tr}}
\newcommand{\diff}[1]{\mathsf{d}#1}

\newcommand{\op}{\mathrm{op}}

\newcommand{\upca}{\bm U^{\mathrm{pc}}}
\newcommand{\vpca}{\bm V^{\mathrm{pc}}}
\newcommand{\dpca}{\bm D^{\mathrm{pc}}}

\newcommand{\zlpca}{\bm Z^{\mathrm{pc,\ell}}}
\newcommand{\zrpca}{\bm Z^{\mathrm{pc,r}}}
\newcommand{\hV}{\wh{\bm V}}
\newcommand{\hU}{\wh{\bm U}}
\newcommand{\hG}{\wh{\bm G}}
\newcommand{\hF}{\wh{\bm F}}
\newcommand{\bU}{\wb{\bm U}}
\newcommand{\tU}{\wt{\bm U}}

\newcommand{\bdU}{{\bm U}}
\newcommand{\bdV}{{\bm V}}
\newcommand{\bdX}{\bm X}
\newcommand{\bX}{\wb{\bm X}}
\newcommand{\tX}{\wt{\bm X}}
\newcommand{\bdJ}{{\bm J}}
\newcommand{\hM}{\wh{\bm M}}

\newcommand{\hS}{\wh{\bm \Sigma}}

\newcommand{\supp}{\mathrm{supp}}

\allowdisplaybreaks

\makeatletter
\def\widebreve{\mathpalette\wide@breve}
\def\wide@breve#1#2{\sbox\z@{$#1#2$}%
     \mathop{\vbox{\m@th\ialign{##\crcr
\kern0.08em\brevefill#1{0.8\wd\z@}\crcr\noalign{\nointerlineskip}%
                    $\hss#1#2\hss$\crcr}}}\limits}
\def\brevefill#1#2{$\m@th\sbox\tw@{$#1($}%
  \hss\resizebox{#2}{\wd\tw@}{\rotatebox[origin=c]{90}{\upshape(}}\hss$}
\makeatletter

\newcommand{\rsc}{\mathrm{rsc}}

\usepackage{subcaption}
\usepackage[
            CJKbookmarks=true,
            bookmarksnumbered=true,
            bookmarksopen=true, 
            colorlinks=true,
            citecolor=blue,
            linkcolor=blue,
            anchorcolor=red,
            urlcolor=purple
            ]{hyperref}

\renewcommand{\arraystretch}{2}
\DeclareMathOperator*{\argmin}{arg\,min}
\DeclareMathOperator*{\argmax}{arg\,max}

\usepackage{graphicx} 
\usepackage{color, float}

\usepackage{booktabs, multirow} 
\usepackage{array} 
\usepackage{verbatim} 
\usepackage[round]{natbib}
\makeatletter
\newcommand{\numcite}[1]{%
  [\begingroup
   \renewcommand{\NAT@sep}{,}%
   \citenum{#1}%
   \endgroup]%
}
\makeatother

\title{DAIF: A Data-Driven Intermediate Fusion Framework for Multimodal Supervised Learning via Approximate Message Passing}
\author{
Sagnik Nandy\thanks{Ohio State University. Email: \texttt{nandy.15@osu.edu}.}
\and
Samriddha Lahiry\thanks{National University of Singapore. Email: \texttt{slahiry@nus.edu.sg}.}
\and
Pragya Sur\thanks{Harvard University. Email: \texttt{pragya@fas.harvard.edu}.}
\and
Subhabrata Sen\thanks{Harvard University. Email: \texttt{subhabratasen@fas.harvard.edu}.}
}

\begin{document}
	\maketitle
	
\begin{abstract}
Multimodal supervised learning leverages information across diverse feature types to improve predictive performance. A central challenge lies in determining the granularity of fusion across modalities: over-integration may amplify noise while under-integration fails to exploit cross-modal dependence. Existing approaches rely on pre-specified fusion architectures, from early to late fusion, that may not adapt to the underlying dependence structure among modalities. We propose DAIF, a data-adaptive fusion framework that learns fusion structure directly from data. We operate under a multimodal factor model where a prior on  latent factors determines the cross-modal dependence. Our method clusters modalities based on estimated intermodal dependence, then performs clusterwise empirical Bayes estimation of the priors. The estimated priors are used to construct denoisers within an approximate message passing framework, yielding informative low-dimensional features that borrow strength across modalities while preserving modality-specific signals. The resulting embeddings are used for downstream prediction. We evaluate our framework through extensive simulations and demonstrate its utility on two real datasets where we predict the expression level of a marker protein in a trimodal TEA-seq data \citep{Swanson2021} and patient survival in the multimodal TCGA-BRCA data \citep{Goldman2020}. Our method outperforms existing approaches in both problems, demonstrating its versatility across diverse tasks.
\end{abstract}

\section{Introduction}
Modern data acquisition technologies allow users to collect multiple complementary pieces of information about subjects simultaneously. These yield rich multimodal datasets. For instance, biomedical studies increasingly collect genomic, proteomic, metabolomic, lipidomic, and other omics measurements alongside clinical or phenotypic outcomes
\citep{gentles2015integrating,zhao2019learning,dombowsky2025bayesian,sc-mult-deep}.
The availability of such data has the potential to substantially improve prediction capabilities in supervised learning by enabling principled integration of complementary information across modalities, facilitating highly accurate prediction in settings where features from individual component modalities are noisy and high-dimensional. 

In this paper, we study supervised multimodal learning based on observations
$\{(y_i,\bm{x}_{i1},\ldots,\bm{x}_{im})\}_{i=1}^n$, where $y_i$ is the response for subject $i$ and $\bm{x}_{ih}\in\mathbb{R}^{p_h}$ is its feature vector from modality $h$. Our goal is to learn a predictor of the response of a new subject from its multimodal features. We allow the prediction task to involve continuous, binary, or categorical responses, as well as survival outcomes subject to censoring, and allow the modalities to differ substantially in dimensionality and signal strength.

In many application domains, it is empirically observed that intermodal dependence and the association between the response and the observed features across modalities are mediated by a lower-dimensional collection of latent factors \citep{anceschi2024bayesian}. The supervised learning workflow in these settings typically involves constructing low-dimensional embeddings from the observed features by information integration across modalities and followed by training a supervised predictor using the engineered features.
The information fusion can enhance signal strength in the learned representations and improve the recovery of latent structure by borrowing strength across complementary data sources. However, a fundamental and often overlooked question in this context is: \emph{which modalities should be fused?} Insufficient integration may fail to exploit the full potential of multimodal data, resulting in weak or uninformative representations. In contrast, excessive integration may force together unrelated or weakly related modalities; 
in this case, one neglects the clustering among the modes, and is thus forced to estimate the complete signal (the shared signal and the modality specific signals) from the available data. Both these extremes are undesirable and one stands to gain significantly in prediction accuracy by suitably integrating information across strongly linked modalities.

In the literature \citep{Boulahia2021,ding2022cooperative}, this problem is commonly addressed through three broad architectural paradigms for feature engineering: (I) {\bf Late Fusion:} Each modality is modeled independently to produce modality-specific predictions, which are then combined at the decision level to obtain the final prediction \citep{snoek2005early}, (II) {\bf Intermediate Fusion:} Information from multiple modalities is partially fused during feature learning, typically by sharing representations or interactions among selected modalities, while others are processed independently \citep{Boulahia2021}, and (III) {\bf Early Fusion:} All modalities are jointly embedded into a single unified latent representation, which is subsequently used for downstream prediction tasks \citep{wu2018multimodal,srivastava2014multimodal}.
These three paradigms of information integration have been successfully deployed across a wide range of biomedical applications \citep{GUARRASI2025105509,Cahan2023} as well as in computer vision \citep{Boulahia2021,hu2024intermediatefusionvitenables}.

A significant limitation of existing data-fusion pipelines is that the choice of fusion architecture is largely \emph{ad hoc}, 
rather than inferred from the data and is fixed \emph{a priori}, often leading to undesirable performance in datasets where the choice of the fusion granularity depends on the composition of the data under study. For example, in Section~\ref{sec:comp_bio_illustrate}, we consider predicting the expression of a T-cell differentiation marker in a trimodal single cell dataset \citep{Swanson2021}. We observe that depending on the cell type under consideration, either early or late fusion has better predictive performance (cf. Figure~\ref{fig:teaseq-rmse}). Therefore, one can easily imagine that in a generic supervised learning problem involving such datasets, the most efficient granularity of fusion will be determined by the relative abundance of different cell-types in the sample. In this case, pre-fixing the fusion architecture might lead to a loss of efficiency and a more desirable framework would determine the appropriate granularity of fusion in a fully data-driven manner. In this paper, we propose \fancyname{}, a data adaptive framework for learning the appropriate fusion granularity for data integration that preserves modality-specific structure while enhancing shared signals. 

\begin{figure}[t]
    \centering
    \includegraphics[width=\textwidth]{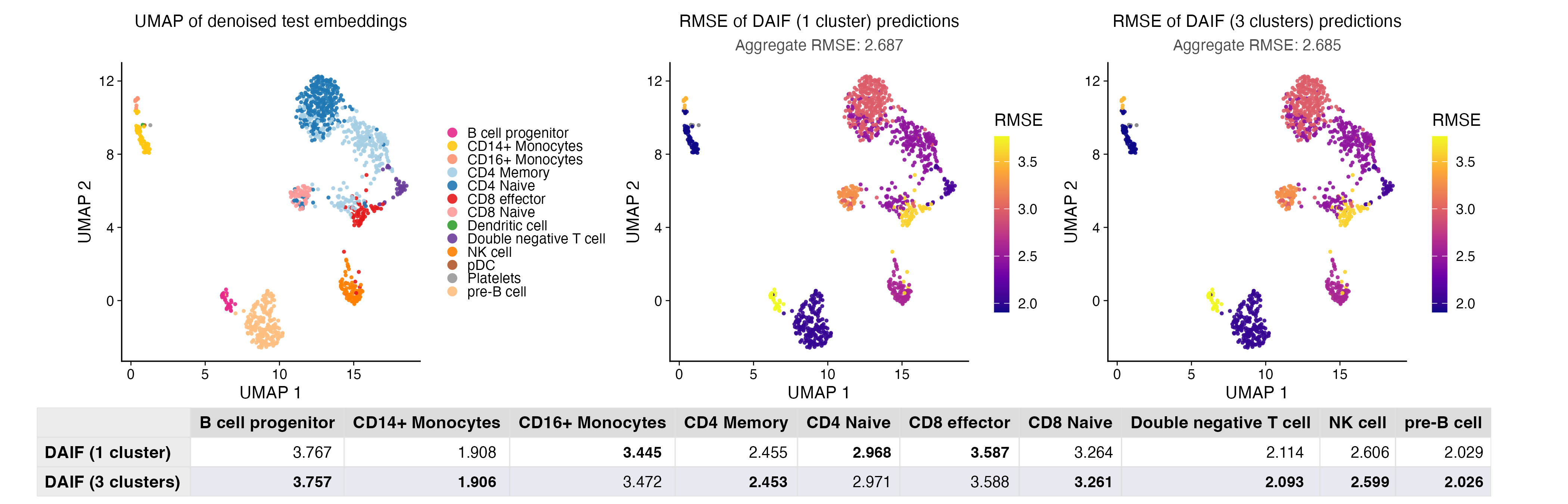}
    \caption{Prediction of CD45RA protein expression from RNA (2{,}000 HVGs), normalized ATAC
    counts (5{,}000 HV ATAC reads), and normalized expression of 39 HV cell-surface protein
    markers on TEA-seq data \citep{Swanson2021}. The \fancyname{} predictor is trained on
    $4{,}560$ randomly selected cells and evaluated on $1{,}141$ held-out test cells. Each panel
    shows per-celltype RMSE overlaid on the UMAP of denoised test embeddings; the bottom table reports
    per-cell-type RMSE, with bold entries marking the best method. Early fusion ($K=1$,
    OrchAMP) and late fusion ($K=3$, EB-PCA) achieve comparable aggregate RMSE ($2.687$ and
    $2.685$, respectively) yet differ substantially at the cell-type level, motivating
    data adaptive fusion granularity selection.}
    \label{fig:teaseq-rmse}
\end{figure}

\subsection{An illustration from computational biology}
\label{sec:comp_bio_illustrate}

To illustrate the challenge of selecting an appropriate fusion granularity, we consider a supervised
prediction task in the trimodal TEA-seq data \citep{Swanson2021}, which jointly profiles
nuclear mRNA expression, chromatin accessibility (ATAC), and antibody-derived tags (ADTs)
measuring cell-surface protein abundance with each modality exhibiting sharply different statistical properties. Proteomic measurements are typically
low-dimensional with high signal-to-noise ratios, whereas RNA and ATAC data are extremely
high-dimensional and sparse. RNA and ATAC are strongly coupled biologically, since chromatin
accessibility regulates gene expression, while ADTs are only weakly linked to ATAC due to
post-transcriptional regulation. This heterogeneity makes fusion granularity both practically
important and statistically nontrivial, and it is precisely what \fancyname{} is designed to
learn from data.

Figure~\ref{fig:teaseq-rmse} illustrates a prediction task where the goal is to predict CD45RA
expression which is a key marker of T cell differentiation---from RNA, ATAC, and the remaining
cell-surface protein measurements. We apply \fancyname{}, which clusters modalities using an
empirical dependence measure Centered Kernel Alignment (CKA; \cite{kornblith2019similarity}) and
integrates modalities within the same cluster using Approximate Message Passing (AMP). The number of clusters, selected by the gap statistic \citep{Tibshirani_2001}, determines the fusion granularity: $K=1$ recovers early fusion (OrchAMP; \cite{nandy2024multimodal}), $K=3$ recovers late fusion (EB-PCA; \cite{zhong2022empirical}), and intermediate values yield partial fusion.

Strikingly, no single fusion strategy uniformly dominates across cell types. Late fusion performs better in B cell progenitors
and CD14$^+$ Monocytes, while early fusion is better for CD16$^+$ Monocytes and CD4
Naive cells. A fixed fusion architecture cannot adapt to the relative abundance of the cell types in the observed sample to choose the most efficient fusion strategy but \fancyname{} addresses this by learning the appropriate fusion granularity from data and automatically takes into account the relative abundance of different cell-types in the observed sample. We further explore this dataset in Section~\ref{sec:tea_seq}, where we show that \fancyname{} outperforms state-of-the-art procedures for multimodal prediction in terms of test RMSE (cf.\ Table~\ref{tab:tea_regression_split3}).

\subsection{Major Contributions}
Our major contributions in this paper are as follows.

\noindent\textbf{Bayesian modelling of varying granularity of intermodal dependence.} In \autoref{sec:stat_model}, we propose a Bayesian multimodal factor model for the multimodal features that employs a structured prior, decomposing the features into components corresponding to strongly dependent modalities. This formulation flexibly interpolates between complete dependence and independence across modalities. It further allows intermodal dependence to be mediated by latent features of differing dimensions, accommodating heterogeneity across data sources. This contrasts with common integration methods such as multiset canonical correlation analysis \citep{kettenring1971canonical} and JIVE \citep{lock2013jive}, which assume dependence through shared latent factors, an assumption that can be restrictive and less interpretable in heterogeneous settings. \emph{Cooperative component analysis} (CoCA) \citep{ding2024coca} provides an intermediate-fusion based dimension reduction technique that interpolates between PCA and CCA. Unlike \fancyname{}, however, CoCA focuses on recovering a shared linear component across views, whereas \fancyname{} preserves modality-specific latent structure while enabling nonlinear information integration through data adaptive clustering and cluster-wise empirical Bayes denoising. Our framework is closely related to \citet{anceschi2024bayesian}, but extends it by allowing random loading matrices and latent factors of varying dimensions across modalities which is important in the context of multimodal analyses in single cell genomics, where biological signals in the modalities are often captured by varying numbers of latent factors \citep{HAO20213573}.

\noindent\textbf{A data adaptive intermediate fusion pipeline.} In \autoref{sec:methodology}, we propose \fancyname{}, an empirical Bayes pipeline for integrating and engineering high-dimensional multimodal features. The pipeline first learns the clustering structure in the prior via hierarchical clustering, using similarity measures based on non-parametric notions of dependence, such as Centered Kernel Alignment \citep{kornblith2019similarity}. This yields a fully adaptive determination of fusion granularity and distinguishes our approach by learning the fusion architecture directly from the data. We then estimate cluster-specific priors using empirical Bayes and leverage them to integrate information across modalities. 

Our integration pipeline builds upon a variant of the Orchestrated Approximate Message Passing framework (OrchAMP) (cf. Algorithm~\ref{alg:daif_amp_feature_extract} in the appendix) for multimodal unsupervised learning proposed in \cite{nandy2024multimodal} which in turn is connected to the rich literature of Approximate Message Passing (AMP) in modern high-dimensional statistics across supervised and unsupervised learning \citep{donoho2009message,BayatiMontanari2011AMP, sur_candes,rush_venkatramanan,montanari2021estimation,feng2022unifying,li_wei,mondelli2021approximate}. Typically, an AMP algorithm alternates between power iterations that extract low-rank signal from the feature matrices, Bayesian denoising of the resulting iterates, and a debiasing step that removes dependence on past iterates.\footnote{In the AMP literature \citep{BayatiMontanari2011AMP}, this debiasing step is termed the \emph{Onsager correction}.} Specialized to multimodal data, the Bayesian denoiser computes posterior means under a joint prior on the modality-specific latent factors, thereby integrating information across modalities. The central challenge is that the appropriate joint prior, and hence which modalities should be denoised together for downstream prediction, is unknown. DAIF addresses this challenge through our cluster based prior estimation framework.

\noindent\textbf{Test-time prediction.}
The estimated loading matrices obtained from \fancyname{} can be decomposed into a weighted sum of the true loading matrices in the data generating process and independent Gaussian noise. If the test observations share the same loading matrices across modalities (a reasonable assumption in our setting), we can use this decomposition to project the test observations onto the same latent space as the training data, which enables the learned predictor to predict the response on test data. Further details are provided in Section~\ref{sec:ols_pretraining}.  

\noindent\textbf{Numerical experiments and applications to biological datasets.}
We benchmark \fancyname{} against popular data integration pipelines on two fronts: recovery of
latent features and prediction of held-out responses. We present our numerical experiments in 
\autoref{sec:numerical_exp} which demonstrate that \fancyname{} is about 10 times more accurate than the non-AMP based integration benchmarks both in terms of recovering the latent embeddings and predicting the response.

We further validate \fancyname{} on two real datasets: a trimodal single-cell dataset from
\cite{Swanson2021}, where we predict the expression level of a marker protein from genomic,
epigenomic, and proteomic features (cf. Section~\ref{sec:tea_seq}); and a trimodal TCGA-BRCA dataset from \cite{Goldman2020},
where we predict survival risk (cf. Section~\ref{sec:tcga_brca}). In both settings, \fancyname{} outperforms competing
integration-based prediction pipelines in computational biology.

\noindent\textbf{Theoretical justification of the \fancyname{} framework.}
On the theoretical front, we establish consistency of the cluster-based prior estimation procedure in Theorem~\ref{thm:consistency_prior_main} provided the true prior indeed exhibits a clustered structure and the components within each cluster are strongly dependent. Furthermore, in the appendix, we establish asymptotic properties of the AMP iterates through state evolution recursions (cf. Theorem~\ref{thm:w_2_amp}). Therein, we also show that these iterates can be employed to construct Bayes optimal estimators for the true signals (cf. Theorem~\ref{thm:bayes_optimality}). Finally, the asymptotic properties of the projected training and test features used for prediction are established in Theorem~\ref{thm:ols_estimator}. These asymptotic properties motivate the construction of the loss function defined in \eqref{eq:def_frak_l} which is used to train the predictor. All the proofs are deferred to the appendix.

\subsection{Related Work}
\label{sec:related_work}
The idea of multimodal integration traces back to Canonical Correlation Analysis
\citep{hotelling1933analysis} and its multi-set generalizations \citep{kettenring1971canonical}. This classical methods assumes the same modality dimensions which is restrictive in many real multimodal integration problems. For example, in \cite{HAO20213573,gayoso2021joint}, the authors use different latent dimensions for modeling the transcriptomic and proteomic modalities.

Another approach of multimodal integration is through multimodal factor models \citep{lock2013jive, feng2018angle, divas, gaynanova2019structural, sergazinov2024spectral} which posits a shared low-dimensional latent structure across views. In \citet{yang2025estimating}, the authors recently established
rigorous performance guarantees for such methods but all the related procedures assume the same latent dimension for the shared signal component and sometimes assume that the correlated signal complement is identical across views. this assumption is also restrictive. A complementary approach is through multimodal spectral analysis \citep{ma2026optimal,baharav2025stacked}, however the procedures either require same dimensions of the features across modalities or are provably sub-optimal.

In computational biology, dominant approaches involve complete data integration, including weighted nearest-neighbor graph-based Seurat v4 \citep{HAO20213573}, MOFA+ \citep{Argelaguet2020}, and variational autoencoder based totalVI \citep{gayoso2021joint} and Multigrate \citep{LitinetskayaPoEVAE2022} and is appropriate only when the modalities are strongly dependent. Furthermore, these methods do not provide rigorous statistical gurantees. \fancyname{} allows data-based determination of the right fusion granularity and has rigororus statistical guarantees.

In supervised multimodal learning, most existing methods adopt fixed fusion architecture \citep{yuan2014assessing,gentles2015integrating}. Cooperative learning framework of \citep{ding2022cooperative} provides an option to interpolate between early and late fusion but does not theoretically characterize the prediction error, while the Bayesian procedure JAFAR \citep{anceschi2024bayesian} is restricted to early fusion and relies on response annotations for dimension reduction, making it sensitive to unreliable annotations. Other related theoretical approaches \citep{NEURIPS2025_b0e7cfb9,yang2025estimating,tabanelli2025computational,rosetti_reeves,gerbelot2023graph,ma_nandy,KeupZdeborova2025OptimalThresholds} do not address data-adaptive selective fusion based on intermodal dependence. Although \citet{liang2025meta} study data-driven selection of fusion granularity through their \emph{Meta Fusion} framework; their guarantees require Gaussian latent factors; but, \fancyname{} estimates the fusion structure from data without such restriction.

The two direct AMP-based predecessors of \fancyname{} are EB-PCA \citep{zhong2022empirical} and OrchAMP \citep{nandy2024multimodal}. EB-PCA applies AMP
with a non-parametric empirical Bayes prior to a single data matrix and OrchAMP extends this to multiple modalities under a dependent
multifactor model, fusing all views simultaneously for optimal joint signal recovery. In our
framework, applying EB-PCA independently to each modality is late fusion, while OrchAMP-style joint denoising is early fusion. \fancyname{} interpolates the gap between these extremes by clustering modalities according to an empirical measure of pairwise dependence while applying AMP-based denoising within each cluster effectively achieving an intermediate fusion granularity that adapts to the data and subsumes both predecessors as special cases. 

\subsection{Notations}
For any natural number $n \in \mathbb N$, we shall denote the set $\{1,\ldots,n\}$ by $[n]$. All the matrices will be denoted by upper-case bold letters throughout the text. For any matrix $\bm A$, $\bm A_{i*}$ denotes its $i$-th row and $\bm A_{*i}$ its $i$-th column.  
The Frobenius norm of a matrix $\bm A$ is denoted by $\|\bm A\|_F$.
Next, suppose $\theta \sim \pi$ and $\mathsf{X}:=f(\theta, \mathsf Z)$, where $\mathsf Z$ is independent of $\theta$. Then $\mathbb E_\pi[g(\theta) \mid \mathsf{X}]$ shall denote the expectation with respect to the posterior distribution $\pi(\theta \mid \mathsf{X})$.

\section{Statistical model behind our approach}
\label{sec:stat_model}
\subsection{A multimodal factor model}
In this paper, we consider the following factor model for the high-dimensional feature matrices $\{\bm X_h \in \mathbb{R}^{n \times p_h}: h \in [m]\}$:
\begin{align}
\label{eq:multimodal_factor_model}
\bm X_h=\frac{1}{\sqrt{n}}\,\bm U_h \bm D_h \bm V_h^\top+\bm Z_h,
\end{align}
where $\bm U_h \in \mathbb{R}^{n \times r_h}$, $\bm V_h \in \mathbb{R}^{p_h \times r_h}$, and $\bm D_h \in \mathbb{R}^{r_h \times r_h}$ is diagonal with strictly decreasing positive entries, so that $(\bm D_h)_{11}>\cdots>(\bm D_h)_{r_h r_h}>0$. The entries of the noise matrices $\{\bm Z_h \in \mathbb{R}^{n \times p_h}:h\in[m]\}$ are independently and identically distributed according to the standard Gaussian distribution. Furthermore, we operate in the proportional asymptotic regime where $p_h/n \rightarrow \gamma_h \in (0,\infty)$ for all $h \in [m]$, as $n \rightarrow \infty$. The intrinsic dimensions of the features are assumed to satisfy $r_1,\ldots,r_m=O(1)$.
In addition to the high-dimensional modalities, we also consider $\wt m$ low-dimensional feature matrices $\{\wt{\bm X}_\ell:\ell=1,\ldots,\wt m\}$ modeled as
\begin{align}
\label{eq:multimodal_low_dim_model}
\wt{\bm X}_\ell = \wt{\bm U}_\ell \bm L_\ell^\top + \wt{\bm Z}_\ell,
\end{align}
where $\wt{\bm U}_\ell \in \mathbb{R}^{n \times \wt r_\ell}$, $\bm L_\ell \in \mathbb{R}^{\wt r_\ell \times \wt r_\ell}$, and $\wt{\bm Z}_\ell \in \mathbb{R}^{n \times \wt r_\ell}$, with $\wt r_\ell=O(1)$ for all $\ell \in [\wt m]$. The matrices $\{\bm L_\ell: \ell \in [\wt m]\}$ are assumed to be symmetric and positive definite. We assume that the entries of $\{\wt{\bm Z}_\ell:\ell \in [\wt m]\}$ are sampled i.i.d from the standard Gaussian distribution and are independent of the collection $\{\bm X_h:h \in [m]\}$. In the foregoing formulation, the signals are captured through $\bm U_h \bm D_h \bm V_h^\top$ for $h \in [m]$ and $\wt{\bm U}_\ell \bm L_\ell^\top$ for $\ell \in [\wt m]$. To ensure that the signal component separates out from the noise component, we assume that $(\bm D_h)_{jj} >\gamma^{-1/4}_h$ for all $j \in [r_h]$ and $h \in [m]$. In particular, the collection of subject-specific latent feature vectors
\(
\left\{\bigl((\bm U_1)_{i*},\ldots,(\bm U_m)_{i*},(\wt{\bm U}_1)_{i*},\ldots,(\wt{\bm U}_{\wt m})_{i*}\bigr) : i = 1,\ldots,n \right\}
\)
constitutes the latent representation associated with the $n$ observational units and are therefore the target of the feature engineering step. To enforce identifiability, we assume $(\bm U_h)^\top \bm U_h \approx n \bm I_{r_h}$ and $(\bm V_h)^\top \bm V_h \approx p_h \bm I_{r_h}$ for all $h \in [m]$ \footnote{See Assumption~\ref{assu:init} for a formal statement.}. Therefore, the columns of $\bm U_h$ and $\bm V_h$ are approximately the left and right singular vectors of the signal component in $\bm X_h$, endowing them with natural interpretability. 

Next, we model the response vector $y=(y_1,\ldots,y_n)^\top$ as a function of the subject-level signatures. Let $g(\cdot\mid \theta)$ denote a parametric family of distributions, and suppose that
\begin{align}
\label{eq:response_vector}
y_i \overset{\mathrm{ind}}{\sim} g(\cdot \mid \theta_i),\qquad
\theta_i =f\left((\bm U_1)_{i*},\ldots,(\bm U_m)_{i*},
(\wt{\bm U}_1)_{i*},\ldots,(\wt{\bm U}_{\wt m})_{i*}\right),
\qquad i=1,\ldots,n.
\end{align}
Here, the response distribution $g(\cdot\mid \theta)$ is assumed to be known. The regression function $f(\cdot)$ is assumed to be known up to a finite-dimensional set of unknown parameters, which are estimated from the data. For instance, taking
$g(x\mid \theta)=\phi(x-\theta)$ for $x\in\R$, where
$\phi(x)=(2\pi)^{-1/2}\exp(-x^2/2)$, and taking $f(u)=\beta^\top u$, yields the usual linear regression model with standard normal errors. We provide additional examples in Section~\ref{sec:numerical_exp}.

For notational clarity, define the disjoint modality index sets $\mathcal E_{\mathrm H}:=[m]$, $\mathcal E_{\mathrm L}:=\{m+1,\ldots,m+\widetilde m\}$, and $\mathcal E:=\mathcal E_{\mathrm H}\cup\mathcal E_{\mathrm L}
=[m+\widetilde m].$ For each $e\in\mathcal E$, the corresponding latent vector $\mathcal U_e$ is defined to be $U_e$ if $e\in\mathcal E_{\mathrm H}$ and $\widetilde U_{e-m}$ is $e\in\mathcal E_{\mathrm L}$. Similarly, the latent dimension $\rho_e$ is defined as $r_e \mathbbm 1\{e\in\mathcal E_{\mathrm H}\}+\widetilde r_{e-m}\mathbbm 1\{e\in\mathcal E_{\mathrm L}\}$.
For any $\mathcal C\subseteq\mathcal E$, we further define $\mathcal C_{\mathrm H}:=\mathcal C\cap\mathcal E_{\mathrm H},$ and $\mathcal C_{\mathrm L}:=
\mathcal C\cap\mathcal E_{\mathrm L}.$
Finally, let $r_{\mathrm H}:=\sum_{h=1}^m r_h$, $r_{\mathrm L}:=\sum_{\ell=1}^{\widetilde m}\widetilde r_\ell$, and $r:=r_{\mathrm H}+r_{\mathrm L}.$

\subsection{Cluster based prior on subject signatures}
We postulate that the multimodal dependence is captured through a latent prior $\mu$ on the subject signatures $\bm U_1,\ldots,\bm U_m$ and $\wt{\bm U}_1,\ldots,\wt{\bm U}_{\wt m}$. Formally, 
\begin{align}
    \label{eq:prior_specification}
    \left((\bm U_{1})_{i*},\ldots,(\bm U_{m})_{i*},(\wt{\bm U}_{1})_{i*},\ldots,(\wt{\bm U}_{\wt m})_{i*}\right) \overset{\text{i.i.d}}{\sim} \mu.
\end{align}
The prior $\mu$ is supported on $\mathbb R^{r}$. Furthermore, we assume that the rows of $\bm V_1,\ldots,\bm V_m$ are mutually independent with 
\(
(\bm V_h)_{i*} \overset{\text{i.i.d}}{\sim} \nu_h,
\)
for $i \in [p_h]$ and $h \in [m]$.

We assume that $\mathcal E$ partitions into $K$ disjoint subsets $\mathcal C_1,\ldots,\mathcal C_K$ and the prior $\mu :=\bigotimes_{k=1}^K
\mu_k\left(\{\mathcal U_e:e\in\mathcal C_k\}\right),$ a  natural prior structure from the viewpoint of graphical models \citep{jordan1998learning,lauritzen1996graphical}, where conditional independence assumptions are commonly used to encode sparse dependence among random variables. Our assumption states that modalities within the same cluster may exhibit arbitrary dependence, whereas modalities belonging to different clusters are independent under the prior, effectively implying that the dependence graph is a collection of disjoint cliques corresponding to the clusters. This structure naturally induces three distinct granularities of intermodal dependence: (1) \emph{Complete dependence} ($K=1$), (2) \emph{Intermediate dependence} ($1 < K< m+\wt m$), and (3) \emph{Independence} ($K=m+\wt m$). Importantly, the partition $\mathcal C_1,\ldots,\mathcal C_K$ is not assumed to be known a priori and is learned from the data, effectively providing a mechanism for adaptive information sharing avoiding over integration. 

\section{Methodology}
\label{sec:methodology}
In this section, we formally introduce \fancyname{}. 
The major steps in the algorithm are as follows:

{\bf Clustered Empirical Bayes.} We begin by computing the top $r_h$ singular vectors of each $\wb{\bm X}_h=\bm X_h/\sqrt{n}$, given by $(\upca_h,\vpca_h)$, rescaled so that $(\upca_h)^\top \upca_h = n \bm I_{r_h}$ and $(\vpca_h)^\top \vpca_h = p_h \bm I_{r_h}$ for $h \in [m]$. We assume $r_h$ are known; in practice, they are selected via scree plots by identifying an elbow in the singular value distribution. We construct a similarity measure by computing pairwise dependence between all modality-specific representations, including $\{\upca_h,\upca_{h'}\}$, $\{\upca_h,\wt{\bm X}_\ell\}$, and $\{\wt{\bm X}_\ell,\wt{\bm X}_{\ell'}\}$. For each modality pair, we compute row-wise similarities between the corresponding matrices and average them to form a similarity matrix, which is used for hierarchical clustering. To quantify similarity between the modalities, we can employ non-parametric measures of dependence like Centered Kernel Alignment (CKA) \citep{kornblith2019similarity}, distance correlation \citep{szekely2007measuring}, nearest neighbor-based measures \citep{Deb2020Measuring}, and mutual information-based methods such as InfoNCE \citep{oord2018representation} which are close to zero under independence and approach one under strong dependence, while capturing nonlinear relationships. However, in this paper, we consistently use CKA with Gaussian kernel as the measure of dependence. A detailed discussion on the relative strengths and weaknesses of the above-mentioned measures of independence is provided in Section~\ref{sec:measures_of_dependence} of the appendix.
We convert similarities to dissimilarities by subtracting from one and apply hierarchical clustering with UPGMA-based average linkage \citep{upgma} to obtain clusters $\{\wh{\mathcal C}_k:k \in [\wh K]\}$. 
Within each estimated cluster $\widehat{\mathcal C}_k$, we estimate
the prior $\mu_k$ using $\{\bm U_h^{\mathrm{pc}}:h\in\widehat{\mathcal C}_{k,\mathrm H}\}$ and $\{\widetilde{\bm X}_\ell:\ell\in\widehat{\mathcal C}_{k,\mathrm L}\}$, where $\widehat{\mathcal C}_{k,\mathrm H}:=\widehat{\mathcal C}_k\cap\mathcal E_{\mathrm H}$, and $\widehat{\mathcal C}_{k,\mathrm L}:=\{\ell\in[\widetilde m]:m+\ell\in\widehat{\mathcal C}_k\}.$ This step constitutes the empirical Bayes component of our pipeline. The estimation of the prior under the product structure imposed by the learned clusters mitigates the curse of dimensionality associated with high-dimensional prior estimation in finite samples. As the clusters are inferred from the data, the pipeline provides fully data adaptive intermediate fusion.

{\bf Feature extraction using Orchestrated Approximate Message Passing.} Given the estimated prior, we recover the low-dimensional latent factors from $\{\bm X_h:h \in [m]\}$ using \emph{Orchestrated Approximate Message Passing} \citep{nandy2024multimodal}. The algorithm alternates between an Onsager-corrected power iteration \citep{zhong2022empirical} and a Bayesian denoising step using the estimated clustered prior from the previous step, which shrinks the iterates toward the prior mean. This shrinkage stabilizes the iterates via noise correction and facilitates information sharing across modalities within each cluster, effectively borrowing strength from strongly dependent modalities encoded by the estimated prior. The use of a clustered prior naturally induces an intermediate fusion scheme, enabling feature integration at the level of modality groups rather than enforcing full or no fusion. For low-dimensional modalities, power iteration is unnecessary; instead, we directly apply Bayesian denoising to shrink the raw features toward the estimated prior using information from related modalities from the same cluster. This yields a sequence of estimates $\{\wh{\bm F}_{t,h},\wt{\bm U}_{t,\ell}:h \in [m],\ell \in [\wt m]\}$ for $\{\bm U_h,\wt{\bm U}_\ell:h \in [m],\ell \in [\wt m]\}$, along with a sequence of estimates $\{\wh{\bm V}_{t,h}:h \in [m]\}$ of $\{\bm V_h:h \in [m]\}$.

{\bf Predictor training and test prediction.} At the test stage, directly estimating the low-rank embeddings from the high-dimensional features via singular value decomposition is inappropriate since typically $n_{\test} \ll p_h$ for the high-dimensional modalities. Instead, we use the estimated loading matrices $\wh{\bm V}_{T,h}$ (obtained after $T$ AMP iterations) and project the test features $\bm X^{\test}_h$ corresponding to the high-dimensional modalities onto the latent space via ordinary least squares to obtain $\wh{U}^{\test,\ols}_h$. Using properties of the AMP iterates (cf. Theorem~\ref{thm:ols_estimator}), these projected features can be represented as Gaussian perturbation of the true test latent factors $U^{\test}_h$ which are predictive for the response $y$.
To align the asymptotic distribution of the training and test-time representations, we rescale the Onsager-corrected training embeddings $\wh{\bm F}_{T,h}$ by right-multiplying them with $(\wh{\bm \Sigma}^{L}_{T,h})^{-1}\wh{\bm D}_h^{-1}$. The resulting embeddings have the same asymptotic distribution as the test-time OLS embeddings $\wh{U}^{\test,\ols}_h$. 
We use the rescaled embeddings to train a predictor $\mathfrak H(\,\cdot\,,\theta)$ by minimizing, with respect to $\theta$, the loss function $\mathfrak L_\theta(\cdot)$ defined in \eqref{eq:def_frak_l}, which explicitly accounts for the Gaussian perturbations present in the training embeddings.
The predictor architecture is kept flexible depending on the application. The trained predictor is finally applied to $\{\wh{ U}^{\test,\ols}_h:h \in [m]\}$ to obtain the test predictions. See Figure~\ref{fig:thematic} for a thematic diagram and Algorithm~\ref{alg:daif_compact} for a pseudo-code style summary. 

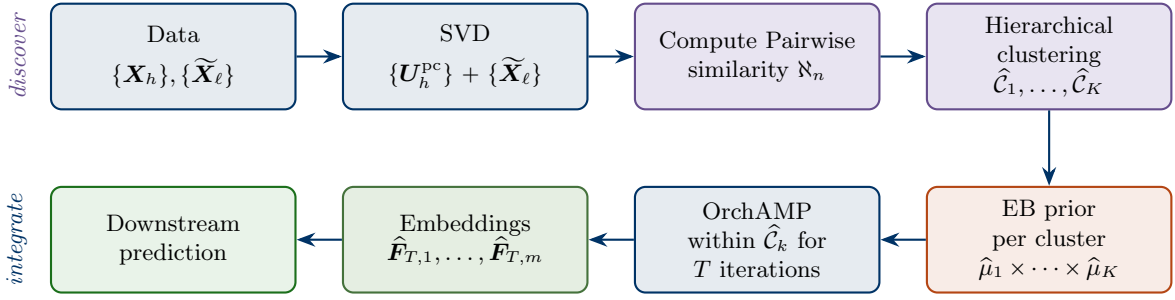
\begin{figure}
    \centering
    \begin{tikzpicture}[
  node distance = 0.7cm and 0.6cm,
  mybox/.style   = {rectangle, rounded corners=4pt, draw=MyBlue, thick,
                    fill=MyBlue!10, align=center, minimum height=1.4cm,
                    text width=3cm, font=\small},
  depbox/.style  = {mybox, draw=RoyalPurple!80!black, fill=RoyalPurple!10},
  priorbox/.style= {mybox, draw=FocusColor, fill=FocusColor!10},
  outbox/.style  = {mybox, draw=OliveGreen!80!black, fill=OliveGreen!10},
  taskbox/.style = {mybox, draw=ForestGreen!80!black, fill=ForestGreen!10},
  arr/.style     = {-{Stealth[length=2.5mm]}, MyBlue, thick}
]

\node[mybox]  (data)  {Data\\[4pt]$\{\bm X_h\},\{\wt{\bm X}_\ell\}$};
\node[mybox,  right=of data]  (svd)
    {SVD\\[4pt]$\{\bm U^{\mathrm{pc}}_h\}+\{\wt{\bm X}_\ell\}$};
\node[depbox, right=of svd]   (cka)
    {Compute Pairwise\\similarity $\aleph_n$};
\node[depbox, right=of cka]   (clust)
    {Hierarchical\\clustering\\$\wh{\mathcal C}_1,\ldots,\wh{\mathcal C}_K$};

\node[priorbox, below=1.0cm of clust] (prior)
    {EB prior\\per cluster\\$\wh{\mu}_1\!\times\!\cdots\!\times\!\wh{\mu}_K$};
\node[mybox,   left=of prior]  (amp)   {OrchAMP\\within $\wh{\mathcal{C}}_k$ for $T$ iterations};
\node[outbox,  left=of amp]    (out)   {Embeddings\\$\wh{\bm F}_{T,1},\ldots,\wh{\bm F}_{T,m}$};
\node[taskbox, left=of out]    (task)  {Downstream\\prediction};

\draw[arr] (data)  -- (svd);
\draw[arr] (svd)   -- (cka);
\draw[arr] (cka)   -- (clust);
\draw[arr] (clust) -- (prior);
\draw[arr] (prior) -- (amp);
\draw[arr] (amp)   -- (out);
\draw[arr] (out)   -- (task);

\node[font=\small\itshape, text=RoyalPurple!80!black,
      left=0.2cm of data, rotate=90, anchor=south] {discover};
\node[font=\small\itshape, text=MyBlue,
      left=0.7cm of task, rotate=90, anchor=north] {integrate};

\end{tikzpicture}
    \caption{Thematic diagram for the \fancyname{} pipeline}
    \label{fig:thematic}
\end{figure}

\subsection{Initialization and clustered empirical Bayes estimation}{\label{sec:clus_emp_Bayes}

We initialize the AMP-based feature reconstruction using the appropriately rescaled top-$r_h$ principal components $\{(\upca_h,\vpca_h):h \in [m]\}$ of the matrices $\wb{\bm X}_h:=\bm X_h/\sqrt{n}$ for all $h \in [m]$. As $n,p_h \to \infty$, these empirical singular vectors satisfy \footnote{See Proposition~\ref{prop:singular_vect_approx} for a formalization of the approximation.}
\begin{align}
\label{eq:approx_singular_vectors}
\upca_h &\approx \bm U_h (\wh{\bm M}^{L}_h)^\top 
+ \zlpca_h (\wh{\bm \Sigma}^{L}_h)^{1/2}, 
\qquad 
\vpca_h \approx \bm V_h (\wh{\bm M}^{R}_h)^\top 
+ \zrpca_h (\wh{\bm \Sigma}^{R}_h)^{1/2},
\end{align}
where $\zlpca_h \in \R^{n \times r_h}$ and $\zrpca_h \in \R^{p_h \times r_h}$ have i.i.d.\ standard normal entries. Here, we define the diagonal matrices
$\wh{\bm M}^{\star}_h=\mathrm{diag}(\wh m^{\star}_{1,h},\ldots,\wh m^{\star}_{r_h,h})$
and
$\wh{\bm \Sigma}^{\star}_h=\mathrm{diag}(\wh \sigma^{\star}_{1,h},\ldots,\wh \sigma^{\star}_{r_h,h})$
for $\star \in \{L,R\}$, where the diagonal entries are given by \eqref{eq:nuisance_param_est_1} as follows.
\begin{align}
\label{eq:nuisance_param_est_1}
\wh\sigma^L_{i,h}&:=\frac{1+(\wh{\bm D}_h^2)_{ii}}{(\wh{\bm D}_h^2)_{ii}\{1+\gamma_h(\wh{\bm D}_h^2)_{ii}\}},\qquad
\wh m^L_{i,h}:=(1-\wh\sigma^L_{i,h})^{1/2}, \notag\\\wh\sigma^R_{i,h}&:=\frac{1+\gamma_h(\wh{\bm D}_h^2)_{ii}}{\gamma_h(\wh{\bm D}_h^2)_{ii}\{1+(\wh{\bm D}_h^2)_{ii}\}},\qquad\wh m^R_{i,h}:=(1-\wh\sigma^R_{i,h})^{1/2}, \quad \mbox{for $h \in [m]$ and $i \in [r_h]$.}
\end{align}
Furthermore, $\{\wh{\bm D}_h \in \R^{r_h \times r_h}:h \in [m]\}$ are diagonal matrices with entries given by
\begin{align}
\label{eq:approx_snr}
(\wh{\bm D}^2_h)_{ii}:=\frac{\gamma_h(\dpca_h)^2_{ii}-(1+\gamma_h)+
\sqrt{
\left[\gamma_h(\dpca_h)^2_{ii}-(1+\gamma_h)\right]^2-4\gamma_h}}{2\gamma_h},\; h \in [m],\;\; i \in [r_h],
\end{align}
where $\{\dpca_h:h \in [m]\}$ are defined in \eqref{eq:low_rank_init} of the appendix.
For a unified notation, we write $\widehat{\bm U}^{\mathrm{init}}_e:=\bm U^{\mathrm{pc}}_e$ when $e \in [m]$ and $\widehat{\bm U}^{\mathrm{init}}_e:=\widetilde{\bm X}_{e-m}$ when $e \in \{m+1,\ldots,m+\wt m\}$.

\begin{algorithm}[t]
\caption{Data Adaptive Intermediate Fusion (DAIF)}
\label{alg:daif_compact}
\begin{algorithmic}[1]
\footnotesize
\Input High-dimensional modalities $\{\bm{X}_h\}_{h=1}^m$,
       low-dimensional modalities $\{\widetilde{\bm{X}}_\ell\}_{\ell=1}^{\widetilde{m}}$,
       response $y$,
       ranks $\{r_h\}_{h=1}^m$ and $\{\widetilde r_\ell\}_{\ell=1}^{\widetilde{m}}$,
       AMP iterations $T$,
       clustering and prediction hyperparameters.
\Output Integrated embeddings $\bm{U}_T \in \mathbb{R}^{n \times r}$
        and pre-trained predictor $\mathfrak{H}(\,\cdot\,,\widehat\theta_{\mathrm{pred}})$. \vspace{2pt}

\State
$(\widehat{\mathscr F},
  \widehat{\mathscr U}^{\mathrm{init}})
\gets
\Call{SpectralInitialize}{
  \{\bm X_h\}_{h=1}^m,
  \{\widetilde{\bm X}_\ell\}_{\ell=1}^{\widetilde m},
  \{r_h\}_{h=1}^m
}$
using Algorithm~\ref{alg:daif_spectral_init}.

\State Obtain the estimated priors $\wh \mu$ and $\{\wh \nu_h:h \in [m]\}$ by 
\[
\left(
\mathfrak N,\widehat K,
\{\widehat{\mathcal C}_k\}_{k=1}^{\widehat K},
\{\widehat\mu_k\}_{k=1}^{\widehat K},
\{\widehat\nu_h\}_{h=1}^m
\right)
\gets
\Call{ClusteredEB}{
  \widehat{\mathscr F},
  \widehat{\mathscr U}^{\mathrm{init}}
}
\]
using Algorithm~\ref{alg:daif_clustered_eb}.

\State Get
\[
(\widebar{\bm U}_T,\widehat{\mathscr A}_T)
\gets
\Call{AMPFeatureExtract}{
    \{\widebar{\bm X}_h\}_{h=1}^m,\,
    \{\widetilde{\bm X}_\ell\}_{\ell=1}^{\widetilde m},\,
    \widehat{\mathscr F},\,
    \{(\widehat{\mathcal C}_k,\widehat\mu_k)\}_{k=1}^{\widehat K},\,
    \{\widehat\nu_h\}_{h=1}^m,\,
    T
},\]
using Algorithm~\ref{alg:daif_amp_feature_extract}.
\State Set $\wb{\bm{U}}_T \gets
       \bigl[\widehat{\bm{U}}_{T,1}\ \cdots\ \widehat{\bm{U}}_{T,m}\ \;
              \widetilde{\bm{U}}_{T,1}\ \cdots\
              \widetilde{\bm{U}}_{T,\widetilde{m}}\bigr]$.\vspace{2pt}

\State For each $h \in [m]$, compute $\widehat{\bm{U}}^{\rsc}_h$ using $\widehat{\mathscr A}_T$.
\State Optimize $\widehat\theta_{\mathrm{pred}}$ via \eqref{eq:opt_theta}
       and return predictor $\mathfrak{H}(\,\cdot\,,\widehat\theta_{\mathrm{pred}})$
       defined in \eqref{eq:frak_pred}.

\State \Return $\wb{\bm{U}}_T$,\;
       $\mathfrak{H}(\,\cdot\,,\widehat\theta_{\mathrm{pred}})$.
\end{algorithmic}
\end{algorithm}

Using \eqref{eq:approx_singular_vectors}, we can conclude that the dependence across the principal components are mediated by the dependence of $\{\bm U_h:h \in [m]\}$ and $\{\wt{\bm U}_\ell:\ell \in [\wt m]\}$. In the hierarchical clustering step, we construct a dissimmilarity matrix $\mathfrak N\in
\mathbb R^{(m+\widetilde m)\times(m+\widetilde m)},$ for $e,e'\in\mathcal E,$ such that $\mathfrak N_{ee'} \in [0,1]$ with $\mathfrak N_{ee'}=0$ implying $e$ and $e'$ are dependent, and $\mathfrak N_{ee'}=1$ implying they are independent. For each pair $e,e'\in\mathcal E$, $\aleph_n$ denotes an empirical non-parametric measure of dependence between two matrices with $n$ rows, and define the empirical affinity $\widehat s_{ee'}:=\aleph_n(\widehat{\bm U}^{\mathrm{init}}_e,\widehat{\bm U}^{\mathrm{init}}_{e'})$. 
In our implementation, we take $\aleph_n(\widehat{\bm U}^{\mathrm{init}}_e,\widehat{\bm U}^{\mathrm{init}}_{e'})
:=\mathrm{CKA}(\widehat{\bm U}^{\mathrm{init}}_e,\widehat{\bm U}^{\mathrm{init}}_{e'})$, where CKA is the centered kernel alignment (cf. \eqref{eq:def_hsic_sample} of the appendix).
The dissimilarity matrix $\mathfrak N\in\mathbb R^{(m+\widetilde m)\times(m+\widetilde m)}$ is defined as $\mathfrak N_{ee'}:=1-\widehat s_{ee'}$, for $e,e'\in\mathcal E.$

To estimate $\mu_k$ for any cluster $\wh{\mathcal C}_k$, we optimize
\begin{align}
\label{eq:lik_cluster}
    \wh \mu_k:=\argmax_{\mathfrak m_k \in \mathcal P(\R^{\wh r_k})}\,\,\mathrm L_k(\{\wh{\bm U}^{\mathrm{init}}_e:e \in \wh{\mathcal C}_k\};\mathfrak m_k ), \quad \mbox{where $\widehat r_k:=\sum_{e\in\widehat{\mathcal C}_k}\rho_e,$}
\end{align}
$\mathcal P(\R^{\wh r_k})$ is an appropriate class of distributions supported on $\R^{\wh r_k}$, and the non-parametric likelihood $\mathrm L_k(\cdot)$ is defined in \eqref{eq:approx_singular_vectors}.
The final empirical Bayes estimate of $\mu$ is $\wh \mu:=\wh \mu_1\times \cdots\times \wh \mu_{K}$. Similarly, for any modality $h \in [m]$, we obtain the estimate $\wh \nu_h$ of $\nu_h$, the prior distribution on the right latent factors of modality $h$, as $\wh \nu_h:=\argmax_{\mathfrak n_h \in \mathcal P_h(\R^{r_h})}\mathrm L^r_h(\vpca_h;\mathfrak n_h),$ where $\mathcal P_h(\R^{r_h})$ is a modality specific prior class for $\nu_h$. They are later utilized in the AMP iterations in Section~\ref{sec:feature_extract}. In our implementations, we use Gaussian mixture models with isotropic component covariances.
Observe that the estimation of $\wh \mu$ fuses information across modalities only under evidence of dependence ensuring data adaptive fusion, thereby improving downstream performance
\footnote{The precise forms of the likelihoods $\{\mathrm L_k:k \in [K]\}$ and $\{\mathrm L^r_h: h \in [m]\}$ are provided in Section~\ref{sec:emp_bayes_clust}.}.

\subsection{Feature extraction using Approximate Message Passing}
\label{sec:feature_extract}
The principal component scores $\{\upca_h,\vpca_h:h \in [m]\}$ are iteratively refined using Approximate Message Passing to reconstruct $\bm U_1,\ldots,\bm U_m$. In that direction, for $t \ge 0$ we initialize $\wh{\bm G}_{0,h}= \vpca_h$, and $\wh{\bm U}_{-1,h}=\upca_h(\wh{\bm \Sigma}^L_h)^{1/2}$ for $h \in [m]$. For each cluster $\wh{\mathcal C}_k$ and modality $h \in [m]$, consider
\begin{align}
    \hV_{t,h} &= v_{t,h}(\hG_{t,h};\wh{\nu}_h)\label{eq:denoise_cols}\\
\hF_{t,h} &= \bX_h\hV_{t,h}-\gamma_h\hU_{t-1,h}(\wh{\bdJ}^R_{t,h}(\hG_{t,h};\wh\nu_h))^\top\\ 
\hU_{t,h} &= u_{t,h,k}\left(\left\{\hF_{t,e}: e \in \wh{\mathcal  C}_{k,\mathrm H}\right\},\left\{\tX_{\ell}: \ell \in \wh{\mathcal  C}_{k,\mathrm L}\right\};\wh{\mu}_{k}\right)\label{eq:denoise_rows}\\
\hG_{t+1,h} &= \bX^\top_h\hU_{t,h}-\hV_{t,h}\left(\wh{\bdJ}^L_{t,h,k}\left(\left\{\hF_{t,e}: e\in \wh{\mathcal  C}_{k,\mathrm H}\right\},\left\{\tX_{\ell}: \ell \in \wh{\mathcal  C}_{k,\mathrm L}\right\};\wh\mu_k\right)\right)^\top.\label{eq:rough_denoise_rows}
\end{align}
In the above, suppose $(U_1,\ldots,U_m,\wt U_1,\ldots,\wt U_{\wt m}) \sim \wh\mu$ and
$V_h \sim \wh\nu_h$ for all $h \in [m]$. Then $u_{t,h,k}:\mathbb R^{\wh r_k} \to \mathbb R^{r_h}$ for $h \in \wh{\mathcal C}_{k,\mathrm H}$ (where $\wh r_k$ is defined in \eqref{eq:lik_cluster}) and
$v_{t,h}:\mathbb R^{r_h} \to \mathbb R^{r_h}$ are defined as
\begin{align}
&v_{t,h}(y_h)
=
\mathbb E_{\wh\nu_h}\left[
\,V_h \,\middle|\,
\hM^R_{t,h}V_h + (\hS^R_{t,h})^{1/2} Z_{h,R} = y_h
\right], \label{eq:denoiser_func_cols}\\
&u_{t,h,k}\bigl(\{x_e:e \in \wh{\mathcal C}_{k,\mathrm H}\},\{\wt x_{\ell}:\ell \in \wh{\mathcal C}_{k,\mathrm L}\}\bigr) = \mathbb E_{\wh\mu_k}\Bigl[ \,U_h \,\Big|\, \bigl\{\hM^L_{t,e}U_e + (\hS^L_{t,e})^{1/2}Z_{e,L} = x_e : e \in \wh{\mathcal C}_{k,\mathrm H}\bigr\}, \notag\\ &\hspace{25em} \bigl\{\wh{\bm L}_\ell \wt U_{\ell} + \wt Z_{\ell,L}=\wt x_\ell : \ell \in \wh{\mathcal C}_{k,\mathrm L}\bigr\} \Bigr], \label{eq:denoiser_func_rows}
\end{align}
where $\{Z_{s,L}: s \in [m]\}$, $\{Z_{h,R}: h \in [m]\}$, and
$\{\wt Z_{\wt s,L}: \wt s \in [\wt m]\}$ have independent $\dnorm(0,1)$
entries, are mutually independent, and are independent of
$(U_1,\ldots,U_m,\wt U_1,\ldots,\wt U_{\wt m})$ and $\{V_h:h \in [m]\}$. The state-evolution matrices are defined as follows: let $\hM^{R}_{0,h}=\hM^R_h$ and $\hS^{R}_{0,h}=\hS^R_h$ and for all $t \ge 0$, construct 
$\hM^{L}_{t,h}:=\hS^{L}_{t,h}\wh{\bm D}_h, \hS^{L}_{t,h}:=n^{-1}\cdot(\hV_{t,h})^\top \hV_{t,h}, \hM^{R}_{t+1,h}:=\hS^{R}_{t+1,h}\wh{\bm D}_h,$ and $\hS^{R}_{t+1,h}:=n^{-1}\cdot(\hU_{t,h})^\top \hU_{t,h},$
where $\wh{\bm D}_h$ is defined in \eqref{eq:approx_snr} for all $h \in [m]$ . Furthermore, the estimates $\wh{\bm L}_\ell$ for the matrices $\bm L_\ell$ for $\ell \in [\wt m]$ are given by $\wh{\bm L}_\ell :=
\left(\frac{1}{n}\tX_\ell^\top \tX_\ell - \bm I_{\wt r_\ell}\right)^{1/2}$, for $\ell \in [\wt m]$.

In \eqref{eq:denoise_cols}, the denoiser $v_{t,h}(\cdot)$ is applied row-wise to the matrix argument,
and the associated Onsager correction $\wh{\bm J}^R_{t,h}(\hG_{t,h};\wh\nu_h)$ is obtained by computing
the Jacobian of $v_{t,h}(\cdot)$ at each row of $\wh{\bm G}_{t,h}$ and averaging across rows (cf. Section~\ref{sec:amp_details}). 
Similarly, in \eqref{eq:denoise_rows}, the denoiser $u_{t,h,k}(\cdot)$ is applied row-wise to the
collection of matrices in its argument. The corresponding Onsager term $\wh{\bm J}^L_{t,h,k}$ in
\eqref{eq:rough_denoise_rows} is computed by evaluating the Jacobian row-wise, restricting to the
columns associated with modality $h$, and then averaging over rows (cf. Section~\ref{sec:amp_details}). 

The low-dimensional latent signals
$\tU_1,\ldots,\tU_{\wt m}$ are estimated by Bayesian denoising motivated by
\eqref{eq:multimodal_low_dim_model} and \eqref{eq:prior_specification}. In
particular, for all $\ell \in \wh{\mathcal C}_k$, we define
\begin{equation}
\label{eq:denoise_low_dim}
\tU_{t,\ell}=\wt u_{t,\ell,k}\!\left(\left\{\hF_{t,s}: s \in \wh{\mathcal C}_{k,\mathrm H}\right\},\left\{\tX_{\wt s}: \wt s \in \wh{\mathcal C}_{k,\mathrm L}\right\};\wh\mu_k\right),
\end{equation}
where the denoiser $\wt u_{t,\ell,k}:\mathbb R^{\wh r_k} \to \mathbb R^{\wt r_\ell}$ in \eqref{eq:denoise_low_dim} is given by
\begin{align}
\wt u_{t,\ell,k}\bigl(\{x_s:s \in \wh{\mathcal C}_{k,\mathrm H}\},\{\wt x_{\wt s}:\wt s \in \wh{\mathcal C}_{k,\mathrm L}\}\bigr)
&=
\mathbb E_{\wh\mu_k}\Bigl[\,\wt U_\ell \,\Big|\,\bigl\{\hM^L_{t,s}U_s + (\hS^L_{t,s})^{1/2} Z_{s,L} = x_s : s \in \wh{\mathcal C}_{k,\mathrm H}\bigr\}, \notag\\
&\hspace{6em}
\bigl\{\wh{\bm L}_{\wt s} \wt U_{\wt s} + Z_{\wt s,L} = \wt x_{\wt s} : \wt s \in \wh{\mathcal C}_{k,\mathrm L}\bigr\}\Bigr],
\end{align}
and is applied row-wise to the matrices in its arguments. The collections
$(U_1,\ldots,U_m,\wt U_1,\ldots,\wt U_{\wt m})$, $\{Z_{s,L}: s \in [m]\}$, and
$\{\wt Z_{\wt s,L}: \wt s \in [\wt m]\}$ are as defined in
\eqref{eq:denoiser_func_cols} and \eqref{eq:denoiser_func_rows}. 

The foregoing scheme alternates between power iteration and Bayesian denoising to refine the noisy latent estimates obtained from the principal components, namely $\{\upca_h,\vpca_h:h \in [m]\}$. The power iteration step amplifies the low-rank signal ensuring greater correlation of the estimates with the true latents, while the Onsager correction removes the bias introduced by repeatedly applying the same data matrix across iterations, thereby ensuring that the iterates behave approximately like the true latent variables corrupted by Gaussian noise. This asymptotic Gaussianity can be used to construct empirical Bayes denoisers without any MCMC iterations, which shrink the iterates toward the estimated prior mean by borrowing information across samples and related modalities while being considerably fast. The cluster based priors, learned adaptively from the data, further prevent over-shrinkage by integrating information only across mutually informative modalities. The resulting embeddings
\(
\bU_t=\left[\hU_{t,1} \;\; \cdots \;\; \hU_{t,m} \;\; \tU_{t,1} \;\; \cdots \;\; \tU_{t,\wt m}\right]
\in \mathbb R^{n \times r},
\)
provide a low-dimensional representation of the observational units for downstream unsupervised analysis. 

\subsection{Downstream prediction using extracted embeddings}
\label{sec:ols_pretraining}
We now describe  our algorithm for downstream prediction using the extracted embeddings after $T$ AMP iterations. While AMP-based methods allow us to construct embeddings $\wh \bdU_{T,h}$ and $\wt \bdU_{T,\ell}$, the validity of such constructions relies on the ``large $n$" asymptotics. Consequently, extending such construction is impossible for the test data, where the sample size is comparatively small. Therefore, the features used for prediction in the test samples are obtained by projecting the test observations onto the estimated loading spaces spanned by the columns of $\wh{\bm V}_{T,h}$.
To maintain parity between the distributions of the training and test embeddings, the features $\wh{\bm F}_{T,h}$ obtained from the AMP iterations \eqref{eq:denoise_cols}-\eqref{eq:rough_denoise_rows} are rescaled as $\wh{\bm U}^{\rsc}_{h}:=\wh{\bm F}_{t,h}(\wh{\bm \Sigma}^L_{T,h})^{-1}\wh{\bm D}^{-1}_h 
\in \R^{n \times r_h},$ for $h \in [m]$ \footnote{An alternative sample-splitting based approach is described in Section~\ref{sec:sample_splitting}.}.
The state evolution of the AMP algorithm (Theorem~\ref{thm:w_2_amp} of \cite{nandy2026daif_supp}) can be leveraged to show that $\wh{\bm U}^{\rsc}_{h} \approx \bm U_h+\bm Z_{h,\rsc}(\wh{\bm D}^{-1}_h(\hS^L_{T,h})^{-1}\wh{\bm D}^{-1}_h)^{1/2},$
where the entries $\bm Z_{h,\rsc} \in \R^{n \times r_h}$ are i.i.d standard Gaussian and independent of $\bm U_1,\ldots,\bm U_m,$ and $\tU_1,\ldots,\tU_{\wt m}$. 

Based on the foregoing representation, we estimate the unknown link function $g$ in \eqref{eq:response_vector} from $y$ and the collections $\{\wh{\bm U}^{\rsc}_{h}: h \in [m]\}$, and $\{\tX_\ell:\ell \in [\wt m]\}$ by optimizing the following parametric loss function $\mathfrak L_\theta:(\R \times \R^{r})^{\otimes n}  \rightarrow \R$ in terms of the parameter $\theta \in \R^r$. The dimension of the parameter depends on the model architecture. To define the loss function, first define
\begin{align}
\label{eq:u_pred_pop_def}
U^{\mathrm{pred}}_e = \begin{cases}
    U_e+(\wh{\bm D}^{-1}_e(\hS^L_{T,e})^{-1}\wh{\bm D}^{-1}_e)^{1/2}Z_{e,\rsc} & \mbox{if $e \in \mathcal E_{\mathrm H}$}\\
    \wh{\bm L}_{e-m} \wt U_{e-m} + Z_{e-m,L} & \mbox{if $e \in \mathcal E_{\mathrm L}$},
\end{cases}
\end{align}
where $Z_{1,\rsc},\ldots,Z_{m,\rsc}$ are vectors of dimensions $r_1,\ldots,r_m$ that are mutually independent, independent of $(U_1,\ldots,U_m,\wt U_1,\ldots,\wt U_{\wt m})$ and have i.i.d standard Gaussian entries. Similarly, $Z_{1,L},\ldots,Z_{\wt m,L}$ are vectors of dimensions $\wt r_1,\ldots,\wt r_{\wt m}$ that are mutually independent, independent of the collection $(U_1,\ldots,U_m,\wt U_1,\ldots,\wt U_{\wt m})$ and $\{Z_{1,\rsc},\ldots,Z_{m,\rsc}\}$ and have i.i.d standard Gaussian entries.
The empirical counterpart $\wh{\bm U}^{\mathrm{pred}}_e$ is defined to be $\wh{\bm U}^{\rsc}_{e}$ if $e \in \mathcal E_{\mathrm H}$ and $\tX_{e-m}$ if $e \in \mathcal E_{\mathrm L}$.
For any $\theta \in \R^r$, consider arbitrary cluster-wise predictors $g_{k,\theta}:\R^{\wh r_k} \rightarrow \R$, and define the aggregated predictor $\mathfrak H:\R^{r} \times \R^r \to \R$ as
\begin{align}
    \label{eq:frak_pred}
  \mathfrak H(x,\theta):= \mathbb E_{\wh \mu}\left[\sum_{k=1}^{K}g_{k,\theta}\left(\{\mathcal U_e:e \in \wh{\mathcal C}_k\}\right) ~\Bigg|~\left\{U^{\mathrm{pred}}_e=x_e: e \in [m+\wt m]\right\}\right],
\end{align}
where $x_e$ is a vector that collects the entries of $x$ corresponding to modality $e$ and the collection of vectors $\{U^{\mathrm{pred}}_e:e \in [m+\wt m]\}$ are defined in \eqref{eq:u_pred_pop_def} with $(U_1,\ldots,U_m,\wt U_1,\ldots,\wt U_{\wt m}) \sim \wh \mu$. 

Now, the loss function is defined as
\begin{align}
\label{eq:def_frak_l}
    \mathfrak L_\theta(y,\{\wh{\bm U}^{\rsc}_{h}: h \in [m]\},\{\tX_\ell:\ell \in [\wt m]\}):=\frac{1}{n}\sum_{i=1}^n\mathcal L\left(y_i,\mathfrak H(\{(\wh{\bm U}^{\mathrm{pred}}_e)_{i*}:e \in [m+\wt m]\},\theta)\right),
\end{align}
where $\mathcal L$ quantifies the discrepancy between the observed response $y_i$ and the prediction produced by $\mathfrak H$ from the estimated latent features of the $i$-th observation. We intentionally keep the loss function general to accommodate diverse supervised learning objectives. For continuous responses, one can adopt the squared error loss given by $\mathcal L(y,\mathfrak H(x)):=(y-\mathfrak H(x))^2$, whereas for binary response one can adopt the binary cross entropy loss given by $\mathcal L(y,\mathfrak H(x)):=-y\log \mathfrak H(x)-(1-y)\log\left(1-\mathfrak H(x)\right)$.
We optimize $\mathfrak L_\theta$ to obtain
\begin{align}
\label{eq:opt_theta}
\wh \theta_\pred:=\argmin_{\theta \in \R^r}\mathfrak L_\theta(y,\{\wh{\bm U}^{\rsc}_{h}: h \in [m]\},\{\tX_\ell:\ell \in [\wt m]\}).
\end{align}
Our loss function is motivated by \citet{HuKeLiu2022MeasurementError}. 
This objective can be minimized by approximating the posterior expectation by averaging over MCMC samples from the posterior and any open source optimization module like \texttt{Adam} \citep{kingma2017adam}. In the special case of linear predictor, when it is assumed that all $g_{k,\theta}$ for $k \in [K]$ are linear functions of the co-ordinates of $\theta$ corresponding to cluster $\wh{\mathcal C}_k$, i.e., $g_{k,\theta}(u):= \theta^\top_k u$, where $u \in \R^{\wh r_k}$, and $\theta:=(\theta_1,\ldots,\theta_K) \in \R^{r}$, the optimization in \eqref{eq:opt_theta} reduces to finding the ordinary least squares estimator of $y$ using the covariates $\{\wb{\bm U}^{\pred}_{t,h}: h \in [m]\}$ and $\{\tU^{\pred}_\ell:\ell \in [\wt m]\}$ defined as follows:
\begin{align}
\label{eq:pred_linear embedding}
    \wb{\bm U}^{\pred}_{h}&:=u^\pred_{T,h,k}\left(\left\{\wh{\bm U}^{\rsc}_{e}: e \in \wh{\mathcal C}_{k,\mathrm H}\right\},\left\{\tX_{\ell}: \ell \in \wh{\mathcal  C}_{k,\mathrm L}\right\};\wh{\mu}_{k}\right), \quad \mbox{where $h \in \wh{\mathcal C}_{k,\mathrm H}$,}\\
    \tU^{\pred}_\ell&:=\wt u^\pred_{T,\ell,k}\left(\left\{\wh{\bm U}^{\rsc}_{e}: e \in \wh{\mathcal  C}_{k,\mathrm H}\right\},\left\{\tX_{\ell}: \ell \in \wh{\mathcal  C}_{k,\mathrm L}\right\};\wh{\mu}_{k}\right), \quad \mbox{where $\ell \in \wh{\mathcal C}_{k,\mathrm L}$}.
\end{align}
In the foregoing display, the denoiser functions are defined as 
\begin{align}
& u^\pred_{T,h,k}\bigl(\{x_e:e \in \wh{\mathcal C}_{k}\}\bigr) := \mathbb E_{\wh\mu_k}\Bigl[ \,U_h \,\Big|\, \left\{U^{\mathrm{pred}}_{e}=x_e: e \in \wh{\mathcal C}_k\right\} \Bigr] \quad \mbox{and}\\
&\wt u^\pred_{T,\ell,k}\bigl(\{x_e:e \in \wh{\mathcal C}_k\}\bigr):=
\mathbb E_{\wh\mu_k}\Bigl[\,\wt U_\ell \,\Big|\,\left\{U^{\mathrm{pred}}_{e}=x_e: e \in \wh{\mathcal C}_k\right\} \Bigr],
\end{align}
where the collection of vectors $(U^{\mathrm{pred}}_1,\ldots,U^{\mathrm{pred}}_m,U^{\mathrm{pred}}_1,\ldots,U^{\mathrm{pred}}_{\wt m})$ is defined as in \eqref{eq:u_pred_pop_def} with $(U_1,\ldots,U_m,\wt U_1,\ldots,\wt U_{\wt m}) \sim \wh \mu$.

Now consider a test sample with multimodal features
$\{X^\test_h:h \in [m]\}$ (high-dimensional) and $\{\wt X^\test_\ell:\ell \in [\wt m]\}$
(low-dimensional), and an unknown response $y_\test$ generated as follows.
\begin{align}
    \label{eq:test_high_dim}
    X^\test_h&= \frac{1}{\sqrt{n}} \bdV_h \bm D_h U^\test_h + Z^\test_h
    \in \mathbb R^{p_h},\;\text{for}\,h \in [m],\quad \wt X^\test_\ell= \bm L_\ell \wt U^\test_\ell + \wt Z^\test_\ell\in \mathbb R^{\wt r_\ell},\; \text{for}\,\ell \in [\wt m],
\end{align}
where $(U^\test_1,\ldots,U^\test_m,\wt U^\test_1,\ldots,\wt U^\test_{\wt m}) \sim \mu$.
The noise vectors $\{Z^\test_h:h \in [m]\}$ and
$\{\wt Z^\test_\ell:\ell \in [\wt m]\}$ have independent standard Gaussian entries,
are mutually independent, and are independent of both the latent variables and the
training data. The latent vectors
$(U^\test_1,\ldots,U^\test_m,\wt U^\test_1,\ldots,\wt U^\test_{\wt m})$
are also independent of the training data. Finally, the test response $y_\test$ is
generated from \eqref{eq:response_vector}, independently of the training data.

Based on \eqref{eq:test_high_dim}, we adopt an ordinary least squares (OLS) approach to estimate the latent factors $U^{\test}_1,\ldots,U^{\test}_m$ corresponding to the test data. Specifically, define $\wh{\bm R}_{T,h} := n^{-1}\cdot\wh{\bm V}_{T,h}\wh{\bm D}_h \in \mathbb R^{p_h \times r_h}.$
For each $h \in [m]$, we form the test-time OLS embedding $\wh U^{\test,\ols}_h:=n^{-1/2}\cdot\bigl(\wh{\bm R}^\top_{T,h}\wh{\bm R}_{T,h}\bigr)^{-1}\wh{\bm R}^\top_{T,h} X^\test_h\in \mathbb R^{r_h}.$ One can again use Theorem~\ref{thm:ols_estimator} to show that $\wh U^{\test,\ols}_h \approx U_h+(\wh{\bm D}^{-1}_h(\hS^L_{T,h})^{-1}\wh{\bm D}^{-1}_h)^{1/2}Z^{\test,\ols}_{h}$, where $Z^{\test,\ols}_{h} \sim \dnorm_{r_h}(0,\bm I_{r_h})$ for all $h \in [m]$. Therefore, the distribution of the features used for training the predictor $\mathfrak H(\cdot)$ match that of the features derived from the test observations during inference. This allows us to use the trained predictor for prediction on the test data. In this direction, define $\wh U^{\test}_e=\wh U^{\test,\ols}_e$ if $e \in \mathcal E_{\mathrm H}$ and $\wh U^{\test}_e=\wt X^\test_{e-m}$ if $e \in \mathcal E_{\mathrm L}$. Then, the test response is predicted by evaluating the pre-trained predictor (defined in \eqref{eq:frak_pred}) on the denoised test embeddings as $\wh y_\test:=\mathfrak H(\{\wh U^{\test}_e:e \in [m+\wt m]\},\wh\theta_{\mathrm{pred}}),$
where $\wh\theta_{\mathrm{pred}}$ is obtained from \eqref{eq:opt_theta}. Note that $\mathfrak H$ uses $\wh \mu$ estimated from the training data.

Observe that when the number of clusters is $m+\wt m$, the information is integrated only at the prediction stage corresponding to \emph{late fusion}. Here, the AMP component reduces to applying \emph{Empirical Bayes PCA} \citep{zhong2022empirical} separately to the high-dimensional modalities and denoising the low-dimensional modalities separately, followed by using the reconstructed features in the supervised learning stage. On the other hand, when the number of clusters is one, information is fused across all modalities both during feature engineering and prediction. In this regime, the AMP component reduces to \emph{Orchestrated Approximate Message Passing} \citep{nandy2024multimodal} applied jointly across all modalities. Therefore, \fancyname{} subsumes these two existing pipelines by adaptively interpolating between them based on the evidence of inter-modal dependence present in the data. Other than these two extreme cases, everything else corresponds to intermediate fusion.

\section{Numerical Experiments}
\label{sec:numerical_exp}

In this section, we benchmark the performance of \fancyname{} against several widely used multimodal data integration techniques. Our evaluation considers both recovery of the underlying latent factors and predictive performance in a downstream supervised task.

In the first experiment, we compared the performance of \fancyname{} in terms of latent factor reconstruction against popular multimodal data integration benchmarks, namely, \texttt{AJIVE} \citep{feng2018angle}, multiview canonical correlation analysis (\texttt{MCCA}) \citep{kettenring1971canonical}, generalized canonical correlation analysis (\texttt{GCCA}) \citep{carroll1968gca}, and a version of hierarchical PCA \footnote{We provide a comparison with two AMP-based data-integration pipelines, EB-PCA \citep{zhong2022empirical}, and OrchAMP \citep{nandy2024multimodal} in Section~\ref{sec:eff_sample_size} which are special cases of DAIF corresponding to $K=3$ and $K=1$.}. The detailed implementation of each non-AMP-based method is described in Section~\ref{sec:baseline_exp_des}. In the second experiment, we compared the performance of predictors constructed using the features obtained from each integration pipeline in terms of prediction accuracy in linear regression. We also compared \fancyname{} against the supervised cooperative learning approach of \cite{ding2022cooperative}.

We considered $m=3$ high-dimensional feature modalities $\bm X_1 \in \mathbb R^{n \times p_1}$, $\bm X_2 \in \mathbb R^{n \times p_2}$, and $\bm X_3 \in \mathbb R^{n \times p_3}$ generated according to \eqref{eq:multimodal_factor_model}. To operate in the proportional asymptotic regime, we set $p_1=\lfloor 0.25n \rfloor$, $p_2=\lfloor 0.8n \rfloor$, and $p_3=\lfloor 3n \rfloor$. The rows of $\bm U_1\in \mathbb R^{n \times 2}$ were generated as $\smash{(\bm U_1)_{i*} \overset{\mathrm{i.i.d}}{\sim}}~\mathrm{Unif}\{\pm 1\}^{\otimes 2}$. To generate the rows of $\bm U_2 \in \mathbb R^{n \times 3}$, we took $\smash{(\bm U_2)_{i,1:2}=(\bm U_1)_{i*}}$ for all $i \in [n]$. The third column of $\bm U_2$ was independently generated as $\smash{(\bm U_2)_{i3} \overset{\mathrm{i.i.d}}{\sim} 0.25 \,\delta_{(-\sqrt{2})}+0.5~\delta_0+0.25 \,\delta_{(\sqrt{2})}}$ for $i \in [n]$. The rows of $\bm U_3 \in \mathbb R^{n \times 2}$ were generated independently from $\mathsf{G}^{\otimes 2}$, where $\mathsf G$ is the distribution of $(\mathsf{W}-\mathbb E[\mathsf{W}])/\sqrt{\mathrm{Var}(\mathsf{W})}$, where $\mathsf{W}:=\mathrm{sgn}(\mathsf U)\times |\mathsf U+\mathsf \epsilon|^{1/3}$, with $\mathsf{U} \sim \mathrm{Unif}\{\pm 1\}$ and $\epsilon \sim \dnorm(0,0.1)$. The rows of $\bm V_1 \in \mathbb R^{p_1 \times 2}$ were generated independently from $\mathrm{Unif}\{\pm 1\}^{\otimes 2}$, those of $\bm V_2\in \mathbb R^{p_2 \times 3}$ were generated independently from $\mathsf{H}^{\otimes 3}$ where $H$ is the distribution of $\mathsf{V}/\sqrt{1.25}$ where $\mathsf{V} \sim 0.5~\dnorm(-1,0.25)+0.5~\dnorm(1,0.25)$, and finally, those of $\bm V_3\in \mathbb R^{p_3 \times 2}$ were generated independently from $\mathrm{Laplace}(1/\sqrt{2})^{\otimes 2}$. The signal strengths were controlled via diagonal matrices $\bm D_h$ with $(\bm D_h)_{ii} = 5i$ for $i \in [r_h]$, where $r_1=2, r_2=3$ and $r_3=2$. 

The above choices ensured that the prior distributions for the latents are heterogeneous across the modalities and vary over different classes of distribution.
In particular, the rows of $\bm U_1$ were generated from a balanced discrete prior, whereas the rows of $\bm U_2$ were partially correlated with $\bm U_1$ but had an additional independent coordinate, yielding an unbalanced discrete prior with richer cluster structure. By contrast, the rows of $\bm U_3$ were generated independently of the first two modalities and had a continuous non-Gaussian distribution. Similarly, the right latent factors were chosen to cover qualitatively different tails and support behaviors. The rows of $\bm V_1$ were discretely supported, the rows of $\bm V_2$ were continuously supported with sub-Gaussian tails, and the rows of $\bm V_3$ had sub-exponential tails.

Next, consider the generation of the response vectors for the supervised learning experiments. The response vector $y$ was generated according to~\eqref{eq:response_vector}, with $g(\,\cdot\,\mid \theta_i) \sim \dnorm(\theta_i,0.01)$. To generate the mean $\theta_i$ according to \eqref{eq:response_vector}, we considered two link functions:
\begin{enumerate}
    \item \emph{Linear link:} $\smash{f_1\big((\bm U_1)_{i*},(\bm U_2)_{i*},(\bm U_3)_{i*}\big)= \beta^\top_1(\bm U_1)_{i*}+\beta_2^\top (\bm U_2)_{i*}+\beta_3^\top (\bm U_3)_{i*}}$,
    \item \emph{Single-index link:} $\smash{f_2\big((\bm U_1)_{i*},(\bm U_2)_{i*},(\bm U_3)_{i*}\big)= (\beta^\top_1(\bm U_1)_{i*}+\beta_2^\top (\bm U_2)_{i*})^2+\beta_3^\top (\bm U_3)_{i*}}$,
\end{enumerate}
where the entries of $\beta:=(\beta_1,\beta_2,\beta_3)^\top$ were generated independently as $\mathrm{Unif}\{\pm 1\}$.

\begin{table}[b]
\centering
\footnotesize
\setlength{\tabcolsep}{5pt}
\renewcommand{\arraystretch}{0.95}
\begin{tabular}{llrrrrr}
\toprule
Modality & $n$ & \textbf{DAIF-CKA} & \textbf{AJIVE} & \textbf{MCCA} & \textbf{GCCA} & \textbf{HPCA} \\
\midrule
\multirow{3}{*}{$\mathcal{E}_1$} & 2000 & \textbf{0.013(0.001)} & 0.183(0.001) & 0.371(0.002) & 0.366(0.002) & 0.367(0.002) \\
 & 2500 & \textbf{0.013(0.001)} & 0.184(0.001) & 0.374(0.001) & 0.369(0.001) & 0.370(0.001) \\
 & 3000 & \textbf{0.012(0.001)} & 0.182(0.001) & 0.371(0.002) & 0.366(0.002) & 0.366(0.002) \\
\midrule
\multirow{3}{*}{$\mathcal{E}_2$} & 2000 & \textbf{0.016(0.001)} & 2.119(0.001) & 0.140(0.001) & 0.140(0.001) & 0.144(0.001) \\
 & 2500 & \textbf{0.015(0.000)} & 2.117(0.001) & 0.139(0.001) & 0.139(0.001) & 0.142(0.001) \\
 & 3000 & \textbf{0.015(0.000)} & 2.116(0.001) & 0.137(0.001) & 0.137(0.001) & 0.139(0.001) \\
\midrule
\multirow{3}{*}{$\mathcal{E}_3$} & 2000 & \textbf{0.006(0.000)} & 2.003(0.001) & 0.035(0.000) & 0.035(0.001) & 0.036(0.000) \\
 & 2500 & \textbf{0.006(0.000)} & 2.003(0.001) & 0.035(0.000) & 0.035(0.000) & 0.036(0.000) \\
 & 3000 & \textbf{0.005(0.000)} & 2.003(0.001) & 0.035(0.000) & 0.035(0.000) & 0.035(0.000) \\
\bottomrule
\end{tabular}
\caption{Average reconstruction error $\mathcal{E}_k$ per modality across sample sizes (mean with standard error in parentheses). \textbf{DAIF-CKA}: AMP with CKA-based clustering (2 clusters). Bold = best mean per row.}
\label{tab:cka_vs_baselines}
\end{table}

\subsection{Comparing \fancyname{} against multimodal benchmarks}
\label{sec:rec_bench}
In this experiment, we benchmarked \fancyname{} against  \texttt{AJIVE}, \texttt{GCCA}, \texttt{MCCA}, and a version of hierarchical PCA described in Section~\ref{sec:baseline_exp_des} in terms of latent factor reconstruction through data integration. The performance as a function of $n$ was compared across the five competing methods as we vary $n$ in $\{2000, 2500, 3000\}$. To measure the reconstruction errors, we defined $\mathcal E_k := n^{-2}\cdot\|\wh{\bm U}_{k}\wh{\bm U}_{k}^\top - \bm U_k\bm U_k^\top\|_F^2$, for $k=1,2,3$.

We repeated each trial independently for 25 iterations and averaged the reconstruction error for all five competing methods. For \fancyname{}, we used 10 AMP iterations and CKA as the similarity measure for the hierarchical clustering step. The results are summarized in Table~\ref{tab:cka_vs_baselines}. Our results indicate that intermediate fusion via the AMP-based pipeline consistently achieves the lowest reconstruction error across all $n$, substantially outperforming classical multiview methods (\texttt{AJIVE}, \texttt{MCCA}, \texttt{GCCA}). \fancyname{} achieves the superior performance compared to the traditional methods by avoiding the assumption of shared latent factors with same dimension across modalities through flexible Bayesian modeling. Furthermore, the hierarchical clustering-based intermediate fusion also provides \fancyname{} an additional advantage when the intermodal dependence is low by preventing overintegration. The existing pipelines do not provide such flexibility and integrate all modalities by projecting them onto a shared latent space.

\begin{table}[t]
\centering
\scriptsize
\setlength{\tabcolsep}{5pt}
\renewcommand{\arraystretch}{1.0}
\begin{tabular}{@{}lccccc@{}}
\toprule
$n$ & DAIF-CKA & AJIVE & MCCA & GCCA & HPCA \\
\midrule
3000 & \textbf{0.223(0.039)} & 28.983(4.617) & 28.920(4.619) & 28.924(4.619) & 28.925(4.619) \\
3500 & \textbf{0.260(0.043)} & 36.015(5.566) & 35.959(5.567) & 35.962(5.566) & 35.962(5.566) \\
4000 & \textbf{0.239(0.040)} & 33.757(5.337) & 33.705(5.335) & 33.705(5.337) & 33.704(5.335) \\
\bottomrule
\end{tabular}
\caption{Prediction error (MSE) across sample sizes $n$. Lower is better; per-row minima are bold.}
\label{tab:pred_err_vs_baselines}
\end{table}

\subsection{Comparing prediction performance against popular fusion benchmarks}
\label{sec:ne_pred_oth_meth}
In this experiment, we compared the performance of intermediate fusion based prediction offered by \fancyname{} with some popular multimodal supervised learning frameworks. As benchmarks, we considered the popular multimodal dimension reduction pipelines \texttt{AJIVE}, \texttt{MCCA}, and \texttt{GCCA} to construct a low-dimensional embedding and fit a linear regression to predict the response $y$ using the denoised embeddings \footnote{A simulation study comparing prediction using \fancyname{} features with other AMP-derived features is provided in Section~\ref{sec:eff_pred_a}.}. The same procedure is adapted for \fancyname{}. It can be verified that in the \fancyname{} framework, this amounts to using the linear-architecture for $\mathfrak H(x,\theta)$ (as defined in \eqref{eq:frak_pred}) for all $k=1,\ldots,K$.
For this experiment, the response $y$ was
generated according to the procedure described at the beginning of the section with linear link function. Test data containing $\lfloor0.1n\rfloor$ samples were separately generated using the same data-generating mechanism as the training data. Observe that the data-generating mechanism in \eqref{eq:multimodal_factor_model} for the training features involves a division by $\sqrt{n_\train}$ in the signal component. To ensure train-test generation alignment, we scaled by $\sqrt{n_\train}$ instead of $\sqrt{n_\test}$ in the test features generation. The test data was projected to the same latent space as the training embeddings. Then the predictor training and test prediction followed as described in Sections~\ref{sec:ols_pretraining}. The test prediction for comparator methods are outlined in Section~\ref{sec:baseline_exp_des}. The prediction performance on test data was compared for all benchmarks using test mean squared error with $n_\train$ varying in $\{3000,3500,4000\}$. Our results are summarized in Table~\ref{tab:pred_err_vs_baselines}.

We also considered the cooperative learning based on $\ell_1$ penalty (henceforth referred to as \emph{Coop-Lasso}) framework from \cite{ding2022cooperative} as a benchmark for intermediate fusion-based prediction with multimodal features. In this framework, we minimized  
\begin{align}
    \mathcal{J}(\boldsymbol{\theta}) &:= \frac{1}{2n}\left\|y -( \bm X_1\theta_1+\bm X_2\theta_2+\bm X_3\theta_3)\right\|^2_2+\frac{\rho}{2n}\sum_{m < m'}\|X_m\theta_m - X_{m'}\theta_{m'}\|^2_2+ \sum_m \frac{\lambda}{p_m}\|\theta_m\|_1,
\end{align}
to estimate $\theta_1 \in \R^{p_1},\theta_2 \in \R^{p_2}$ and $\theta_3 \in \R^{p_3}$. The final estimate was given by $\wh y_\test:=x^\top_{1,\test}\wh\theta_1+x^\top_{2,\test}\wh\theta_2+x^\top_{3,\test}\wh\theta_3$. The hyperparameters $\rho$ and $\lambda$ were selected by assessing the prediction performance on a held-out validation subset consisting of a random subset of the training data (consisting of $20\%$ of the training samples). The regularization parameter $\lambda$ was selected by first computing the full Lasso path on the stacked prediction block with $\rho = 0$, then evaluating each candidate $\lambda$ on a fixed held-out validation set (20\% of training data) and picking the one minimizing validation MSE. This chosen $\lambda$ was subsequently frozen while $\rho$ is tuned separately over a coarse grid ${0.0, 0.5, 2.0}$ using the same validation split. However, since the above optimization is difficult to perform for large sample size (another advantage of \fancyname{} is that it is extremely scalable for large sample sizes), for this experiment, we varied the sample size in $n \in \{1000,1500,2000\}$. The results are summarized in Table~\ref{tab:coop_vs_daif}, and we observe that across all sample sizes \fancyname{} based predictor outperforms the pipeline from \cite{ding2022cooperative}.

\begin{table}[h]
\centering
\setlength{\tabcolsep}{5pt}
\renewcommand{\arraystretch}{1.0}
\begin{tabular}{@{}lccc@{}}
\toprule
Method & $n = 1000$ & $n = 1500$ & $n = 2000$ \\
\midrule
DAIF & \textbf{0.0415(0.0035)} & \textbf{0.0276(0.0017)} & \textbf{0.0252(0.0010)} \\
Coop. Lasso & 0.8051(0.0998) & 0.9092(0.0907) & 0.7338(0.0870) \\
\bottomrule
\end{tabular}
\caption{Prediction error across sample sizes $n$. Lower is better; per-column minima are bold.}
\label{tab:coop_vs_daif}
\end{table}

\section{Real data analysis}
\label{sec:real_data}
To demonstrate the prowess of \fancyname{} for supervised learning problems, we consider two real data analysis scenarios.

\subsection{A TEA-seq data analysis}
\label{sec:tea_seq}
We analyzed the trimodal TEA-seq PBMC dataset of \citet{Swanson2021}, previously considered by \citet{nandy2024multimodal}. After preprocessing and quality control steps (cf. Section~\ref{sec:tea_details}), the data comprised $6{,}335$ cells with RNA-seq feature matrix $\bm X_1\in\mathbb R^{6335\times 2000}$, ATAC-seq feature matrix $\bm X_2\in\mathbb R^{6335\times 5000}$, and transformed surface-protein (ADT) measurement matrix $\widetilde{\bm X}\in\mathbb R^{6335\times 39}$. We treated the CD45RA (discussed in Section~\ref{sec:comp_bio_illustrate}) expression as the continuous response and the remaining 39 proteins as low-dimensional predictors. The prediction experiment used a random split of $4{,}560$ training cells and $1{,}141$ held-out test cells. Data preprocessing and implementation details are provided in Section~\ref{sec:tea_details}.

We trained a linear predictor using the noise-adjusted \fancyname{} embeddings and squared-error loss. Under this specification, the optimization in \eqref{eq:opt_theta} reduces to ordinary least squares using the denoised rescaled embeddings in \eqref{eq:pred_linear embedding}. We evaluated test predictions using RMSE and compared \fancyname{}-CKA against JAFAR \citep{anceschi2024bayesian}, cooperative regression \citep{ding2022cooperative}, MOFA+ \citep{Argelaguet2020}, and Multigrate \citep{LitinetskayaPoEVAE2022}. The gap statistic selected $K=1$ cluster for the empirical Bayes step; we also evaluated the late-fusion boundary case $K=3$ separately.

As shown in Table~\ref{tab:tea_regression_split3}, the gap statistic selected $K=1$ implementation of \fancyname{} based prediction attained an RMSE of $2.6867$, numerically close to the best observed value of $2.6854$ obtained with $K=3$. Both AMP-based prediction pipelines (special cases of \fancyname{}) attained lower RMSE than the non-AMP benchmarks, supporting the competitiveness of the proposed AMP based prediction framework.

\begin{table}[h]
\centering
\scriptsize
\begin{tabular}{lrrrrrr}
\toprule
 & DAIF-CKA (K=1) & DAIF-CKA (K=3) & JAFAR & Co-op. Reg. & MOFA+ & Multigrate \\
\midrule
$\mathsf{RMSE}_{\test}$ & 2.6867 & \textbf{2.6854} & 2.6999 & 2.7513 & 2.7317 & 3.3368 \\
\bottomrule
\end{tabular}
\caption{CD45RA prediction on TEA-seq. Lower RMSE is better. \fancyname{}-CKA for $K=1$ corresponds to OrchAMP \citep{nandy2024multimodal} and \fancyname{}-CKA for $K=3$ corresponds to EB-PCA \citep{zhong2022empirical}. Gap statistic based number of clusters selection resulted in choosing $K=1$.}
\label{tab:tea_regression_split3}
\end{table}

\subsection{A TCGA-BRCA survival analysis}
\label{sec:tcga_brca}

We analyzed the trimodal dataset constructed with the measurements of bulk RNA expression level, DNA methylation intensity, and copy-number variation (CNV) scoped to primary tumors in the TCGA-BRCA cohort \citep{tcga}. After aligning patients across modalities and preprocessing, the analyzed cohort comprised 769 patients with 101 overall survival events. The preprocessed data consisted of two high-dimensional feature matrices, $\bm X_1\in\mathbb R^{769\times 2000}$ for RNA-seq and $\bm X_2\in\mathbb R^{769\times 5000}$ for methylation, together with a five-dimensional CNV representation $\widetilde{\bm X}\in\mathbb R^{769\times 5}$. Detailed cohort-construction and preprocessing procedures are provided in Section~\ref{sec:prproc_tcga}. We used a stratified split comprising 615 training patients (81 events) and 154 test patients (20 events).

We applied Algorithm~\ref{alg:daif_compact}, treating RNA-seq and methylation as high-dimensional modalities and the CNV representation as a low-dimensional modality to obtain features for predicting overall survival. The survival risk was modeled using a linear Cox proportional hazards model and the predictor was trained on the noise-adjusted \fancyname{} embeddings by minimizing the posterior expectation of the negative Cox partial log-likelihood. We evaluated the resulting risk scores using the test-set concordance index and compared \fancyname{}-CKA with $K=1,2,3$ against MOFA+ and Multigrate. The complete Cox objective, test-time risk-score construction, tuning parameters, and baseline implementations are given in Section~\ref{sec:prproc_tcga}.

\begin{table}[t]
\centering
\scriptsize
\begin{tabular}{lrrrrr}
\toprule
 & \fancyname{}-CKA (K=1) & \fancyname{}-CKA (K=2) & \fancyname{}-CKA (K=3) & MOFA+ & Multigrate \\
\midrule
Train C-index & 0.7389 & 0.7397 & 0.7397 & 0.7458 & 0.7585 \\
Test C-index & 0.7367 & 0.7359 & \textbf{0.7375} & 0.7335 & 0.6313 \\
\bottomrule
\end{tabular}
\caption{TCGA-BRCA overall survival prediction (test C-index). Higher is better; bold indicates the best test C-index. \fancyname{}-CKA for $K=1$ correspond to OrchAMP \citep{nandy2024multimodal} and \fancyname{}-CKA for $K=3$ corresponds to EB-PCA \citep{zhong2022empirical}. \fancyname{}-CKA for $K=2$ uses the gap statistic to select the number of clusters, which resulted in choosing $K=2$.}
\label{tab:tcga_brca_survival}
\end{table}

The results are summarized in Table~\ref{tab:tcga_brca_survival}. The best test C-index ($0.738$) is attained by \fancyname{} with $K=3$. Although the gap statistic selected $K=2$, which also attained a slightly higher training C-index than $K=1$, the test C-index was higher for $K=1$. We believe that this discrepancy is primarily due to the limited training sample size, which may prevent the trained model from generalizing reliably to the test data. However, all AMP-based pipelines outperform MOFA+ and Multigrate, showing that \fancyname{} extends successfully from linear regression to survival prediction. While MOFA+ remains competitive (test C-index $0.734$), Multigrate exhibits substantial overfitting in this train-test split: despite achieving the highest training C-index ($0.759$), its test C-index drops to $0.631$, highlighting the instability of deep generative approaches in this small-sample, high-censoring regime. 

\section{Theoretical Properties of DAIF}
\label{sec:theoretical}
In this section, we present the theoretical results underlying the DAIF algorithm. Our principal technical contribution is Theorem~\ref{thm:consistency_prior_main}, which shows that the cluster based prior estimation procedure in Algorithm~\ref{alg:daif_compact} yields an estimator $\wh{\mu}$ that converges weakly to the true prior $\mu$.

We adopt assumptions similar to Assumption~5.1 of \cite{nandy2024multimodal} on $\mu$ and $\{\nu_k:k \in [m]\}$ (cf.  Assumption~\ref{assu:init} of \cite{nandy2026daif_supp}).
\begin{thm}
\label{thm:consistency_prior_main}
Assume that the true prior $\mu$ in \eqref{eq:prior_specification}
satisfies the cluster-wise product structure and $\mathcal C_1,\ldots,\mathcal C_K$ denote the true modality clusters. Consider the estimated clusters
$\widehat{\mathcal C}_1,\ldots,\widehat{\mathcal C}_K$ obtained by
hierarchical clustering with the dissimilarity matrix $\mathfrak N$ defined
in Section~\ref{sec:clus_emp_Bayes}. Assume that there exist deterministic population affinity scores $\{s_{ee'}:e,e'\in\mathcal E\}$ such that 
\begin{align}
\inf_{k\in[K]}\inf_{\substack{e,e'\in\mathcal C_k\\e\neq e'}}s_{ee'}\geq\Delta_{\mathrm{sep}},\quad \mbox{and} \quad
\sup_{\substack{k\neq k'\\e\in\mathcal C_k,\,
e'\in\mathcal C_{k'}}}s_{ee'}=0,
\label{eq:cluster_separation_condition}
\end{align}
and that the empirical affinities used in the construction of $\mathfrak N$ in Section~\ref{sec:clus_emp_Bayes} satisfy $\max_{e,e'\in\mathcal E}\left|\widehat s_{ee'}-s_{ee'}\right|\overset{\mathrm{a.s.}}{\longrightarrow}0.$
Then the maximum likelihood estimator $\widehat\mu$ obtained from
\eqref{eq:lik_cluster} is weakly consistent; that is, $\widehat\mu \overset{w}{\rightarrow} \mu,$ almost surely, as $n \to \infty$. Furthermore, for all $h \in [m]$, the estimated priors $\wh \nu_h\overset{w}{\longrightarrow} \nu_h$ as $p_h \rightarrow \infty$.
\end{thm}

Similar consistency results for empirical Bayes were established in Lemma~B.2 of \cite{zhong2022empirical} for the special case $K=m,\ \wt m=0$, and in Lemma~5.2 of \cite{nandy2024multimodal} for the single-cluster setting $K=1$. However, our setting is more challenging since, the latent clusters in the prior factorization must be learned from the same data used to estimate the component distributions. The foregoing theorem shows that this double use of data does not introduce any additional first-order asymptotic error in the estimated prior and the resulting procedure retains asymptotic behavior of the oracle estimator with access to the true $\mu$. 

Next, we formalize the normal approximation to the rescaled training embeddings $\{\wh{\bm U}^\rsc_h:h \in [m]\}$ and the OLS projected test embeddings $\{\wh{U}^{\test,\ols}_h:h \in [m]\}$ defined in Section~\ref{sec:ols_pretraining}.

\begin{thm}
\label{thm:ols_estimator}
Consider the rescaled training features $\{\wh{\bm U}^\rsc_h:h \in [m]\}$ defined in Section~\ref{sec:ols_pretraining}. Assume that $\bm\Sigma_{T,h}^{L}\succ0$ and $\bm D_h\succ0$ for all $h \in [m]$. Consider any pseudo-Lipschitz function $\psi:\R^r \to \R$, and $Z_{T,h}^{\rsc}\overset{\mathrm{i.i.d}}{\sim} \dnorm_{r_h}(0,\bm I_{r_h})$ for all $h \in [m]$. Then, as $n,p_h \to \infty$, 
\begin{align}
&\lim_{n \rightarrow \infty}\frac{1}{n}\sum_{i=1}^n\psi\left((\wh{\bm U}^\rsc_1)_{i*},\ldots,(\wh{\bm U}^\rsc_m)_{i*},
(\wt{\bm X}_1)_{i*},\ldots,(\wt{\bm X}_{\widetilde m})_{i*}\right) \notag\\
&\overset{a.s.}{=}\mathbb E\bigg[\psi\bigg(U_1+(\bm D^{-1}_1(\bm \Sigma^L_{T,1})^{-1}\bm D^{-1}_1)^{1/2}Z_{T,1}^{\rsc},\ldots,\\
&~~~~~~~~~~~~~~~~~~~~~~~~~~~U_m+(\bm D^{-1}_m(\bm \Sigma^L_{T,m})^{-1}\bm D^{-1}_m)^{1/2}Z_{T,m}^{\rsc},\widetilde X_1,\ldots,\widetilde X_{\widetilde m}\bigg)\bigg],
\label{eq:training_debiased_joint_limit}
\end{align}
where, $(U_1,\ldots,U_m,\wt U_1\ldots,\wt U_{\wt m}) \sim \mu$ and $\wt X_\ell\overset{d}{=}\bm L_\ell \wt U_\ell+\wt Z_\ell$, where $\wt Z_\ell \overset{\mathrm{i.i.d}}{\sim} \dnorm_{\wt r_\ell}(0,\bm I_{\wt r_{\ell}})$ for $\ell \in [\wt m]$.
Furthermore, the test-time OLS embeddings $\{\wh U^{\test,\ols}_h:h \in [m]\}$ satisfy $\wh U^{\test,\ols}_h \xrightarrow{d} U^\test_h+(\bm D^{-1}_h(\bm\Sigma^L_{T,h})^{-1}\bm D^{-1}_h)^{1/2}Z^{\test,\ols}_h,$ where $Z^{\test,\ols}_h \overset{\mathrm{i.i.d}}{\sim} \dnorm_{r_h}(0,\bm I_{r_h})$, for all $h \in [m]$.
\end{thm}
The theorem above justifies the use of the rescaled AMP embeddings $\wh{\bm U}^{\rsc}_h$ in training the predictor in \eqref{eq:def_frak_l} and plugging in $\wh U^{\test,\ols}_h$ in the fitted predictor to obtain the test predictions.

\section{Discussion}
We introduced DAIF, a data-adaptive intermediate fusion framework that learns the fusion granularity from the data rather than relying on a fixed architecture. Our approach combines nonparametric dependence measures, hierarchical clustering, and random matrix theory to estimate empirical Bayes priors that facilitate information integration across modalities. We incorporate these priors into an approximate message passing (AMP) framework to obtain low-dimensional representations which served as features for downstream supervised learning. We derive rigorous theoretical guarantees for our method. In our experiments, our method consistently outperforms popular benchmarks on data integration on synthetic data generated according to a generative model aligning to the hypothesis motivating intermediate fusion. 

We collect here some directions for future inquiry. Our current approach requires strong signal strength in each high-dimensional modality for consistent clustering and informative initialization of the AMP iterations. It remains to be explored whether efficient spectral algorithms as in \citet{du2026optimal} and \cite{yang2026sharp} can relax such strong signal strength assumptions. 

We have used Gaussian mixture models for prior estimation. One can analogously study the performance of non-parametric maximum likelihood-based empirical Bayes as in \citet{zhong2022empirical}, and the performance of advanced techniques from machine learning including variational autoencoders \citep{tang2021empirical,cheng2020generalizing}, diffusion-based denoising \citep{meier2026clustering}, or transformer~\citep{teh2025solving} and score-matching~\citep{ghosh2025stein} based empirical Bayes procedures in terms of prediction accuracy.

\subsection*{Reproducibility and use of Generative AI}
The numerical experiments and the data analysis reported in this paper are reproducible under the following GitHub repository: \url{https://github.com/Sagnik-Nandy/Hierarchical-AMP-GMM}.

We have used GPT for polishing the writing in the manuscript and Claude code for coding up the procedure. The authors checked every detail of the generated code before running experiments. No AI was used for any other aspect of the manuscript.  

\subsection*{Acknowledgment}
S.N.~gratefully acknowledges support from JobsOhio STEM Education Faculty Startup Award (GR142894). P.S.~gratefully acknowledges support from the National Science Foundation (DMS CAREER 2440824) and the Office of Naval Research (N00014-26-1-2144). S.S.~gratefully acknowledges support from NSF (DMS CAREER 2239234), ONR (N00014-23-1-2489) and AFOSR (FA9950-23-1-0429). Furthermore, we acknowledge Ohio State University's ASC Unity cluster for providing the computational resources used in this study.

\bibliographystyle{abbrvnat}
\bibliography{amp_fusion}

\appendix

\section{Pre-AMP procedures for DAIF}
\label{sec:daif_pre_amp_subroutines}

This section provides a pseudocode style summary for the two
pre-AMP procedures invoked by Algorithm~\ref{alg:daif_compact} of the main paper. Algorithm~\ref{alg:daif_spectral_init} describes the construction of the spectral
representations and estimation of the nuisance parameters associated with
the asymptotic Gaussian observation channels, while Algorithm~
\ref{alg:daif_clustered_eb} describes the estimation of the modality partition and
the corresponding empirical Bayes priors.

These algorithms are a modular summary of the procedures outlined in Section~\ref{sec:clus_emp_Bayes}. We refer to equations~\eqref{eq:approx_singular_vectors}-\eqref{eq:approx_snr} of the main
text for the formal definitions of the mathematical terms. 

\begin{algorithm}[t]
\caption{\textsc{SpectralInitialize}: Spectral initialization and nuisance-parameter estimation}
\label{alg:daif_spectral_init}
\begin{algorithmic}[1]

\Input High-dimensional modalities
       $\{\bm X_h\}_{h=1}^m$,
       low-dimensional modalities
       $\{\widetilde{\bm X}_\ell\}_{\ell=1}^{\widetilde m}$,
       and ranks $\{r_h\}_{h=1}^m$.

\Output Spectral state $\widehat{\mathscr F}:=\left\{\left(\bm U_h^{\mathrm{pc}},\bm D_h^{\mathrm{pc}},\bm V_h^{\mathrm{pc}},\widehat{\bm D}_h,\widehat{\bm M}_h^L,\widehat{\bm M}_h^R,\widehat{\bm\Sigma}_h^L,\widehat{\bm\Sigma}_h^R\right):h\in[m]\right\}$ and initial
modality representations $\widehat{\mathscr U}^{\mathrm{init}}:=\left\{\widehat{\bm U}_e^{\mathrm{init}}:e\in\mathcal E\right\}$.

\For{$h\in[m]$}
    \State Set $\widebar{\bm X}_h\gets\bm X_h/\sqrt n$.

    \State Compute the rescaled rank-$r_h$ truncated SVD
    $(\bm U_h^{\mathrm{pc}},
      \bm D_h^{\mathrm{pc}},
      \bm V_h^{\mathrm{pc}})$
    of $\widebar{\bm X}_h$, with
    \[
    (\bm U_h^{\mathrm{pc}})^\top\bm U_h^{\mathrm{pc}}
       =n\bm I_{r_h},
    \qquad (\bm V_h^{\mathrm{pc}})^\top\bm V_h^{\mathrm{pc}}=p_h\bm I_{r_h}.
    \]

    \State Compute $\widehat{\bm D}_h$ using
    \eqref{eq:approx_snr}.

    \State Compute
    $\widehat{\bm M}_h^L,\widehat{\bm M}_h^R,
    \widehat{\bm\Sigma}_h^L,\widehat{\bm\Sigma}_h^R$
    using \eqref{eq:nuisance_param_est_1}.

    \State Set
    $\widehat{\bm U}_h^{\mathrm{init}}
       \gets\bm U_h^{\mathrm{pc}}$.
\EndFor

\For{$\ell\in[\widetilde m]$}
    \State Set
    $\widehat{\bm U}_{m+\ell}^{\mathrm{init}}
       \gets\widetilde{\bm X}_\ell$.
\EndFor

\State \Return
       $\widehat{\mathscr F},
        \widehat{\mathscr U}^{\mathrm{init}}$.

\end{algorithmic}
\end{algorithm}

\section{Discussion on measures of dependence}
\label{sec:measures_of_dependence}
In this paper, we use the centered kernel alignment (CKA) from \citet{kornblith2019similarity} as a measure of dependence. This measure is defined as follows: Fix $\varrho^2>0$ and consider two modality latent matrices $\wh{\bm U}^{\mathrm{init}}_e \in \mathbb{R}^{n \times r_e}$ and
$\wh{\bm U}^{\mathrm{init}}_{e'} \in \mathbb{R}^{n \times r_{e'}}$, let $\ksf_e \in \R^{n \times n}$ and $\ksf_{e'} \in \R^{n \times n}$ be defined as
RBF kernel matrices,
\begin{equation}
  (\ksf_\star)_{ij} = \exp\!\left(-\frac{\left\|(\wh{\bm U}^{\mathrm{init}}_\star)_{i*} - (\wh{\bm U}^{\mathrm{init}}_\star)_{j*}\right\|^2}{2\varrho^2}\right), \quad \mbox{where $\star \in \{e,e'\}$.}
\end{equation}
Let $\wt{\ksf}_e = \bm H \ksf_e \bm H$, where
$\bm H = \bm I_n - \frac{1}{n}\mathbbm{1}_n\mathbbm{1}^\top_n$ is the centering matrix. Here $\mathbbm{1}_n \in \R^n$ denotes the all ones vector. The biased HSIC estimator between modalities $e$ and $e'$ is defined below in \eqref{eq:def_hsic_sample}.
\begin{equation}
\label{eq:def_hsic_sample}
  \mathsf{HSIC}_{e,e'}
  \;=\;
  \frac{1}{n^2}\,\mathrm{Tr}\!\left(\wt{\ksf}_e\wt{\ksf}_{e'}\right), \quad \mbox{where for any matrix $\bm A \in \mathbb R^{n \times n}$, $\mathrm{Tr}(A):=\sum_{i=1}^n(\bm A)_{ii}$.}
\end{equation}
Then CKA between modalities indexed by $e$ and $e'$ is defined as:
\begin{equation}
\label{eq:def_cka}
  \mathrm{CKA}_{e,e'}
  \;:=\;
  \frac{\mathsf{HSIC}_{e,e'}}
       {\sqrt{\mathsf{HSIC}_{e,e}\,
              \mathsf{HSIC}_{e',e'}}}.
\end{equation}
This normalization makes CKA bounded between zero and one whenever the denominator is nonzero. Moreover, at the population level, HSIC vanishes if and only if the two random elements are independent (the kernel should be a characteristic kernel e.g. the Gaussian kernel) \citep{hsic_gretton}. Thus, CKA provides a convenient normalized kernel-based measure of dependence between latent modality representations.

Although CKA is mathematically convenient and builds on the widely used HSIC criterion \citep{hsic_gretton}, other dependence measures can also be used as the function $\aleph_n$ in Algorithm~\ref{alg:daif_compact}. One alternative is distance correlation \citep{szekely2007measuring}, which avoids the explicit choice of a kernel bandwidth parameter and is instead defined through pairwise distances. A potential drawback, however, is that pairwise Euclidean distances can be affected by distance concentration in high-dimensional settings and can also be sensitive to the relative scaling of different modalities.

Another possibility is to use nearest-neighbor-based measures of association, such as the one proposed in \citet{Deb2020Measuring}. These measures have appealing population-level characterizations of dependence. For instance, they can distinguish independence from dependence and, at the other extreme, characterize when one random object is a measurable function of another. However, nearest-neighbor constructions may become unstable or statistically inefficient in high-dimensional latent spaces, especially when the effective sample size is limited.

Finally, one may also consider mutual-information-based objectives, such as InfoNCE \citep{oord2018representation}, which are widely used in contrastive learning. Such approaches can be powerful when the goal is to learn representations that maximize shared information across modalities. However, they typically require additional modeling choices, such as the construction of positive and negative pairs, the specification of a critic or scoring function, and often an embedding of different modalities into a common representation space. In settings where the modality-specific latent signatures have different dimensions or different geometric interpretations, such a common embedding may not be natural.

Overall, CKA offers a simple, normalized, and kernel-based measure of dependence that can be directly applied to latent variables of different dimensions. For this reason, we adopt it as our default choice in Algorithm~\ref{alg:daif_compact}. Nevertheless, the proposed framework is modular: CKA can be replaced by other valid dependence measures when they are better suited to the geometry, dimensionality, or scientific goals of a particular application.

\begin{algorithm}[H]
\caption{\textsc{ClusteredEB}: Data-adaptive clustering and empirical Bayes prior estimation}
\label{alg:daif_clustered_eb}
\begin{algorithmic}[1]

\Input Spectral state $\widehat{\mathscr F}$,
       initial representations
       $\widehat{\mathscr U}^{\mathrm{init}}$,
       dependence measure $\aleph_n$,
       cluster-selection rule,
       and prior classes
       $\mathcal P$ and
       $\{\mathcal P_h\}_{h=1}^m$.

\Output Dissimilarity matrix $\mathfrak N$,
        number of selected clusters $\widehat K$,
        clusters
        $\widehat{\mathcal C}_1,\ldots,
         \widehat{\mathcal C}_{\widehat K}$,
        cluster priors
        $\{\widehat\mu_k\}_{k=1}^{\widehat K}$,
        and column priors
        $\{\widehat\nu_h\}_{h=1}^m$.

\State Set $\mathfrak N_{ee}\gets0$ for every $e\in\mathcal E$.

\For{$1\leq e<e'\leq m+\widetilde m$}
    \State Compute
    \[
    \widehat s_{ee'}
    \gets
    \aleph_n\!\left(
       \widehat{\bm U}_e^{\mathrm{init}},
       \widehat{\bm U}_{e'}^{\mathrm{init}}
    \right).
    \]
    \State Set
    $\mathfrak N_{ee'}
       \gets 1-\widehat s_{ee'}$
    and
    $\mathfrak N_{e'e}\gets\mathfrak N_{ee'}$.
\EndFor

\State Construct an average-linkage dendrogram from
       $\mathfrak N$ as described in Section~\ref{sec:hierarchical_clust}.

\If{the number of clusters $K$ is supplied}
    \State Set $\widehat K\gets K$ and cut the dendrogram
           into $K$ clusters.
\Else
    \State Select $\widehat K$ using the gap-statistic
           procedure described in Section~\ref{sec:hierarchical_clust}.
\EndIf

\State Denote the resulting partition by
       $\widehat{\mathcal C}_1,\ldots,
        \widehat{\mathcal C}_{\widehat K}$.

\For{$k\in[\widehat K]$}
    \State Set
    \[
    \widehat{\mathcal C}_{k,\mathrm H}
      \gets \widehat{\mathcal C}_k\cap[m],
    \qquad
    \widehat{\mathcal C}_{k,\mathrm L}
      \gets
      \{\ell\in[\widetilde m]:
          m+\ell\in\widehat{\mathcal C}_k\}.
    \]

    \State Set
    \[
    \widehat r_k
      \gets
      \sum_{h\in\widehat{\mathcal C}_{k,\mathrm H}}r_h
      +
      \sum_{\ell\in\widehat{\mathcal C}_{k,\mathrm L}}
          \widetilde r_\ell .
    \]

    \State Estimate $\widehat\mu_k$ by maximizing the
           cluster likelihood in \eqref{eq:lik_cluster}
           over $\mathcal P(\mathbb R^{\widehat r_k})$.
\EndFor

\For{$h\in[m]$}
    \State Estimate
    \[
    \widehat\nu_h
      \gets
      \argmax_{\mathfrak n_h\in
          \mathcal P_h(\mathbb R^{r_h})}
      \mathrm L_h^r
      (\bm V_h^{\mathrm{pc}};\mathfrak n_h).
    \]
\EndFor

\State Set
       $\widehat\mu
          \gets
          \widehat\mu_1\times\cdots\times
          \widehat\mu_{\widehat K}$.

\State \Return
       $\mathfrak N,\widehat K,
        \{\widehat{\mathcal C}_k\}_{k=1}^{\widehat K},
        \{\widehat\mu_k\}_{k=1}^{\widehat K},
        \{\widehat\nu_h\}_{h=1}^m$.

\end{algorithmic}
\end{algorithm}

\section{Details of the hierarchical clustering}
\label{sec:hierarchical_clust}
Given the dissimilarity matrix $\mathfrak N \in \R^{(m+\wt m) \times (m+\wt m)}$ defined in Algorithm~\ref{alg:daif_clustered_eb}, to compute the clusters we use the in-built hierarchical clustering module in \texttt{scipy} with the \emph{average linkage} criterion (UPGMA) \citep{upgma}:
\begin{equation}
  d(A, B)
  \;=\;
  \frac{1}{|A||B|}
  \sum_{i \in A}\sum_{j \in B} D_{ij}.
\end{equation}
This produces a dendrogram over the $(m+\wt m)$ modalities. If a fixed number of clusters is provided, as in Section~\ref{sec:numerical_exp} then we cut the dendrogram to ensure that many clusters. However, in real data applications (Sections~\ref{sec:tea_seq} and~\ref{sec:tcga_brca}), where the number of clusters is not apparent, we cut the dendrogram at a range of heights according to the number of clusters varying between 1 and $m+\wt m-1$. Then the threshold optimizing the gap statistic \cite{Tibshirani_2001} is chosen and used for the downstream pipeline. The late fusion setting with $m+\wt m$ many clusters is always evaluated separately.

To compute the dissimilarity matrix, we used the centered Kernel discrepancy \citep{kornblith2019similarity} where we use the Gaussian kernel with bandwidth $1$. This default choice is adopted across all numerical simulations and data examples.

\section{Further details about the implementation of the Approximate Message Passing framework}
\label{sec:amp_details}

In this section, we provide additional details regarding the implementation of the Approximate Message Passing (AMP) iterations described in Section~\ref{sec:feature_extract}. In particular, we clarify the rowwise application of the denoiser functions appearing in \eqref{eq:denoise_cols}, \eqref{eq:denoise_rows}, and \eqref{eq:denoise_low_dim}, along with the corresponding Jacobian terms used in the Onsager corrections.

We begin with the denoisers associated with the high-dimensional modalities. For each $h \in [m]$, the matrix-valued denoised embeddings $v_{t,h}(\hG_{t,h};\wh{\nu}_h)\in\R^{p_h\times r_h}$ is obtained by applying $v_{t,h}(\,\cdot\,;\wh{\nu}_h)$ to the rows of $\hG_{t,h}$ as follows:
\begin{align}
(v_{t,h}(\hG_{t,h};\wh{\nu}_h))_{j*}:=v_{t,h}((\hG_{t,h})_{j*};\wh{\nu}_h), \qquad \text{for all } j\in[p_h].
\end{align}

Next, consider the denoisers associated with the latent factors corresponding to the high-dimensional modalities. For each $k\in[K]$ and $h\in[m]$, we define
\[
u_{t,h,k}\Big(\{\hF_{t,e}:e\in\mathcal C_{k,\mathrm H}\},\{\tX_{\ell}:\ell\in\mathcal C_{k,\mathrm L}\};\wh{\mu}_k\Big)\in\R^{n\times r_h}
\]
through rowwise application of the denoiser function in \eqref{eq:denoiser_func_rows}. Specifically,
\begin{align}
\left(u_{t,h,k}\Big(\{\hF_{t,e}:e\in\mathcal C_{k,\mathrm H}\},\{\tX_{\ell}:\ell\in\mathcal C_{k,\mathrm L}\};\wh{\mu}_k\Big)\right)_{i*}:=u_{t,h,k}\Big(\{(\hF_{t,e})_{i*}:e\in\mathcal C_{k,\mathrm H}\},\{(\tX_{\ell})_{i*}:\ell\in\mathcal C_{k,\mathrm L}\};\wh{\mu}_k\Big),
\end{align}
for all $i\in[n]$.

Similarly, for the low-dimensional modalities, we define
\[
\wt u_{t,\ell,k}\Big(\{\hF_{t,e}:e\in\mathcal C_{k,\mathrm H}\},\{\tX_{t,\wt\ell}:\wt\ell\in\mathcal C_{k,\mathrm L}\};\wh{\mu}_k\Big)\in\R^{n\times \wt r_\ell}
\]
through the rowwise rule
\begin{align}
\left(\wt u_{t,\ell,k}\Big(\{\hF_{t,e}:e\in\mathcal C_{k,\mathrm H}\},\{\tX_{t,\wt\ell}:\wt\ell\in\mathcal C_{k,\mathrm L}\};\wh{\mu}_k\Big)\right)_{i*}:=\wt u_{t,\ell,k}\Big(\{(\hF_{t,e})_{i*}:e\in\mathcal C_{k,\mathrm H}\},\{(\tX_{t,\wt\ell})_{i*}:\wt\ell\in\mathcal C_{k,\mathrm L}\};\wh{\mu}_k\Big),
\end{align}
for all $\ell\in[\wt m]$, $i\in[n]$, and $k\in[K]$.

We now describe the Jacobian terms appearing in the Onsager corrections. Recall that these terms account for the correlation introduced across AMP iterations and play a fundamental role in ensuring the asymptotic Gaussianity of the effective noise variables.

We first consider the Jacobian of the denoiser function
\(
v_{t,h}(\cdot;\wh\nu_h):\R^{r_h}\to\R^{r_h},
\)
which we denote by
\(
\diff{v_{t,h}}(\cdot;\wh\nu_h):\R^{r_h}\to\R^{r_h\times r_h}.
\)
For each $h\in[m]$, the corresponding Jacobian average is defined as
\begin{align}
\wh{\bm J}^{R}_{t,h}:=\frac{1}{p_h}\sum_{i=1}^{p_h}\diff{v_{t,h}}\big((\wh{\bm G}_{t,h})_{i*};\wh\nu_h\big).
\end{align}

Next, consider the denoisers $u_{t,h,k}(\cdot)$ for $h \in [m]$ and $k \in [K]$.
For a fixed cluster $k\in[K]$, let $\mathcal I_h$ denote the collection of coordinate indices corresponding to the $h$-th high-dimensional modality among all coordinates associated with the modalities in $\mathcal C_k$. Further, let
\(
\diff{u_{t,h,k}}:\R^{r_k}\to\R^{r_h\times r_k}
\)
denote the Jacobian of the denoiser $u_{t,h,k}(\cdot)$ with respect to its vector input. Then the corresponding Onsager matrix is defined by
\begin{align}
\wh{\bm J}^{L}_{t,h,k}:=\left[\frac{1}{n}\sum_{i=1}^{n}\diff{u_{t,h,k}}\big((\wh{\bm F}_{t,i_{k,1}})_{i*},\ldots,(\wh{\bm F}_{t,i_{k,r_k}})_{i*},(\wt{\bm X}_{t,j_{k,1}})_{i*},\ldots,(\wt{\bm X}_{t,j_{k,\wt r_k}})_{i*};\wh{\mu}_k\big)\right]_{*\mathcal I_h},
\end{align}
where for any $\mathcal F \subseteq \{1,\ldots,n\}$, the submatrix formed by the columns with indices in $\mathcal F$ is denoted by $\bm A_{*,\mathcal F}$. The entire AMP-based feature extraction procedure is summarized in Algorithm~\ref{alg:daif_amp_feature_extract} which accepts the modality matrices $\{\wb{\bm X}_h:h \in [m]\}$, $\{\wt{\bm X}_\ell:\ell \in [\wt m]\}$, the outputs of Algorithm~\ref{alg:daif_spectral_init} and~\ref{alg:daif_clustered_eb}, and returns the denoised AMP features $\wb{\bm U}_T$ and the AMP-state 
\begin{align}
\label{eq:terminal_amp}
\widehat{\mathscr A}_T
:=\left\{\left(\widehat{\bm V}_{T,h},\widehat{\bm F}_{T,h},\widehat{\bm U}_{T,h},\widehat{\bm M}_{T,h}^L,\widehat{\bm S}_{T,h}^L,\widehat{\bm M}_{T+1,h}^R,\widehat{\bm S}_{T+1,h}^R,\widehat{\bm G}_{T+1,h}\right):h\in[m]\right\}\cup
\left\{\widetilde{\bm U}_{T,\ell}:\ell\in[\widetilde m]\right\},\nonumber\\
\end{align}
after $T$ steps of AMP iterations. This AMP-state $\widehat{\mathscr A}_T$ is used for training the supervised prediction head.

\section{Implementation of baseline procedures in Section~\ref{sec:numerical_exp}}
\label{sec:baseline_exp_des}
In Section~\ref{sec:rec_bench}, we compare \fancyname{} with AJIVE \citep{feng2018angle}, MCCA \citep{kettenring1971canonical}, GCCA \citep{carroll1968gca}, and a version of hierarchical PCA. 

We used the built-in implementation of AJIVE, MCCA, and GCCA from the \texttt{multiblock} package in \texttt{Python}. For the experiments in Section~\ref{sec:rec_bench} and~\ref{sec:ne_pred_oth_meth}, in each of these methods, we have constructed an estimate of $\wh{\bm U}_{\mathrm{shared}} \in \R^{2}$ of the shared subspace between the modalities (all of AJIVE, MCCA, and GCCA make the restrictive assumption that the shared signal lies in a common subspace of reduced dimension). The choice of $2$ for the dimension of the shared latent space is dictated by the data generating mechanism. The contributions of individual components $\bm I_h$ have been constructed by the best rank $r_h-2$ approximation of $\bm I^\circ_h:=\bm X_h-\bm J_h$, where $\bm J_h:=\wh{\bm U}_{\mathrm{shared}}\wh{\bm U}^\top_{\mathrm{shared}}\bm X_h$, where $\bm X_h$ is the data matrix. The final estimate $\wh{\bm U}_h$ of $\bm U_h$ was taken as the top $r_h$ left singular vectors of $\bm I_h+\bm J_h$ (multiplied by $\sqrt{n}$) for $h=1,2,3$. 

\begin{algorithm}[H]
\caption{\textsc{AMPFeatureExtract}: Clusterwise AMP feature extraction}
\label{alg:daif_amp_feature_extract}
\small
\begin{algorithmic}[1]
\Input Normalized high-dimensional modalities
       $\{\widebar{\bm X}_h\}_{h=1}^m$,
       low-dimensional modalities
       $\{\widetilde{\bm X}_\ell\}_{\ell=1}^{\widetilde m}$,
       spectral state $\widehat{\mathscr F}$ (from Algorithm~\ref{alg:daif_spectral_init}),
       estimated clusters
       $\{\widehat{\mathcal C}_k\}_{k=1}^{\widehat K}$ and
       estimated priors
       $\{\widehat\mu_k\}_{k=1}^{\widehat K}$ and
       $\{\widehat\nu_h\}_{h=1}^m$ (from Algorithm~\ref{alg:daif_clustered_eb}),
       and terminal AMP index $T$.
\Output Integrated embedding $\widebar{\bm U}_T$
        and terminal AMP state $\widehat{\mathscr A}_T$.     
\State Compute $\widehat{\bm L}_\ell\gets
    \left(\frac{1}{n}\widetilde{\bm X}_\ell^\top
      \widetilde{\bm X}_\ell-\bm I_{\widetilde r_\ell}\right)^{1/2},$ for $\ell \in [\wt m]$.
\State Initialize $\widehat{\bm G}_{0,h}\gets \bm V_h^{\mathrm{pc}},\widehat{\bm U}_{-1,h}\gets\bm U_h^{\mathrm{pc}}(\widehat{\bm\Sigma}_h^L)^{1/2}, \widehat{\bm M}_{0,h}^R \gets\widehat{\bm M}_h^R, \widehat{\bm S}_{0,h}^R
       \gets\widehat{\bm\Sigma}_h^R$ for $h \in [m]$

\For{$t=0$ \textbf{to} $T$}
\For{$h\in[m]$}
\State Apply the column denoiser:
\(\widehat{\bm V}_{t,h}\gets v_{t,h}(\widehat{\bm G}_{t,h};\widehat\nu_h).\)

\State Update the left state-evolution matrices:
\(\widehat{\bm S}_{t,h}^L\gets\frac{1}{n}\widehat{\bm V}_{t,h}^{\top}\widehat{\bm V}_{t,h},\)
and
\(
\widehat{\bm M}_{t,h}^L\gets\widehat{\bm S}_{t,h}^L\widehat{\bm D}_h.\)

\State Perform the Onsager-corrected left power iteration:
\[
\widehat{\bm F}_{t,h}\gets\widebar{\bm X}_h\widehat{\bm V}_{t,h} -
\gamma_h\widehat{\bm U}_{t-1,h}
\left( \widehat{\bm J}_{t,h}^R(\widehat{\bm G}_{t,h};\widehat\nu_h)\right)^\top.
\]
\EndFor

\For{$k\in[\widehat K]$}
\State Form the clusterwise collections
        \(\widehat{\mathcal F}_{t,k}\gets\left\{\widehat{\bm F}_{t,e}:e\in\widehat{\mathcal C}_{k,\mathrm H}\right\},\) and \(\widetilde{\mathcal X}_k\gets\left\{\widetilde{\bm X}_\ell:\ell\in\widehat{\mathcal C}_{k,\mathrm L}\right\}.\)

\For{$h\in\widehat{\mathcal C}_{k,\mathrm H}$}
            \State Apply the clusterwise row denoiser:
            \(
            \widehat{\bm U}_{t,h}
            \gets
            u_{t,h,k}
            \left(
              \widehat{\mathcal F}_{t,k},
              \widetilde{\mathcal X}_k;
              \widehat\mu_k
            \right).
            \)

            \State Update the right state-evolution matrices:
            \(
            \widehat{\bm S}_{t+1,h}^R
            \gets
            \frac{1}{n}
            \widehat{\bm U}_{t,h}^{\top}
            \widehat{\bm U}_{t,h},\)
            and 
            \(
            \widehat{\bm M}_{t+1,h}^R
            \gets
            \widehat{\bm S}_{t+1,h}^R
            \widehat{\bm D}_h.
            \)

            \State Perform the Onsager-corrected right power iteration:
            \[
            \widehat{\bm G}_{t+1,h}
            \gets
            \widebar{\bm X}_h^\top\widehat{\bm U}_{t,h}
            -
            \widehat{\bm V}_{t,h}
            \left[
              \widehat{\bm J}_{t,h,k}^L
              \left(
                \widehat{\mathcal F}_{t,k},
                \widetilde{\mathcal X}_k;
                \widehat\mu_k
              \right)
            \right]^\top.
            \]
        \EndFor

        \For{$\ell\in\widehat{\mathcal C}_{k,\mathrm L}$}
            \State Apply the low-dimensional denoiser:
            \(
            \widetilde{\bm U}_{t,\ell}
            \gets
            \widetilde u_{t,\ell,k}
            \left(
              \widehat{\mathcal F}_{t,k},
              \widetilde{\mathcal X}_k;
              \widehat\mu_k
            \right).
            \)
        \EndFor
    \EndFor
\EndFor

\State Construct
\(
\widebar{\bm U}_T\gets \left[\widehat{\bm U}_{T,1}\ \cdots\
  \widehat{\bm U}_{T,m}\ \;
  \widetilde{\bm U}_{T,1}\ \cdots\
  \widetilde{\bm U}_{T,\widetilde m}\right].
\)

\State Construct $\widehat{\mathscr A}_T$ defined in \eqref{eq:terminal_amp}.

\State \Return $(\widebar{\bm U}_T,\widehat{\mathscr A}_T)$.

\end{algorithmic}
\end{algorithm}

For the hierarchical PCA (HPCA) benchmark, we concatenated all modality matrices into $\bm X_{\mathrm{ct}}:=[\bm X_1 \; \bm X_2 \; \bm X_3] \in \R^{n \times (p_1+p_2+p_3)}$. Then, we computed the rank $2$ truncated SVD to obtain joint score $\wh{\bm U}_{\mathrm{shared}} \in \R^{n \times 2}$ and individual loadings $\wh{\bm L}_h$ for $h=1,2,3$ (the loadings were computed by splitting the global loading matrix into appropriate blocks corresponding to the modality dimensions). The individual component $\wh{\bm I}_h$ was computed as the best rank $r_h-2$ approximation of $\wh{\bm I}^\circ_h:=\bm X_h-\wh{\bm U}_{\mathrm{shared}}\wh{\bm L}^\top_h$. The final embedding $\wh{\bm U}_h$ was taken as the top $r_h$ truncated left singular vector matrix of $\wh{\bm I}_h+\wh{\bm U}_{\mathrm{shared}}\wh{\bm L}^\top_h$. This is in effect a generalization of the stacked SVD procedure from \cite{baharav2025stacked}, accounting for different ranks in different modalities. 

Across all the benchmarking methods, we also obtained the right singular vectors $\wh{\bm V}_h$ for $h=1,\ldots,m$ in the final step of the construction of the embeddings. In Section~\ref{sec:ne_pred_oth_meth}, we have trained the predictor by fitting a linear regression of the training response $y_\train$ using the covariates $\bm X_h\wh{\bm V}_h(\wh{\bm V}^\top_h\wh{\bm V}_h)^{-1}$ for all four competing benchmarks using the estimated loading matrices $\wh{\bm V}_h$ in the final truncated SVD step of the construction for all four benchmarks. The test data was similarly projected as $(\wh{\bm V}^\top_h\wh{\bm V}_h)^{-1}\wh{\bm V}^\top_hX^\test_h$ for $h=1,2,3$ and the trained predictor was applied on the projected scores.

\begin{table}[t]
\centering
\footnotesize
\setlength{\tabcolsep}{5pt}
\renewcommand{\arraystretch}{0.95}
\begin{tabular}{lrrrrr}
\toprule
Method & $n = 3000$ & $n = 3500$ & $n = 4000$ & $n = 4500$ & $n = 5000$ \\
\midrule
\multicolumn{6}{l}{\textbf{Modality 1 $(\mathcal{E}_1)$}} \\
EB-PCA & 0.057(0.001) & 0.058(0.001) & 0.056(0.001) & 0.055(0.001) & 0.054(0.001) \\
DAIF & \textbf{0.012(0.000)} & \textbf{0.012(0.000)} & \textbf{0.011(0.000)} & \textbf{0.011(0.000)} & \textbf{0.011(0.000)} \\
OrchAMP & 0.015(0.001) & 0.014(0.001) & 0.014(0.001) & 0.012(0.001) & 0.013(0.001) \\
\midrule
\multicolumn{6}{l}{\textbf{Modality 2 $(\mathcal{E}_2)$}} \\
EB-PCA & 0.023(0.002) & 0.022(0.002) & 0.022(0.002) & 0.016(0.001) & 0.019(0.002) \\
DAIF & \textbf{0.015(0.000)} & \textbf{0.014(0.000)} & \textbf{0.014(0.000)} & \textbf{0.014(0.000)} & \textbf{0.013(0.000)} \\
OrchAMP & 0.017(0.001) & 0.017(0.001) & 0.017(0.001) & 0.015(0.001) & 0.016(0.001) \\
\midrule
\multicolumn{6}{l}{\textbf{Modality 3 $(\mathcal{E}_3)$}} \\
EB-PCA & \textbf{0.005(0.000)} & \textbf{0.005(0.000)} & \textbf{0.005(0.000)} & \textbf{0.005(0.000)} & \textbf{0.005(0.000)} \\
DAIF & \textbf{0.005(0.000)} & \textbf{0.005(0.000)} & \textbf{0.005(0.000)} & \textbf{0.005(0.000)} & \textbf{0.005(0.000)} \\
OrchAMP & 0.031(0.000) & 0.031(0.000) & 0.032(0.001) & 0.032(0.000) & 0.031(0.000) \\
\bottomrule
\end{tabular}
\caption{Average reconstruction error $\mathcal{E}_k$ per modality across sample sizes (mean with standard error in parentheses). \textbf{EB-PCA}: no fusion (3 clusters); \textbf{DAIF}: intermediate fusion (2 clusters, CKA clustering); \textbf{OrchAMP}: complete fusion (1 cluster). Bold = best mean per column.}
\label{tab:modality_by_n}
\end{table}

\section{Comparison against AMP based integration procedures.}
\label{sec:other_amp}

\subsection{Comparison in terms of latent reconstruction}
\label{sec:eff_sample_size}
In this experiment, we compared the performance of \fancyname{} in reconstructing $\bm U_1,\bm U_2$ and $\bm U_3$ with EB-PCA and OrchAMP as a function of sample size $n$. In that direction, we have varied the sample size $n \in \{3000,3500,4000,4500,5000\}$. In the hierarchical clustering step, we have used CKA with Gaussian kernel to construct the similarity measures, and the AMP was run for $T=10$ iterations. The reconstruction error was measured in terms of $\mathcal E_1,\mathcal E_2$ and $\mathcal E_3$ defined in Section~\ref{sec:rec_bench}.

Each experiment was repeated over $50$ independent runs, and Table~\ref{tab:modality_by_n} reports the averaged errors.
As seen from the table, \fancyname{} substantially improves reconstruction for latent factors across all the modalities.
For modality~3 ($\bm U_3$), which is independent of other modalities, both EB-PCA and \fancyname{} have equal performance. While EB-PCA is an oracle procedure that is aware of this independence, \fancyname{} shows remarkable adaptation by learning such a relation from the observations. The performance degradation in OrchAMP can be explained by the difficulty in estimating high-dimensional priors with finite samples, a classic phenomenon observed in high-dimensional non-parametric procedures. Across all sample sizes and replications, \fancyname{} consistently grouped modalities 1 and 2 in a single cluster and identified modality 3 as a separate cluster.

\begin{table}[b]
\centering
\footnotesize
\setlength{\tabcolsep}{6pt}
\renewcommand{\arraystretch}{1.05}
\begin{tabular}{lrrrrr}
\toprule
Method & $n = 3000$ & $n = 3500$ & $n = 4000$ & $n = 4500$ & $n = 5000$ \\
\midrule
\multicolumn{6}{l}{\textbf{Linear link, linear}} \\
\quad EB-PCA & 0.028(0.002) & 0.027(0.002) & 0.027(0.002) & 0.024(0.001) & 0.024(0.002) \\
\quad DAIF-CKA & \textbf{0.022(0.001)} & \textbf{0.022(0.001)} & \textbf{0.023(0.001)} & \textbf{0.023(0.001)} & \textbf{0.021(0.001)} \\
\quad OrchAMP & 0.035(0.002) & 0.034(0.001) & 0.036(0.002) & 0.033(0.000) & 0.034(0.002) \\
\midrule
\multicolumn{6}{l}{\textbf{Non-linear link, neural net}} \\
\quad EB-PCA & 0.260(0.057) & 0.330(0.063) & 0.258(0.045) & 0.189(0.034) & 0.259(0.058) \\
\quad DAIF-CKA & \textbf{0.207(0.037)} & \textbf{0.251(0.041)} & \textbf{0.233(0.038)} & \textbf{0.185(0.031)} & \textbf{0.214(0.040)} \\
\quad OrchAMP & 0.242(0.038) & 0.270(0.040) & 0.277(0.046) & 0.203(0.030) & 0.268(0.058) \\
\bottomrule
\end{tabular}
\caption{Average prediction error across sample sizes (mean with standard error in parentheses). \textbf{DAIF-CKA}: AMP with CKA-based clustering (2 clusters); \textbf{EB-PCA}: no fusion (3 clusters); \textbf{OrchAMP}: complete fusion (1 cluster). Bold = best mean per column within each block.}
\label{tab:pred_error}
\end{table}

\subsection{Comparison with AMP-Based Prediction Benchmarks}
\label{sec:eff_pred_a}
In this experiment, we evaluated \fancyname{} for downstream supervised learning by comparing it against the performance of similar embeddings derived from EB-PCA or OrchAMP which are special cases of DAIF with pre-fixed $K=3$ and $K=1$. In particular, this experiment was designed to evaluate the effectiveness of intermediate fusion in supervised learning objectives with multimodal features. 

For this experiment, the response $y$ was generated according to the procedure described at the beginning of Section~\ref{sec:numerical_exp} with both linear and single-index links. We evaluated the performance of predicting the response $y$ using the reconstructed embeddings from \fancyname{}, EB-PCA and OrchAMP as a function of training sample size $n$ which we varied in $\{3000,3500,4000,4500,5000\}$.  Test data containing $\lfloor0.1n\rfloor$ samples were separately generated using the same data-generating mechanism as the training data. Observe that the data-generating mechanism in \eqref{eq:multimodal_factor_model} for the training features involves a division by $\sqrt{n_\train}$ in the signal component to ensure the signal-to-noise ratio in the problem is $O(1)$. To ensure alignment of the distributions of the test and training features, we generated the test features using the same scaling by $\sqrt{n_\train}$ instead of $\sqrt{n_\test}$.
First we used 10 AMP iterations for all three procedures, and used CKA as a similarity measure for the hierarchical clustering in \fancyname{} to estimate the loading matrices $\bm V_1,\bm V_2$ and $\bm V_3$. Then we considered the rescaled training features $\wh{\bm U}^\rsc_1,\wh{\bm U}^\rsc_2$ and $\wh{\bm U}^\rsc_3$ through  constructed as follows:
\begin{align}
\label{eq:def_proj_ols_cons}
\wh{\bm U}^{\rsc}_{h}:=\wh{\bm F}_{t,h}(\wh{\bm \Sigma}^L_{T,h})^{-1}\wh{\bm D}^{-1}_h 
\in \R^{n \times r_h}, \quad \mbox{for $h \in [m]$.}
\end{align}
For the linear-link function, we used a linear predictor fitted by regressing the training response $y$ on the denoised projections $\wb{\bm U}^{\pred}_1,\wb{\bm U}^{\pred}_2$ and $\wb{\bm U}^{\pred}_3$  constructed using $\wh{\bm U}^\rsc_1,\wh{\bm U}^\rsc_2$ and $\wh{\bm U}^\rsc_3$ through \eqref{eq:pred_linear embedding}. 
For the single-index link, we considered a non-linear architecture where for all $k \in [K]$, the clusterwise predictors $g_{\theta,k}$ in \eqref{eq:def_frak_l} are two-layer neural networks with 32 hidden units and ReLU link function and train the predictor by minimizing $\mathfrak L_\theta$ as in \eqref{eq:def_frak_l} with $\mathcal L(y,\mathfrak H(x))=(y-\mathfrak H(x))^2$. The estimated loading matrices constructed from the training features were also used to project the test features as described in Section~\ref{sec:ols_pretraining}. The learned predictors were applied to reconstructed test embeddings to predict the test response. The prediction accuracy was measured using 
\[
\mathcal E_{\mathrm{pred}}:=\frac{1}{n_{\text{test}}} \sum_{i \in \mathcal I_\text{test}}(\wh y_{\test,i}-y_i)^2,
\]
where $\wh y_\test$ is the predicted test response. The results are summarized in Table~\ref{tab:pred_error}. We observe that across all combinations of $n$, link, and predictor, intermediate fusion outperforms the other two fusion strategies in terms of the prediction accuracy.

\begin{table}[tbp]
\centering
\footnotesize
\setlength{\tabcolsep}{6pt}
\renewcommand{\arraystretch}{1.05}
\begin{tabular}{lrrrrr}
\toprule
Method & $n = 3000$ & $n = 3500$ & $n = 4000$ & $n = 4500$ & $n = 5000$ \\
\midrule
\multicolumn{6}{l}{\textbf{Linear link, linear}} \\
\quad EB-PCA & 0.025(0.002) & 0.027(0.002) & 0.026(0.002) & 0.026(0.002) & 0.025(0.001) \\
\quad DAIF-CKA & \textbf{0.024(0.001)} & \textbf{0.024(0.001)} & \textbf{0.024(0.001)} & \textbf{0.023(0.001)} & \textbf{0.022(0.001)} \\
\quad OrchAMP & 0.036(0.002) & 0.034(0.002) & 0.035(0.002) & 0.034(0.001) & 0.033(0.000) \\
\midrule
\multicolumn{6}{l}{\textbf{Non-linear link, neural net}} \\
\quad EB-PCA & 0.263(0.053) & 0.286(0.049) & 0.273(0.051) & 0.211(0.039) & 0.248(0.050) \\
\quad DAIF-CKA & \textbf{0.196(0.032)} & \textbf{0.233(0.037)} & \textbf{0.208(0.033)} & \textbf{0.171(0.028)} & \textbf{0.187(0.032)} \\
\quad OrchAMP & 0.229(0.034) & 0.289(0.051) & 0.244(0.041) & 0.188(0.029) & 0.208(0.033) \\
\bottomrule
\end{tabular}
\caption{Average prediction error across sample sizes (mean with standard error in parentheses) with sample splitting for training feature construction. \textbf{DAIF-CKA}: AMP with CKA-based clustering (2 clusters); \textbf{EB-PCA}: no fusion (3 clusters); \textbf{OrchAMP}: complete fusion (1 cluster). Bold = best mean per column within each block.}
\label{tab:pred_error_sample_split}
\end{table}

\section{A sample splitting based training feature construction method}
\label{sec:sample_splitting}
Observe that the construction of predictor features for the high-dimensional modalities in the training and the test arms of \fancyname{} are different. While the training features $\{\wh{\bm U}^\rsc_{h}:h \in [m]\}$ are constructed by rescaling the Onsager corrected embeddings $\{\wh{\bm F}_{T,h}:h \in [m]\}$ (cf. \eqref{eq:def_proj_ols_cons}), the test features $\{\wh{U}^{\test,\ols}_h:h \in [m]\}$ are obtained by projecting the test observations onto the column space of $\{\wh{\bm R}_{T,h}:h \in [m]\}$ using ordinary least squares. The same OLS projection based approach is not suitable for construction of the training features since the dependence between $\wh{\bm R}_{T,h}$ and $\bm X_h$ ensures that the projected features do not satisfy
\begin{align}
\label{eq:def_proj_ols}
\wh{\bm U}^{\rsc}_{h} \approx \bm U_h+\bm Z_{h,\rsc}(\wh{\bm D}^{-1}_h(\hS^L_{T,h})^{-1}\wh{\bm D}^{-1}_h)^{1/2}
\in \R^{n \times r_h},
\end{align}
where the entries $\bm Z_{h,\rsc} \in \R^{n \times r_h}$ are i.i.d standard Gaussian and independent of $\bm U_1,\ldots,\bm U_m,$ and $\tU_1,\ldots,\tU_{\wt m}$, thereby invalidating the rationale behind the construction in \eqref{eq:frak_pred}.

Nevertheless, one might adopt a sample splitting approach where the predictor is trained on a held out fraction of the training sample that never sees the construction of $\wh{\bm R}_{T,h}$. This removes the dependence between the projector and the data being projected which in turn restores \eqref{eq:def_proj_ols} for the resulting embeddings generated through OLS projection.

While this procedure is perfectly reasonable when the number of training samples is large, it might lead to suboptimal performance when the number of training samples is low but the feature dimension is large. Recall that the validity of AMP state evolution relies on proportional asymptotics which requires the number of samples used in the AMP iterations \eqref{eq:denoise_cols}-\eqref{eq:rough_denoise_rows} to be comparable to the feature dimensions in the high-dimensional modalities. Furthermore, the efficiency of the empirical Bayes prior estimation also requires large number of training samples. Therefore, sample-splitting based feature reconstruction might fail in settings with a small number of training examples.

To investigate the performance of sample splitting based feature reconstruction, we considered a synthetic experiment using the same setting as Section~\ref{sec:eff_pred_a}. The only difference is that we use 70\% of the training data to construct $\wh{\bm V}^{\mathrm{AMP}}_{T,h}$ using AMP iterations \eqref{eq:denoise_cols}-\eqref{eq:rough_denoise_rows} run for $T=10$ iterations. The estimated loading matrices are used to construct 
\[
\wh{\bm R}^{\mathrm{AMP}}_{T,h}:=\frac{1}{0.7\,n_{\train}}\wh{\bm V}^{\mathrm{AMP}}_{T,h}\wh{\bm D}_h.
\]
The training features supplied to the objective in \eqref{eq:opt_theta} are then obtained as
\[
\wh{\bm U}^{\rsc}_{h}:= \wb{\bdX}^{\mathrm{held}}_h \wh{\bm R}^{\mathrm{AMP}}_{T,h} ((\wh{\bm R}^{\mathrm{AMP}}_{T,h})^\top \wh{\bm R}^{\mathrm{AMP}}_{T,h})^{-1}\in \mathbb R^{(0.3 \cdot \,n_{\train})\times r_h},
\]
where $\wb{\bdX}^{\mathrm{held}}_h:=\bdX^{\mathrm{held}}_h/\sqrt{0.7\,n_\train}$ for $h \in [m]$ and $\bdX^{\mathrm{held}}_h$ is the data matrix constructed from the training data by restricting to the subset of held-out observations reserved for predictor training (30\% of the total training data). The rest of the pipeline, proceeds as in Section~\ref{sec:eff_pred_a}. The results of the experiment are outlined in Table~\ref{tab:pred_error_sample_split} and the trend is similar to Table~\ref{tab:pred_error}. This similarity can be attributed to large number of training samples.

Next, we performed the same TCGA-BRCA data analysis described in Section~\ref{sec:tcga_brca} using sample splitting. This experiment provides a stress test for the method since the number of training samples is limited (615 patients). We again held out 30\% of the training samples (184 patients) for estimating $\wh \theta$ by optimizing \eqref{eq:cox_prop_hazard} and used the rest for representation learning (431 patients). The rest of the experiment was the same as Section~\ref{sec:tcga_brca} and the results are outlined in Table~\ref{tab:tcga_brca_survival_sample_split}. We observe that MOFA+ outperforms all AMP-based pipelines exposing the limitation of sample splitting while using AMP and empirical Bayes-based pipeline when the sample size is limited.

\begin{table}[t]
\centering
\scriptsize
\begin{tabular}{lrrrrr}
\toprule
 & \fancyname{}-CKA (K=1) & \fancyname{}-CKA (K=2) & \fancyname{}-CKA (K=3) & MOFA+ & Multigrate \\
\midrule
Train C-index & 0.8201 & 0.8165 & 0.8098 & 0.9079 & 0.9913 \\
Test C-index & 0.6019 & 0.6089 & 0.6003 & \textbf{0.6873} & 0.5157 \\
\bottomrule
\end{tabular}
\caption{TCGA-BRCA overall survival prediction (test C-index) with sample splitting for training feature construction. Higher is better; bold indicates the best test C-index. \fancyname{}-CKA for $K=1$ correspond to OrchAMP \citep{nandy2024multimodal} and \fancyname{}-CKA for $K=3$ corresponds to EB-PCA \citep{zhong2022empirical}. \fancyname{}-CKA for $K=2$ uses the gap statistic to select the number of clusters, which resulted in choosing $K=2$.}
\label{tab:tcga_brca_survival_sample_split}
\end{table}

\section{Details of the TEA-seq data analysis}
\label{sec:tea_details}

\subsection{Data and preprocessing}
The original TEA-seq peripheral-blood mononuclear cell (PBMC) dataset of \citet{Swanson2021} contained $8{,}213$ cells with measurements from three modalities: RNA expression for approximately $36{,}601$ genes, chromatin accessibility at approximately $66{,}828$ ATAC features, and expression measurements for 48 surface proteins. After standard quality-control filtering, $6{,}335$ cells were retained. Each modality was
normalized for sequencing depth and transformed according to $x\mapsto\log(1+x)$ (to reduce skewness in the normalized expressions). We retained the $2{,}000$ most variable genes, the $5{,}000$ most variable ATAC peaks, and 40 highly variable surface-protein markers. The resulting feature matrices were $\bm X_1\in\mathbb R^{6335\times 2000}$ (RNA features), $\bm X_2\in\mathbb R^{6335\times 5000}$ (ATAC-seq features), and $\widetilde{\bm X}_{3}\in\mathbb R^{6335\times 40}$ (ADT features). These matrices correspond, respectively, to the two high-dimensional
modalities in \eqref{eq:multimodal_factor_model} and the low-dimensional modality in \eqref{eq:multimodal_low_dim_model}.

We treated the protein marker CD45RA as the continuous response. Afterremoving CD45RA from the protein matrix, the remaining 39 markers formed the low-dimensional feature matrix $\widetilde{\bm X}\in\mathbb R^{6335\times 39},$ corresponding to the protein modality. After preprocessing the complete dataset, we randomly selected $4{,}560$ cells for training the predictor and $1{,}141$ cells for
held-out testing.

\subsection{Prediction procedure}

We used a linear predictor in $\mathfrak H(\,\cdot\,,\theta)$ and defined the training loss $\mathfrak L_\theta$ through the squared-error loss. It follows from
\eqref{eq:def_frak_l}--\eqref{eq:pred_linear embedding} that, under
this linear specification, optimizing \eqref{eq:opt_theta} is
equivalent to fitting an ordinary least-squares regression of CD45RA
expression on the denoised rescaled training embeddings defined in
\eqref{eq:pred_linear embedding}. This construction accounts for the
Gaussian perturbations in the AMP-derived training features.

For the held-out cells, we first constructed the OLS-projected
embeddings described in Section~\ref{sec:ols_pretraining} and then
applied the corresponding empirical Bayes denoisers. Test predictions
were obtained by applying the fitted regression coefficients to these
denoised test embeddings.

Prediction accuracy was measured using
\[
\mathsf{RMSE}_{\test}=\left\{\frac{1}{n_{\test}}\sum_{i\in\mathcal I_{\test}}\bigl(\widehat y_{\test,i}-y_i\bigr)^2\right\}^{1/2}.
\]

\subsection{Implementation of \fancyname{}}

We applied \fancyname{} using the CKA similarity measure in Algorithm~\ref{alg:daif_compact}. We considered early fusion ($K=1$), late fusion ($K=3$), and data-adaptive fusion, for which the number of clusters was selected by applying the gap statistic \citep{Tibshirani_2001} to the CKA-based modality dissimilarity
matrix. Since the gap statistic compared only $K=1$ and $K=2$, the late-fusion boundary case $K=3$ was evaluated separately. The gap statistic selected $K=1$, consistent with strong dependence among the
three modalities.

Across all fusion regimes, we used 10 AMP iterations with PCA ranks $r_{\mathrm{RNA}}=20$ and $r_{\mathrm{ATAC}}=15.$ For the empirical Bayes steps, the clusterwise priors were modeled
using isotropic Gaussian mixture models. The number of mixture components was allowed to vary over $\left\{1,\ldots,\left\lfloor\sqrt[3]{n}\right\rfloor\right\},$ and the model with the smallest BIC was
selected.

\subsection{Implementation of the benchmark procedures}

We compared \fancyname{} with MOFA+ \citep{Argelaguet2020}, Multigrate \citep{LitinetskayaPoEVAE2022}, JAFAR
\citep{anceschi2024bayesian}, and multimodal cooperative regression \citep{ding2022cooperative}.

We used open source \texttt{R} implementations of MOFA+ \citep{Argelaguet2020}, and JAFAR \citep{anceschi2024bayesian}, and a similar \texttt{Python} implementation of Multigrate for our experiments. For MOFA+, we used 30 factors for the joint embeddings from all modalities. 
For JAFAR, we first computed the top 100 left singular vectors of the RNA and ATAC views via PCA, then used JAFAR to jointly infer 30-dimensional per-modality factor loadings and response coefficients from the training data (PCA transformed RNA and ATAC along with 39 protein features) via Gibbs sampling; test-set predictions were obtained by analytically projecting held-out cells onto the posterior factor loadings at each MCMC draw and averaging the resulting linear predictions across draws.
Multigrate \citep{LitinetskayaPoEVAE2022} is a one-layer variational autoencoder-based method where we used 70 dimensions for the joint encoder encoding all modalities. The resulting embeddings were used to train downstream predictors of CD45RA using linear regression.

\section{Details of the TCGA-BRCA survival analysis}
\label{sec:prproc_tcga}

\subsection{Data and preprocessing}
The TCGA-BRCA data was obtained through the UCSC Xena platform \citep{Goldman2020}. The original data comprised three molecular modalities: bulk RNA-seq gene expression, containing RSEM $\log_2$-normalized measurements for approximately $20{,}530$ genes; copy-number variation (CNV), containing GISTIC2 continuous log-ratio scores for approximately $24{,}776$ genes; and DNA methylation, measured at approximately $485{,}577$ CpG sites using the Illumina 450K array. Overall survival information consisted of event indicators and observed follow-up times in days.

The analysis was restricted to primary-tumor samples and we intersected patients across the three modalities, to get a matched cohort of 772 patients. After excluding patients with missing or non-positive follow-up times, the final cohort analyzed in Section~\ref{sec:tcga_brca} comprised 769 patients, including 101 overall survival events.

Each modality was preprocessed separately. For RNA-seq, genes with more than $20\%$ missing values were discarded, the remaining missing values were imputed using gene-wise medians, and the top $2{,}000$ highly variable genes were retained and z-scored yielding $\bm X_1\in\mathbb R^{769\times 2000}.$ For DNA methylation, CpG sites with more than $20\%$ missingness were discarded, and the remaining missing values were imputed using site-wise medians. The top $5{,}000$ highly variable sites were retained for the final analysis. The corresponding beta values were transformed to M-values according to $M=\log_2\left(\frac{\beta}{1-\beta}\right),$ clipped to $[-10,10]$, and z-scored, yielding $\bm X_2\in\mathbb R^{769\times 5000}.$

The CNV data underwent a two-stage dimension-reduction procedure to construct a low-dimensional modality compatible with \eqref{eq:multimodal_low_dim_model}. In the first stage, features with
more than $20\%$ missingness were discarded, the remaining missing values were imputed by zero, corresponding to diploid copy number, and the $1{,}000$ genes with the largest marginal variances were retained
and z-scored. In the second stage, the five genes with the largest marginal variances were selected from this subset. For each retained feature, the noise variance $\tau_j^2$ was estimated by removing the
top three signal principal components and computing the corresponding column-wise residual variance. Each feature was then divided by $\widehat{\tau}_j$, so that its residual noise was approximately
$\dnorm(0,1)$, yielding the standardized matrix $
\widetilde{\bm X}\in\mathbb R^{769\times 5},$ which was treated as the low-dimensional third modality. Although the original CNV matrix was high-dimensional (with count values), only this reduced representation was used because it approximately satisfied the Gaussian-noise assumption.

After preprocessing, we constructed a stratified training--test split. The training set contained 615 patients, including 81 events, and the test set contained 154 patients, including 20 events.

\subsection{Noise-adjusted Cox regression}

We applied Algorithm~\ref{alg:daif_compact} to the RNA-seq and methylation matrices using ranks $r_{\mathrm{RNA}}=10$ and $r_{\mathrm{Meth}}=8$, respectively, while treating $\widetilde{\bm X}$ as a five-dimensional low-dimensional modality. The resulting rescaled training features were $\widehat{\bm U}^{\rsc}_1\in\mathbb R^{615\times 10}$ (for the RNA-seq), $\widehat{\bm U}^{\rsc}_2\in\mathbb R^{615\times 8}$ (for the DNA methylation), and $\widetilde{\bm X}^{\train}\in\mathbb R^{615\times 5}$ (for CNV).

We fitted a linear Cox proportional hazards model using these features, while accounting for the Gaussian perturbations in the features derived from the AMP embeddings (as described in
Theorem~\ref{thm:ols_estimator}). Let $(T_i,\delta_i)$ denote the observed follow-up time
and event indicator for training patient $i\in\mathcal I_{\train}$,and define the risk set $\mathcal R(T_i):=\{j\in\mathcal I_{\train}:T_j\geq T_i\},$ and $D:=\sum_{i\in\mathcal I_{\train}}\delta_i.$ Let $U_i=\left((\bm U_1)_{i*},(\bm U_2)_{i*},\widetilde{\bm U}_{i*}\right)\in\mathbb R^{23}$
denote the latent multimodal representation driving the underlying randomness in the data and let
\[
\widehat U_i=\left((\widehat{\bm U}^{\rsc}_1)_{i*},(\widehat{\bm U}^{\rsc}_2)_{i*},\widetilde{\bm X}^{\train}_{i*}\right)\in\mathbb R^{23}
\]
denote its observed counterpart. We estimated the Cox coefficientvector by minimizing the posterior expectation of the negative partial log-likelihood:
\begin{align}
\label{eq:cox_prop_hazard}
\widehat\theta:=\argmin_{\theta\in\mathbb R^{23}}\mathbb E_{\widehat\mu}\left[\frac{1}{D}\sum_{\substack{i\in\mathcal I_{\train}\\ \delta_i=1}}\left\{\log\left(\sum_{j\in\mathcal R(T_i)}\exp(\theta^\top U_j)\right)-\theta^\top U_i\right\}\,\middle|\,\widehat{\bm U}\right],
\end{align}
where $\widehat{\bm U}=(\widehat U_1^\top,\ldots,\widehat U_{615}^\top)^\top$
and $\widehat\mu$ is the clustered empirical Bayes prior used in the AMP iterations. The expectation in \eqref{eq:cox_prop_hazard} accounts for the Gaussian perturbations in the estimated training embeddings.

For a test patient $i\in\mathcal I_{\test}$, define
\[
\widehat U_i^{\test}=\left((\widehat{\bm U}^{\test,\ols}_1)_{i*},(\widehat{\bm U}^{\test,\ols}_2)_{i*},\widetilde{\bm X}^{\test}_{i*}\right)\in\mathbb R^{23},
\]
where $\widehat{\bm U}^{\test,\ols}_1$ and $\widehat{\bm U}^{\test,\ols}_2$ are the OLS-projected embeddings
defined in Section~\ref{sec:ols_pretraining}. The corresponding latent representation is
\[
U_i^{\test}=\left((\bm U_1^{\test})_{i*},(\bm U_2^{\test})_{i*},\widetilde{\bm U}^{\test}_{i*}\right)\in\mathbb R^{23}.
\]
The predicted log-risk score was
\[
\widehat R_i=\widehat\theta^\top\mathbb E_{\widehat\mu}\left[U_i^{\test}\,\middle|\,\widehat U_i^{\test}\right],
\]
where the conditional expectation was computed based on the distributional properties of the projected test embeddings derived in Theorem~\ref{thm:ols_estimator} which plug-in estimates of the scaling matrices constructed using \fancyname{}. Here, the prior $\wh \mu$ represent the empirical Bayes estimate of the data generating prior $\mu$ computed using the clustered empirical Bayes procedure in Algorithm~\ref{alg:daif_clustered_eb}.

\subsection{Implementation and evaluation}

For each fusion regime, we ran 10 AMP iterations. The empirical Bayes priors were modeled using isotropic Gaussian mixture models, with the number of mixture components ranging over $\{1,\ldots,\lfloor\sqrt[3]{n}\rfloor\}$ and selected by the Bayesian information criterion. We considered early fusion ($K=1$), late
fusion ($K=3$), and data-adaptive intermediate fusion. The gap statistic applied to the CKA-based modality dissimilarities selected $K=2$. Since the gap-statistic procedure did not include the boundary case $K=3$, late fusion was evaluated separately. The cases $K=1$ and $K=3$ correspond to OrchAMP \citep{nandy2024multimodal} and EB-PCA \citep{zhong2022empirical}, respectively.

We compared these AMP-based procedures with MOFA+ \citep{Argelaguet2020} and Multigrate
\citep{LitinetskayaPoEVAE2022}. For each baseline, multimodal embeddings were first estimated from the training data, after which a linear Cox proportional hazards model was fitted by minimizing the usual negative Cox partial log-likelihood, without the conditional expectation correction in \eqref{eq:cox_prop_hazard}. The fitted model was then used to construct risk scores for the test patients.

Prediction performance was evaluated using the concordance index
\begin{align}
\label{eq:def_c_index}
\text{Test C-index}=
\left(\sum_{\substack{i,j\in\mathcal I_{\text{test}}\\T_i<T_j,\;\delta_i=1}}
\mathbf 1\!\left(\wh{\mathrm R}_i>\wh{\mathrm R}_j\right)\right)
\left(\sum_{\substack{i,j\in\mathcal I_{\text{test}}\\T_i<T_j,\;\delta_i=1}}1\right)^{-1},
\end{align}
The concordance index ranges from zero to one, with larger values
indicating better concordance between predicted risk and observed
survival ordering and $0.5$ corresponding to random ordering. The
training C-index was defined analogously.

\section{Theoretical details on clustered empirical Bayes}
\label{sec:clust_emp_bayes}

\subsection{Initialization and nuisance parameter estimation}

The clustered empirical Bayes component of \fancyname{} relies on the asymptotic behavior of the top principal components of the high-dimensional feature matrices. Recall that the initialization step is based on the best rank-$r_h$ approximation defined below in \eqref{eq:low_rank_init} as follows:
\begin{align}
\label{eq:low_rank_init}
\frac{1}{n}\,\upca_h \dpca_h (\vpca_h)^\top,
\end{align}
where $\dpca_h$ contains the top $r_h$ singular values of $\wb{\bm X}_h$, and the singular vectors are rescaled so that $(\upca_h)^\top \upca_h = n\bm I_{r_h}$ and $(\vpca_h)^\top \vpca_h = p_h\bm I_{r_h}$. Throughout, we fix the signs of the empirical singular vectors by imposing
\[
(\upca_h)^\top_{*j}(\bm U_h)_{*j}\ge 0,
\qquad
(\vpca_h)^\top_{*j}(\bm V_h)_{*j}\ge 0,
\]
for all $h \in [m]$ and $j \in [r_h]$.

Our analysis relies on asymptotic characterizations of empirical singular vectors from random matrix theory \citep{benayech_nadakuditi,baik2005phase,paul2007asymptotics}, combined with Gaussian mixture deconvolution results underlying empirical Bayes denoising \citep{kiefer1956consistency,zhong2022empirical,jiang_zhang}. To formalize these limits, 
define pseudo-Lipschitz test functions, following \cite[(5.1)]{nandy2024multimodal}. In other words, $\varphi:\R^{\mathfrak s}\to\R$ is pseudo-Lipschitz is it satisfies
\begin{align}
\label{eq:pseudo_lips}
|\varphi(x)-\varphi(y)|
\le
C(1+\|x\|_2+\|y\|_2)\|x-y\|_2,
\qquad x,y \in \R^{\mathfrak s},
\end{align}
for some constant $C>0$. In the literature, such functions are also referred as pseudo-Lipschitz of order 2.

Let us also adopt the following assumptions which are similar to Assumption~5.1 of \cite{zhong2022empirical} and Assumption~5.1 of \cite{nandy2024multimodal}.
\begin{assumption}[Initialization assumptions]
\label{assu:init}
The following conditions hold.
\begin{enumerate}
  \item $(U_1,\ldots,U_m,\wt U_1,\ldots,\wt U_{\wt m}) \sim \mu$,
    \(
    \mathbb E_\mu[U_hU_h^\top]=\bm I_{r_h},\) 
    \(
    \mathbb E_\mu[\wt U_\ell \wt U_\ell^\top]=\bm I_{\wt r_\ell},
    \)
    for all $h \in [m]$ and $\ell \in [\wt m]$. Furthermore, $V_h \sim \nu_h$ and 
    \(
    \mathbb E_{\nu_h}[V_hV_h^\top]=\bm I_{r_h},\)
    for $h \in [m]$. Finally, the collection $\{\bm U_1,\ldots,\bm U_m,\wt{\bm U}_1,\ldots,\wt{\bdU}_{\wt m}\}$ 
is independent of the loading matrices $\{\bdV_1,\ldots,\bdV_m\}.$
\item For all $h \in [m]$ and $k \in [r_h]$, the signal strengths satisfy
\(
(\bm D_h)_{kk} > \gamma_h^{-1/4}.
\)
\item For the theoretical analysis, we adopt the population-alignment convention that, for every $h\in[m]$ and $j\in[r_h]$, the signs of the empirical singular vectors are chosen such that
\(
(\bm U_h^{\mathrm{pc}})^\top_{* j}(\bm U_h)_{* j} \ge 0
\)
and
\(
(\bm V_h^{\mathrm{pc}})^\top_{*j}(\bm V_h)_{*j}\ge 0.
\)
\end{enumerate}
\end{assumption}

Then we restate Proposition~5.1 of \cite{nandy2024multimodal} characterizing the asymptotic properties of the empirical principal component embeddings.

\begin{prop}[Proposition~5.1 of \cite{nandy2024multimodal}]
\label{prop:singular_vect_approx}
Suppose Assumption~\ref{assu:init} holds. Consider pseudo-Lipschitz functions $\varphi:\R^r \rightarrow \R $ and $\{\vartheta_h:\R^{r_h} \rightarrow \R~~\mbox{for $h \in [m]$}\}$. Then, if $p_h/n =\gamma_h$ for all $h \in [m]$ as $n,p_h \to \infty$, almost surely we have
\begin{align}
&\lim_{n \rightarrow \infty}\frac1n\sum_{i=1}^n\varphi\!\left((\upca_1)_{i*},\ldots,(\upca_m)_{i*},(\wt{\bm X}_1)_{i*},\ldots,(\wt{\bm X}_{\wt m})_{i*}\right)
=\E_{\mu}\!\left[\varphi\!\left(Y^{\mathrm{pc}}_1,\ldots,Y^{\mathrm{pc}}_m,\wt X_1,\ldots,\wt X_{\wt m}\right)\right], \notag\\
&\lim_{n \rightarrow \infty}\frac1{p_h}\sum_{j=1}^{p_h}\vartheta_h\!\left((\vpca_h)_{j*}\right)=\E_{\nu_h}\!\left[\vartheta_h(Y^{\mathrm{pc}}_{R,h})\right],\qquad \mbox{for all $h \in [m]$,}
\end{align}
where $Y^{\mathrm{pc}}_h=\bm M^L_hU_h+(\bm\Sigma^L_h)^{1/2}Z^{\mathrm{pc}}_{L,h} \in \R^{r_h}$, $Y^{\mathrm{pc}}_{R,h}=\bm M^R_hV_h+(\bm\Sigma^R_h)^{1/2}Z^{\mathrm{pc}}_{R,h} \in \R^{r_h}$, and $\wt X_\ell=\bm L_\ell \wt U_\ell+\wt Z_\ell \in \R^{\wt r_\ell}.$
Here $(U_1,\ldots,U_m,\wt U_1,\ldots,\wt U_{\wt m})\sim\mu$, $V_h\sim\nu_h$, and the noise vectors
$\{Z^{\mathrm{pc}}_{L,h}\}_{h\in[m]}$,
$\{Z^{\mathrm{pc}}_{R,h}\}_{h\in[m]}$, and
$\{\wt Z_\ell\}_{\ell\in[\wt m]}$
have independent standard normal entries and are mutually independent of the latent variables. Furthermore, the alignment matrices $\{(\bm M_h^\star,\bm\Sigma_h^\star):h\in[m],\star\in\{L,R\}\}$ are defined as follows:
\begin{align}
\label{eq:define_alignement_matrices_2}
\bm M^{\star}_h
= \mathsf{diag}(m^{\star}_{1,h},\ldots,m^{\star}_{r_h,h}),
\quad
\bm \Sigma^{\star}_h
= \mathsf{diag}(\sigma^{\star}_{1,h},\ldots,\sigma^{\star}_{r_h,h}),
\quad \star \in \{L,R\},
\end{align}
where
\begin{align}
\label{eq:start_pca_state_evol}
\sigma^L_{i,h}&:= \frac{1 + (\bm D_h^2)_{ii}}
{(\bm D_h^2)_{ii}\{1 + \gamma_h (\bm D_h^2)_{ii}\}},&\
\sigma^R_{i,h}&:= \frac{1 + \gamma_h (\bm D_h^2)_{ii}}
{\gamma_h (\bm D_h^2)_{ii}\{1 + (\bm D_h^2)_{ii}\}},\\
m^L_{i,h}&:= \sqrt{1 - \sigma^L_{i,h}},&\ m^R_{i,h}&:= \sqrt{1 - \sigma^R_{i,h}},
\end{align}
for all $h \in [m]$ and $i \in [r_h]$. Moreover, on an appropriate common probability space, there exist mutually independent matrices $\{\bm Z^{\mathrm{pc}}_{L,h}:h\in[m]\}$ with i.i.d. standard Gaussian entries such that, for every $h\in[m]$,
\begin{align}
\frac{1}{n}\left\|\bm U_h^{\mathrm{pc}}-\bm U_h(\bm M_h^{L})^\top-
\bm Z^{\mathrm{pc}}_{L,h}(\bm\Sigma_h^{L})^{1/2}
\right\|_F^2\xrightarrow{\mathrm{a.s.}}0.
\label{eq:pc_frobenius_coupling}
\end{align}
\end{prop}

The matrices $\{(\bm M_h^\star,\bm\Sigma_h^\star):h\in[m],\star\in\{L,R\}\}$ in \eqref{eq:start_pca_state_evol} are approximated by $\{(\wh{\bm M}_h^\star,\wh{\bm\Sigma}_h^\star):h\in[m],\star\in\{L,R\}\}$ (defined in \eqref{eq:nuisance_param_est_1}) by plugging the estimator $\wh{\bm D}_h$ from \eqref{eq:approx_snr} into \eqref{eq:start_pca_state_evol}. For the low-dimensional modalities, we additionally estimate the loading matrices $\{\bm L_\ell:\ell\in[\wt m]\}$ using 
\[
\wh{\bm L}_\ell :=
\left(\frac{1}{N}\tX_\ell^\top \tX_\ell - \bm I_{\wt r_\ell}\right)^{1/2}.
\]
The following lemma establishes consistency of these nuisance parameter estimates.
\begin{lem}
\label{lem:first_order_nuisance}
As $n\to\infty$, the following hold almost surely:
\begin{enumerate}
    \item For all $h\in[m]$, $\wh{\bm D}_h \xrightarrow{a.s} \bm D_h$, as $n \rightarrow \infty$.
    \item For all $\ell\in[\wt m]$, $\wh{\bm L}_\ell \xrightarrow{a.s} \bm L_\ell$, as $n \rightarrow \infty$.
    \item For all $h\in[m]$ and $\star\in\{L,R\}$, $\wh{\bm M}_h^\star \xrightarrow{a.s.} \bm M_h^\star$ and $\wh{\bm \Sigma}_h^\star \xrightarrow{a.s.} \bm \Sigma_h^\star$, as $n \rightarrow \infty$,
\end{enumerate}
where all the convergences hold in Frobenius norm.
\end{lem}
\begin{proof}
    The proof of the above theorem follows using Proposition~\ref{prop:singular_vect_approx} and the techniques adopted to prove Lemma~5.1 of \cite{nandy2024multimodal}.
\end{proof}

\subsection{Clustering of Modalities}
In this section, we analyze the theoretical properties underlying the success of the hierarchical clustering algorithm to recover the modality clusters $\mathcal C_1,\ldots,\mathcal C_K$. In that direction, we assume the following condition on the independence measure, which ensures that the misclassification error asymptotically goes to zero.

\begin{assumption}[Clustering assumptions]
\label{assu:clust}
For every $e,e'\in\mathcal E$, let $s_{ee'}\in[0,1]$ denote a
deterministic population affinity score. We assume that there exists
$\Delta_{\mathrm{sep}}>0$ such that
\begin{align}
\label{eq:popluation_separation}
\inf_{k\in[K]}\inf_{\substack{e,e'\in\mathcal C_k\\e\neq e'}}
s_{ee'}\geq\Delta_{\mathrm{sep}},\quad \mbox{and} \quad \sup_{\substack{k\neq k'\\e\in\mathcal C_k,\,e'\in\mathcal C_{k'}}}s_{ee'}=0.
\end{align}
\end{assumption}

\begin{assumption}[Uniform convergence of empirical affinities]
\label{assu:uniform_affinity}
Let $\widehat s_{ee'}$ denote the empirical affinity used to construct
the dissimilarity matrix $\mathfrak N$ in Algorithm~\ref{alg:daif_clustered_eb}. We assume
that
\begin{align}
\label{eq:uniform_affinity_convergence}
\max_{e,e'\in\mathcal E}\left|\widehat s_{ee'}-s_{ee'}\right|
\overset{\mathrm{a.s.}}{\rightarrow}0, \quad \mbox{where $s_{ee'}$ satisfies \eqref{eq:popluation_separation}.}
\end{align}
\end{assumption}
For Gaussian-kernel CKA used in our implementation, Assumption~\ref{assu:uniform_affinity} is verified in the following lemma using its V-statistic representation.

\begin{lem}
\label{lem:gaussian_example_e_e}
Consider the random vectors $\{Y^{\mathrm{pc}}_h:h \in [m]\}$ and $\{\wt X_\ell:\ell \in [\wt m]\}$ defined in Proposition~\ref{prop:singular_vect_approx}. For a unified notation, let us denote $U^{\mathrm{init}}_e:=Y^{\mathrm{pc}}_e$ if $e \in \mathcal E_{\mathrm H}$, and $U^{\mathrm{init}}_e:=\wt X_{e-m}$ if $e \in \mathcal E_{\mathrm L}$. For each $e,e'\in\mathcal E$, define the empirical affinity
\begin{align}
\label{eq:defn_using_cka}
\widehat s_{ee'}:=\mathrm{CKA}\left(\widehat{\bm U}^{\mathrm{init}}_e,
\widehat{\bm U}^{\mathrm{init}}_{e'}
\right)
\end{align}
and the corresponding population affinity as
\begin{align}
\label{eq:defn_popn_cka}
s_{ee'}:=\frac{\mathsf H_{e,e'}}{\sqrt{\mathsf H_{e,e}\mathsf H_{e',e'}}},
\end{align}
where for independent copies $(\mathrm X_e,\mathrm X_{e'})$ and $(\mathrm Y_{e},\mathrm Y_{e'})$ of $(U^{\mathrm{init}}_e,U^{\mathrm{init}}_{e'})$ we have
\begin{align}
\mathsf H_{e,e'}&:=  \mathbb E_{\mathrm X_e,\mathrm Y_e,\mathrm X_{e'},\mathrm Y_{e'}}\left[\rsf_{\varrho}(\mathrm X_e,\mathrm Y_e)\rsf_{\varrho}(\mathrm X_{e'},\mathrm Y_{e'})\right]+\mathbb E_{\mathrm X_e,\mathrm Y_e}\left[\rsf_{\varrho}(\mathrm X_e,\mathrm Y_e)\right]E_{\mathrm X_{e'},\mathrm Y_{e'}}\left[\rsf_{\varrho}(\mathrm X_{e'},\mathrm Y_{e'})\right]\\
&~~~~-2\,\mathbb E_{\mathrm X_e,X_{e'}}\left[\mathbb E_{\mathrm Y_e}\left[\rsf_{\varrho}(\mathrm X_e,\mathrm Y_e)\right]\times \mathbb E_{\mathrm Y_{e'}}\left[\rsf_{\varrho}(\mathrm X_{e'},\mathrm Y_{e'})\right]\right].
\end{align}
In the foregoing expression, for two vectors $x,y$ of same dimension, $\rsf_{\varrho}(x,y):=\exp(-\|x-y\|^2/(2\varrho^2))$. Then
\[
\widehat s_{ee'}\overset{\mathrm{a.s.}}{\rightarrow}
s_{ee'}.
\]
Since $|\mathcal E|=m+\widetilde m=O(1)$, the above convergence further implies
\[
\max_{e,e'\in\mathcal E}\left|\widehat s_{ee'}-s_{ee'}\right|\overset{\mathrm{a.s.}}{\rightarrow}0.
\]
\end{lem}
Under Assumptions~\ref{assu:clust} and \ref{assu:uniform_affinity}, we have the following theorem about consistency of modality clustering.
\begin{thm}
\label{thm:consistent_cluster}
Consider the clusters $\widehat{\mathcal C}_1,\ldots,\widehat{\mathcal C}_K$
obtained by applying average-linkage hierarchical clustering to the
dissimilarity matrix $\mathfrak N_{ee'}=1-\widehat s_{ee'}$ for $e,e' \in \mathcal E$.
If $K,m,\wt m=O(1)$ and Assumptions~\ref{assu:clust} and \ref{assu:uniform_affinity} hold, then with probability 1
\[
\big\{\widehat{\mathcal C}_1,\ldots,\widehat{\mathcal C}_K\big\}=
\big\{\mathcal C_1,\ldots,\mathcal C_K\big\}
\]
Equivalently, with probability one, there exists $N_0<\infty$ such that, for every $n\geq N_0$, there exists a permutation $\Pi_n:[K] \to [K]$ satisfying
\[
\widehat{\mathcal C}_{\Pi_n(k)}=\mathcal C_k,
\quad \mbox{for all $k\in[K]$.}
\]
where for two sets $\mathrm A,\mathrm B$, we denote $\mathrm A\,\Delta\,\mathrm B:=(\mathrm A \cap \mathrm B^c) \cup (\mathrm A^c \cap \mathrm B)$.
\end{thm}

\begin{proof}
Write
\[
s(A,B)
:=
\frac{1}{|A||B|}
\sum_{e\in A}\sum_{e'\in B}s_{ee'}
\]
for the population average similarity between two nonempty disjoint sets
\(A,B\subseteq \mathcal E\), and define its empirical counterpart
\[
\widehat s(A,B)
:=
\frac{1}{|A||B|}
\sum_{e\in A}\sum_{e'\in B}\widehat s_{ee'} .
\]

By Assumption~\ref{assu:uniform_affinity}, there exists an almost sure event
\(\Omega_0\) of probability one such that, on \(\Omega_0\),
\[
\max_{e,e'}|\widehat s_{ee'}-s_{ee'}|\to 0 .
\]
Then for every \(\omega\in\Omega_0\) and all nonempty disjoint
\(A,B\subseteq\mathcal E\),
\[
|\widehat s(A,B)-s(A,B)|
\le
\max_{e,e'}|\widehat s_{ee'}-s_{ee'}| .
\]
Hence, for all sufficiently large \(n\),
\[
\max_{A,B}
|\widehat s(A,B)-s(A,B)|
\le \frac{\Delta_{\mathrm{sep}}}{3},
\]
where the maximum is over all nonempty disjoint subsets \(A,B\subseteq \mathcal E\).

We shall prove the result by induction on the stage of hierarchical clustering. Observe that at
the initial stage, since the initial clusters are singletons, every empirical cluster (singleton) is contained in the true cluster up to permutation of cluster labels. Therefore for the initial stage the result holds. Now consider any stage of the hierarchical clustering algorithm and assume that every current empirical cluster is contained in one of the true clusters up to permutation of cluster labels. 

Let \(A\) and \(B\) be two current clusters contained in the same true cluster
\(\mathcal C_k\). Then, by Assumption~\ref{assu:clust},
\[
s(A,B)
=
\frac{1}{|A||B|}
\sum_{e\in A}\sum_{e'\in B}s_{ee'}
\ge \Delta_{\mathrm{sep}}.
\]
Therefore, for all sufficiently large \(n\),
\[
\widehat s(A,B)
\ge
s(A,B)-\frac{\Delta_{\mathrm{sep}}}{3}
\ge
\frac{2\Delta_{\mathrm{sep}}}{3}.
\]

On the other hand, if \(A\subseteq\mathcal C_k\) and
\(B\subseteq\mathcal C_{\ell}\) with \(k\neq \ell\), then by Assumption~\ref{assu:clust} $s(A,B)=0$. Hence
\[
\widehat s(A,B)\le \frac{\Delta_{\mathrm{sep}}}{3}.
\]

Thus, for all sufficiently large \(n\), for every pair of disjoint sets $A$ and $B$ within the true clusters the average similarity is strictly larger than the average similarity for every pair of disjoint sets $A$ and $B$ which lies in separate true clusters. Consequently, average-linkage clustering can only merge two current clusters that are contained in the same true cluster (up to permutation of labels).

By induction over the agglomerative steps, no merge ever crosses a true cluster boundary before all true clusters have been formed. Since the algorithm stops when \(K\) clusters remain, and the true partition has exactly
\(K\) clusters, the resulting partition must equal
\(\{\mathcal C_1,\ldots,\mathcal C_K\}\) (up to permutation of labels). This proves the claim on
\(\Omega_0\), and therefore the claim holds true.
\end{proof}

A natural corollary of the above theorem and Lemma~\ref{lem:gaussian_example_e_e} is the following.
\begin{cor}
If $\wh s_{ee'}$ is defined using \eqref{eq:defn_using_cka}, then 
then if $\mu$ satisfies Assumption~\ref{assu:clust} with $s_{ee'}$ defined by \eqref{eq:defn_popn_cka}, with probability one, there exists $N_0<\infty$ such that, for every $n\geq N_0$, there exists a permutation $\Pi_n:[K] \to [K]$ satisfying
\[
\widehat{\mathcal C}_{\Pi_n(k)}=\mathcal C_k,
\quad \mbox{for all $k\in[K]$.}
\]
\end{cor}
This corollary justifies the use of CKA with Gaussian kernel in \fancyname{}.

\subsection{Proof of Lemma~\ref{lem:gaussian_example_e_e}}
First define matrices $\bm U^{\mathrm{init}}_e \in \R^{n \times r_h}$ where the rows $(\bm U^{\mathrm{init}}_e)_{i*}$ are i.i.d. copies of $U^{\mathrm{init}}_e$ defined in Lemma~\ref{lem:gaussian_example_e_e} and consider the surrogate kernel
\begin{equation}
  (\ck_\star)_{ij} = \exp\!\left(-\frac{\left\|(\bm U^{\mathrm{init}}_\star)_{i*} - (\bm U^{\mathrm{init}}_\star)_{j*}\right\|^2}{2\varrho^2}\right), \quad \mbox{where $\star \in \{e,e'\}$,}
\end{equation}
using the oracle random vectors. 
Let $\wt{\ck}_e = \bm H \ck_e \bm H$, where $\bm H
=
\bm I_n-\frac1n\bm 1\bm 1^\top$, and consider the centered HSIC estimator between modalities $e$ and $e'$ defined as
\begin{equation}
  \wt{\mathsf{HSIC}}_{e,e'}
  \;=\;
  \frac{1}{n^2}\,\mathrm{Tr}\!\left(\wt{\ck}_e\wt{\ck}_{e'}\right),
\end{equation}
Based on the aforementioned definition, the population CKA between modalities indexed by $e$ and $e'$ is defined as:
\begin{equation}
\label{eq:def_cka_tild}
  \wt{\mathrm{CKA}}_{e,e'}
  \;:=\;
  \frac{\wt{\mathsf{HSIC}}_{e,e'}}
       {\sqrt{\wt{\mathsf{HSIC}}_{e,e}\,
              \wt{\mathsf{HSIC}}_{e',e'}}}.
\end{equation}
Let us fix arbitrary $e,e' \in [m+\wt m]$. We prove that
\begin{align}
        \wt{\mathrm{CKA}}_{e,e'} \xrightarrow{a.s} \frac{\mathsf H_{e,e'}}{\sqrt{\mathsf H_{e,e}\mathsf H_{e',e'}}}, \quad \mbox{as $n \rightarrow \infty$.}
        \label{eq:mod_CKA}
    \end{align}
Note that it is enough to prove that $\wt{\mathsf{HSIC}}_{e,e'}\xrightarrow{a.s.}\mathsf H_{e,e'}$, as the convergence of the diagonal terms $H_{e,e}$ and $H_{e',e'}$ can be shown similarly. Then the conclusion follows using the continuity of the normalized ratio on the right hand side of \eqref{eq:mod_CKA} and the fact that $\mathsf H_{e,e}, \mathsf H_{e',e'}>0$. Let
$
\left\{\left(U_{i,e}^{\mathrm{init}},U_{i,e'}^{\mathrm{init}}\right)\right\}_{i=1}^n
$
be i.i.d.\ copies of
$
\left(
U_e^{\mathrm{init}},
U_{e'}^{\mathrm{init}}
\right)
$.

Now note that:
\[
\wt{\mathsf{HSIC}}_{e,e'}=\frac1{n^2}\mathrm{Tr}(\bm H\ck_e\bm H\ck_{e'}),\qquad \bm H=\bm I_n-\frac1n\bm 1\bm 1^\top.
\]
Therefore
\[
\bm H\ck_e\bm H=\left(\bm I_n-\frac1n\bm 1\bm 1^\top\right)\ck_e\left(\bm I_n-\frac1n\bm 1\bm 1^\top\right),
\]
and hence
\begin{align*}
\wt{\mathsf{HSIC}}_{e,e'}
&=
\frac1{n^2}
\mathrm{Tr}\!\left[
\left(
\ck_e
-\frac1n\bm 1\bm 1^\top\ck_e
-\frac1n\ck_e\bm 1\bm 1^\top
+\frac1{n^2}\bm 1\bm 1^\top\ck_e\bm 1\bm 1^\top
\right)
\ck_{e'}
\right] \\
&=
\frac1{n^2}\mathrm{Tr}(\ck_e\ck_{e'})
-\frac1{n^3}\mathrm{Tr}(\bm 1\bm 1^\top\ck_e\ck_{e'})
-\frac1{n^3}\mathrm{Tr}(\ck_e\bm 1\bm 1^\top\ck_{e'}) +\frac1{n^4}
\mathrm{Tr}(\bm 1\bm 1^\top\ck_e\bm 1\bm 1^\top\ck_{e'}).
\end{align*}

We evaluate each term separately. Since the kernel matrices are symmetric,
\[
\mathrm{Tr}(\ck_e\ck_{e'})
=
\sum_{i,j=1}^n
(\ck_e)_{ij}(\ck_{e'})_{ij}.
\]
Furthermore,
\begin{align*}
\mathrm{Tr}(\bm 1\bm 1^\top\ck_e\ck_{e'})
&=
\bm 1^\top\ck_e\ck_{e'}\bm 1 =
\sum_{i=1}^n\left(\sum_{j=1}^n(\ck_e)_{ij}\right)\left(
\sum_{k=1}^n(\ck_{e'})_{ik}\right),
\end{align*}
and similarly
\[
\mathrm{Tr}(\ck_e\bm 1\bm 1^\top\ck_{e'})
=\sum_{i=1}^n\left(\sum_{j=1}^n(\ck_e)_{ij}\right)
\left(\sum_{k=1}^n(\ck_{e'})_{ik}\right).
\]
Finally,
\begin{align*}
\mathrm{Tr}(\bm 1\bm 1^\top\ck_e\bm 1\bm 1^\top\ck_{e'})
&=(\bm 1^\top\ck_e\bm 1)(\bm 1^\top\ck_{e'}\bm 1) =
\left(\sum_{i,j=1}^n(\ck_e)_{ij}\right)\left(
\sum_{i,j=1}^n(\ck_{e'})_{ij}\right).
\end{align*}

Substituting the aforementioned identities yields
\begin{align*}
\wt{\mathsf{HSIC}}_{e,e'}
&=\frac1{n^2}\sum_{i,j=1}^n(\ck_e)_{ij}(\ck_{e'})_{ij} 
+\left(\frac1{n^2}\sum_{i,j=1}^n(\ck_e)_{ij}\right)\left(
\frac1{n^2}\sum_{i,j=1}^n(\ck_{e'})_{ij}\right) \\
&\quad-\frac2n\sum_{i=1}^n\left(\frac1n\sum_{j=1}^n(\ck_e)_{ij}\right)\left(\frac1n\sum_{j=1}^n(\ck_{e'})_{ij}\right).
\end{align*}

Therefore,
\[
\wt{\mathsf{HSIC}}_{e,e'}
=
A_n+B_n-2C_n,
\]
where
\[
A_n=\frac1{n^2}\sum_{i,j=1}^n(\ck_e)_{ij}(\ck_{e'})_{ij},
\quad B_n=\left(\frac1{n^2}\sum_{i,j=1}^n(\ck_e)_{ij}\right)
\left(\frac1{n^2}\sum_{i,j=1}^n(\ck_{e'})_{ij}\right),
\]
and
\[
C_n=\frac1n\sum_{i=1}^n\left(\frac1n\sum_{j=1}^n(\ck_e)_{ij}\right)\left(\frac1n\sum_{j=1}^n(\ck_{e'})_{ij}\right).
\]
Since the RBF kernel satisfies $0\leq \rsf_\varrho(x,y)\leq 1$,
all kernels above are bounded. Next, we note that  $A_n$ is a V-statistic where the arguments within the kernel have finite dimensions. We shall use the strong law of large numbers for V-statistics which follows from the SLLN for U-Statistics (cf. page 122 of \cite{Lee1990}) and the relation between U-statistics and V-statistics (cf. page 183 of \cite{Lee1990}). Using the SLLN for V-statistics, we have
\begin{align}
\label{eq:An_as_limit}
A_n \xrightarrow{a.s.}\E_{\mathrm X_{e},\mathrm Y_{e},\mathrm X_{e'},\mathrm Y_{e'}}\left[\rsf_\varrho(\mathrm X_e,\mathrm Y_e)\rsf_\varrho(\mathrm X_{e'},\mathrm Y_{e'})\right], \quad \mbox{as $n \rightarrow \infty$.}
\end{align}
where $(\mathrm X_e,\mathrm X_{e'})$ and $(\mathrm Y_e,\mathrm Y_{e'})$ are independent copies of $(U_e^{\mathrm{init}},U_{e'}^{\mathrm{init}}).$
Similarly, as $n \to \infty$
\[
\frac{1}{n^2}\sum_{i,j=1}^n(\ck_{e})_{ij}
\xrightarrow{a.s.}\E_{\mathrm X_{e},\mathrm Y_{e}}\left[\rsf_\varrho(\mathrm X_e,\mathrm Y_e)\right], \quad \mbox{and} \quad \frac{1}{n^2}\sum_{i,j=1}^n(\ck_{e'})_{ij}\xrightarrow{a.s.}
\E_{\mathrm X_{e'},\mathrm Y_{e'}}\left[\rsf_\varrho(\mathrm X_{e'},\mathrm Y_{e'})\right].
\]
Consequently,
\begin{align}
\label{eq:Bn_as_limit}
B_n \xrightarrow{a.s.}\E\left[\rsf_\varrho(\mathrm X_e,\mathrm Y_e)\right]\E\left[
\rsf_\varrho(\mathrm X_{e'},\mathrm Y_{e'})\right], \quad \mbox{as $n \rightarrow \infty$.}
\end{align}
Next, observe that
\begin{align*}
C_n&=\frac{1}{n^3}\sum_{i,j,k=1}^n\rsf_\varrho\left((\bm U_{e}^{\mathrm{init}})_{i*},(\bm U_{e}^{\mathrm{init}})_{j*}\right)
\rsf_\varrho\left((\bm U_{e'}^{\mathrm{init}})_{i*},(\bm U_{e'}^{\mathrm{init}})_{k*}\right).
\end{align*}
Let $\mathfrak Z_i=\left((\bm U_{e}^{\mathrm{init}})_{i*},(\bm U_{e'}^{\mathrm{init}})_{i*}\right)$ for $i \in [n]$,
and for any \(z_\ell=(x_\ell,x_\ell') \in \R^{\rho_e+\rho_{e'}}\), define the third-order kernel
\[
h(z_1,z_2,z_3)=\rsf_\varrho(x_1,x_2)\rsf_\varrho(x_1',x_3').
\]
Then
\[
C_n=\frac{1}{n^3}\sum_{i,j,k=1}^nh(\mathfrak Z_i,\mathfrak Z_j,\mathfrak Z_k).
\]
Since \(0\leq \rsf_\varrho\leq 1\), the kernel \(h\) is bounded.
Therefore, the strong law for \(V\)-statistics gives
\[
C_n
\xrightarrow{a.s.}
\E\left[h(\mathfrak Z_1,\mathfrak Z_2,\mathfrak Z_3)\right].
\]
Consider independent copies $\{\mathrm X_e,\mathrm Y_e,\mathrm Z_e\}$ of $U_e^{\mathrm{init}}$ and $\{\mathrm X_{e'},\mathrm Y_{e'},\mathrm Z_{e'}\}$ of $U_{e'}^{\mathrm{init}}$. Observe that
\[
\mathfrak Z_1\overset{d}{=}(\mathrm X_e,\mathrm X_{e'}),
\qquad
\mathfrak Z_2\overset{d}{=}(\mathrm Y_e,\mathrm Y_{e'}),
\qquad
\mathfrak Z_3\overset{d}{=}(\mathrm Z_e,\mathrm Z_{e'}),
\]
we have
\begin{align}
\E\left[h(\mathfrak Z_1,\mathfrak Z_2,\mathfrak Z_3)\right]&=\E\left[\rsf_\varrho(\mathrm X_e,\mathrm Y_e)
\rsf_\varrho(\mathrm X_{e'},\mathrm Z_{e'})
\right]\\
&=\E_{\mathrm X_e,\mathrm X_{e'}}\left[\E_{\mathrm Y_e}[\rsf_\varrho(\mathrm X_e,\mathrm Y_e)]
\E_{\mathrm Y_{e'}}[\rsf_\varrho(\mathrm X_{e'},\mathrm Y_{e'})]\right],
\end{align}
where the last line follows since $\mathfrak Z_1,\mathfrak Z_2,\mathfrak Z_3$ are i.i.d copies of $(U_e^{\mathrm{init}},U_{e'}^{\mathrm{init}})$. Therefore
\begin{align}
\label{eq:Cn_as_limit}
C_n \xrightarrow{a.s.}\E_{\mathrm X_e,\mathrm X_{e'}}\left[\E_{\mathrm Y_e}[\rsf_\varrho(\mathrm X_e,\mathrm Y_e)]
\E_{\mathrm Y_{e'}}[\rsf_\varrho(\mathrm X_{e'},\mathrm Y_{e'})]\right].
\end{align}
Combining \eqref{eq:An_as_limit}, \eqref{eq:Bn_as_limit}, and \eqref{eq:Cn_as_limit}, we conclude that
\begin{align}
\label{eq:conv_deterministic_part}
\wt{\mathsf{HSIC}}_{e,e'}
\xrightarrow{a.s.}
\mathsf H_{e,e'}, \quad \mbox{as $n \rightarrow \infty$.}
\end{align}
and implying \eqref{eq:mod_CKA}. Next, we show that 
\[
\wt{\mathrm{CKA}}_{e,e'}-\mathrm{CKA}_{e,e'}\xrightarrow{a.s.}0, \quad \mbox{as $n \rightarrow \infty$.}
\]
Once again it is enough to show that 
\[
\wt{\mathsf{HSIC}}_{e,e'}-\mathsf{HSIC}_{e,e'}\xrightarrow{a.s.}0, \quad \mbox{as $n \rightarrow \infty$.}
\]
since the convergence of the diagonal terms follow similarly. For the RBF kernel
\[
\rsf_\varrho(x,y)
=
\exp\left(
-\frac{\|x-y\|^2}{2\varrho^2}
\right),
\]
the derivatives are bounded, so that by direct computation one can verify that following the Lipschitz bound
\[
|\rsf_\varrho(x,y)-\rsf_\varrho(x',y')|
\leq
\frac{1}{\varrho}
\left(
\|x-x'\|+\|y-y'\|
\right).
\]
holds. Therefore, for $\star \in \{e,e'\}$
\[
|(\ksf_\star)_{ij}-(\ck_\star)_{ij}|
\leq
\frac{1}{\varrho}
\left(
\|(\wh{\bm U}^{\mathrm{init}}_\star)_{i*}
-
(\bm U^{\mathrm{init}}_\star)_{i*}\|
+
\|(\wh{\bm U}^{\mathrm{init}}_\star)_{j*}
-
(\bm U^{\mathrm{init}}_\star)_{j*}\|
\right).
\]
Consequently, for $\star \in \{e,e'\}$
\begin{equation}
\|\ksf_\star-\ck_\star\|_F
\leq
\frac{2\sqrt n}{\varrho}
\|\wh{\bm U}^{\mathrm{init}}_\star-\bm U^{\mathrm{init}}_\star\|_F \label{eq:kernel_pert}.
\end{equation}
Recall the definition of the centered kernel matrix $\wt{\ksf}$ from Section~\ref{sec:measures_of_dependence}. Since $\|\bm H\|_{\mathrm{op}}\leq 1$, we have:
\[
\|\widetilde{\ksf}_\star-\widetilde{\ck}_\star\|_F
\leq
\|\ksf_\star-\ck_\star\|_F.
\]
Hence,
\begin{align*}
\left|
\mathsf{HSIC}_{e,e'}
-
\widetilde{\mathsf{HSIC}}_{e,e'}
\right|
&=
\frac{1}{n^2}
\left|
\tr\left(
\widetilde{\ksf}_e\widetilde{\ksf}_{e'}
-
\widetilde{\ck}_e\widetilde{\ck}_{e'}
\right)
\right|\\
&\leq
\frac{1}{n^2}
\left(
\|\widetilde{\ksf}_e-\widetilde{\ck}_e\|_F
\|\widetilde{\ksf}_{e'}\|_F
+
\|\widetilde{\ck}_e\|_F
\|\widetilde{\ksf}_{e'}-\widetilde{\ck}_{e'}\|_F
\right).
\end{align*}

Since all kernel entries are bounded by $1$, we have $\|\widetilde{\ksf}_{e'}\|_F\leq n$, and $\|\widetilde{\ck}_e\|_F\leq n.$
Therefore,
\[
\left|
\mathsf{HSIC}_{e,e'}
-
\widetilde{\mathsf{HSIC}}_{e,e'}
\right|
\leq
\frac{1}{n}
\left(
\|\ksf_e-\ck_e\|_F
+
\|\ksf_{e'}-\ck_{e'}\|_F
\right).
\]
Using \eqref{eq:kernel_pert}, we obtain
\[
\left|
\mathsf{HSIC}_{e,e'}
-
\widetilde{\mathsf{HSIC}}_{e,e'}
\right|
\leq
\frac{2}{\varrho\sqrt n}
\left(
\|\bm U^{\mathrm{init}}_e-\wh{\bm U}^{\mathrm{init}}_e\|_F
+
\|\bm U^{\mathrm{init}}_{e'}-\wh{\bm U}^{\mathrm{init}}_{e'}\|_F
\right).
\]

By the Frobenius coupling in \eqref{eq:pc_frobenius_coupling} (alternatively, (D.1) of \cite{nandy2024multimodal}), we have
\[
\frac{1}{\sqrt n}
\|\bm U^{\mathrm{init}}_\star-\wh{\bm U}^{\mathrm{init}}_\star\|_F
\xrightarrow{a.s.} 0
\qquad
\text{for } \star\in\{e,e'\},
\]
and hence we have
\[
\left|
\mathsf{HSIC}_{e,e'}
-
\widetilde{\mathsf{HSIC}}_{e,e'}
\right|
\to 0, \quad \mbox{almost surely.}
\]
Using \eqref{eq:conv_deterministic_part}, the lemma follows.

\subsection{Empirical Bayes with clustered priors}
\label{sec:emp_bayes_clust}
Next, we focus on the accuracy of the prior estimation step. For each estimated cluster $\wh C_k$ define 
\[
\wh r_k:=\sum_{h \in \wh C_{k,\mathrm H}}r_h+\sum_{\ell \in \wh C_{k,\mathrm L}}\wt r_\ell.
\]
Consider the functions $\mathrm F_k:\R^{\wh r_k} \rightarrow \R$ as follows.
\begin{align}
    \mathrm F_k(\{x_e:x_e \in \R^{\rho_e},e \in \wh{\mathcal C}_k\};\mathfrak m_k)&:=\bigintsss\Bigg\{\prod_{e \in \wh{\mathcal{C}}_{k,\mathrm H}}|\wh {\bm \Sigma}^L_e|^{-1/2}\cdot\phi_{\rho_e}\left((\wh {\bm \Sigma}^L_e)^{-1/2}(x_e-\wh{\bm M}^L_e u_e)\right)\Bigg\}\\
    & ~~~~~~~~~~~~~~~\times \Bigg\{\prod_{e \in \wh{\mathcal{C}}_{k,\mathrm L}}\phi_{\rho_{e}}\left(x_e-\wh{\bm L}_e u_e\right)\Bigg\}\,\mathrm{d}\mathfrak m_k(\{u_e:e \in \mathcal{\wh C}_k\}),
\end{align}
where for any $d \in \mathbb N$, $\phi_d(\cdot)$ denotes the density of the standard Gaussian distribution in $d$ dimensions.
Consequently, the likelihood considered in \eqref{eq:lik_cluster} is given by
\begin{align}
    \mathrm L_k(\{\wh{\bm U}^{\mathrm{init}}_e:e \in \wh{\mathcal C}_k\};\mathfrak m_k ):=\prod_{i=1}^n\mathrm F_k(\{(\wh{\bm U}^{\mathrm{init}}_e)_{i*}:e \in \wh{\mathcal C}_k\};\mathfrak m_k).
\end{align}
Next, consider $\mathrm F^{\mathrm R}_h:\R^{r_h} \rightarrow \R$ as follows.
\begin{align}
\mathrm F^{\mathrm r}_h(v_h;\mathfrak n_h)&:=\bigintsss|\wh{\bm \Sigma}^R_h|^{-1/2} \cdot \phi_{r_{h}}\left((\wh{\bm \Sigma}^R_h)^{-1/2}(v_h-\wh{\bm M}^R_h v_h)\right)\,\mathrm{d}\mathfrak n_h(v_h).
\end{align}
Then the likelihood $\mathrm L^r_h(\vpca_h;\mathfrak n_h)$ is defined as
\[
\mathrm L^r_h(\vpca_h;\mathfrak n_h):=\prod_{j=1}^{p_h}\mathrm F^{\mathrm r}_h((\vpca_h)_{*j};\mathfrak n_h)
\]
Next, using the idea behind the proof of Lemma~\ref{lem:gaussian_example_e_e} along with Theorem~\ref{thm:consistent_cluster} we shall prove Theorem~\ref{thm:consistency_prior_main}.
\begin{proof}[Proof of Theorem~\ref{thm:consistency_prior_main}]
First observe that permutation of cluster labels do not change the measure $\wh \mu_1 \times \cdots \times \wh \mu_K$. Therefore, without loss of generality, we can assume that for all $n \in \mathbb N$, $\Pi_n$ in Theorem~\ref{thm:consistent_cluster} is the identity permutation (or we relabel the estimated clusters to match the true clusters). Let us denote
\[
\mathcal E_{\mathrm{clust}}:=\left\{\exists\,N_0,\,\text{such that for all $n \ge N_0$,}\,\bigcup_{k=1}^K\{\wh{\mathcal C}_k \,\Delta\, \mathcal C_k\} = \emptyset\right\}.
\]
By Assumptions~\ref{assu:clust} and \ref{assu:uniform_affinity}, Theorem~\ref{thm:consistent_cluster} implies that the estimated modality partition agrees almost surely with the true partition for all sufficiently large $n$ (after relabeling). In particular, $\mathbb P(\mathcal E_{\mathrm{clust}})=1$.
Therefore,
\begin{align}
    \mathbb P\left[\wh \mu \xrightarrow{w} \mu\right]=\mathbb P[\wh \mu \xrightarrow{w} \mu\; \mbox{and} \; \mathcal E_{\mathrm{clust}}].
\end{align}
On $\mathcal E_{\mathrm{clust}}$, $\wh{\mathcal C}_k=\mathcal C_k$ for all $k \in [K]$. Furthermore, since the mixing class $\mathcal P_k$ has finite second moments and satisfies the local uniform Lipschitz condition in
Assumption~\ref{assu:reg}, Proposition~\ref{prop:singular_vect_approx} imply that all the conditions of Corollary~B.7 of \cite{zhong2022empirical} are satisfied. Therefore, using Corollary B.7 and Lemma~\ref{lem:first_order_nuisance}, we can conclude that the solutions $\wh \mu_1,\ldots,\wh \mu_K$ obtained from \eqref{eq:lik_cluster} satisfies
\[
\limsup_{n \rightarrow \infty}\frac{1}{n}\sum_{i=1}^n\log\frac{\mathrm F_k(\{(\wh{\bm U}^{\mathrm{init}}_e)_{i*}:e \in \mathcal C_k\};\wh{\mu}_k)}{\mathrm F_k(\{(\wh{\bm U}^{\mathrm{init}}_e)_{i*}:e \in \mathcal C_k\};\mu_k)} \ge 0, \quad \mbox{almost surely.}
\]
Similarly, using Corollary B.7 of \cite{zhong2022empirical} and Lemma~\ref{lem:first_order_nuisance}, we can also conclude that for all $h \in [m]$, the MLE
$\wh \nu_h:=\argmax_{\mathfrak n_h}\mathrm L^r_h(\vpca_h;\mathfrak n_h)$ satisfies
\[
\limsup_{p_h \rightarrow \infty}\frac{1}{p_h}\sum_{j=1}^{p_h}\log\frac{\mathrm F^{\mathrm r}_h((\vpca_h)_{j*};\wh{\nu}_h)}{\mathrm F^{\mathrm r}_h((\vpca_h)_{j*};\nu_h)} \ge 0, \quad \mbox{almost surely.}
\]
Then, using Lemma B.2 of \cite{zhong2022empirical}, it follows that
\[
\wh \mu_k \xrightarrow{w} \mu_k, \quad \mbox{for all $k \in [K]$,} \quad \mbox{and} \quad \wh \nu_h \xrightarrow{w} \nu_h, \quad \mbox{for all $h \in [m]$, almost surely.}
\]
Finally, since $\wh \mu=\wh \mu_1\times \cdots\times \wh \mu_K$ and $\mu=\mu_1 \times \cdots \times \mu_K$, therefore $\wh \mu \xrightarrow{w} \mu,$ almost surely, proving the theorem.
\end{proof}

\subsection{Asymptotic behaviour and optimality of the AMP algorithm}

In this section, we provide theoretical results on the performance of data integration and denoising using the AMP-based pipeline. Building upon the results on the consistency of clustered empirical Bayes in Theorem~\ref{thm:consistency_prior_main}, in Theorem~\ref{thm:w_2_amp} we establish the asymptotic characterization of the iterates
\[
\{\wh{\bm F}_{t,h},\wh{\bm G}_{t,h}: t\geq 0,\ h\in[m]\}.
\]
The analysis leverages the weak convergence of the estimated priors along with the proof techniques developed in \cite{nandy2024multimodal}. We then leverage these asymptotic characterizations in Theorem~\ref{thm:bayes_optimality} to show that the AMP iterates achieve Bayes-optimal estimation of the signals. 

We assume the following regularity conditions on the priors $\mu$ and $\{\nu_h:h \in [m]\}$.
\begin{assumption}[Regularity assumptions]
\label{assu:reg}
Let $\mu \in \mathcal{P}$ and $\nu_h \in \mathcal{P}_h$ for all $h \in [m]$. We assume that these parameter spaces satisfy the following regularity properties. 
\begin{enumerate}
  \item For any collection of positive definite matrices 
$\{A_h,B_h\in\R^{r_h\times r_h}:h\in[m]\}$ and 
$\{A_\ell\in\R^{\wt r_\ell\times \wt r_\ell}:\ell\in[\wt m]\}$ and 
$\mu\in\mathcal P$, there exists an open neighborhood 
$O_\mu$ of $\mu$ in the topology of weak convergence, such that for $(U_1,\ldots,U_m,\wt U_1,\ldots,\wt U_{\wt m})\sim\mu$, $Y_h^G\sim\dnorm_{r_h}(A_hU_h,B_h),$ and $\wt Y_\ell^G\sim\dnorm_{\wt r_\ell}(A_\ell\wt U_\ell,\bm I_{\wt r_\ell}),$
the functions 
\begin{align}
u_h^G(x_1,\ldots,x_m,\wt x_1,\ldots,\wt x_{\wt m};\mu)
&=\E_\mu\big[U_h\mid Y_1^G=x_1,\ldots,Y_m^G=x_m,\wt Y_1^G=\wt x_1,\ldots,\wt Y_{\wt m}^G=\wt x_{\wt m}\big],\\
\wt u_\ell^G(x_1,\ldots,x_m,\wt x_1,\ldots,\wt x_{\wt m};\mu)
&=\E_\mu\big[\wt U_\ell\mid Y_1^G=x_1,\ldots,Y_m^G=x_m,\wt Y_1^G=\wt x_1,\ldots,\wt Y_{\wt m}^G=\wt x_{\wt m}\big],
\end{align}
are uniformly Lipschitz over the neighborhood $O_\mu$ for all $h\in[m]$ and $\ell\in[\wt m]$.
\item For any positive definite matrices 
$A_h^R,B_h^R\in\R^{r_h\times r_h}$ and 
$\nu_h\in\mathcal P_{\nu_h}$, $h\in[m]$, there exists an open neighborhood 
$O_{\nu_h}$ of $\nu_h$ in the topology of weak convergence, such that for $V_h\sim\nu_h$,  $Y_{h,R}^G\sim\dnorm_{r_h}(A_h^RV_h,B_h^R),$
the function 
$
v_h^G(x;\nu_h)=\E_{\nu_h}[V_h\mid Y_{h,R}^G=x]
$
is uniformly Lipschitz over the neighborhood $O_{\nu_h}$ for all $h\in[m]$.
\end{enumerate}
\end{assumption}

These assumptions on the latent generating priors are mild and are satisfied by many standard priors used in empirical Bayes modeling. We provide some examples of common classes of prior in Appendix~\ref{sec:examples_regular_priors} satisfying the foregoing assumption. These regularity assumptions ensure that the estimated denoisers and Onsager corrections are consistent for the oracle denoisers constructed with the true data-generating priors. This ensures empirical Bayes estimation of the priors does not induce any additional error in the asymptotic limit. Both sets of Assumptions are commonly adopted in the literature (cf. Assumption 5.1~\citep{zhong2022empirical} and Assumption 5.1~\citep{nandy2024multimodal}). 

Henceforth, for any vector $v \in \R^n$, we shall denote the $n\times n$ matrix $vv^\top$ by $v^{\otimes 2}$. Recall the definition of pseudo-Lipschitz functions from (5.1) of \cite{nandy2024multimodal} \footnote{For the sake of completeness, we repeat the definition in \eqref{eq:pseudo_lips}.}. Then under the aforementioned assumptions, we have the following theorem. 

\begin{thm}
    \label{thm:w_2_amp}
Define the alignment matrices $\{(\bm M_{t,h}^\star,\bm\Sigma_{t,h}^\star):h\in[m],\star\in\{L,R\}, t \ge 1\}$ as follows:
\begin{align}
\label{eq:define_alignement_matrices}
\bm \Sigma^{L}_{t,h}
&= \gamma_h\mathbb E_{\nu_h}\left[\mathbb E_{\nu_h}\left[V_h \mid V_{t,h}\right]^{\otimes 2}\right], \quad \bm M^{L}_{t,h}= \bm \Sigma^{L}_{t,h}\bm D_h,\\
\bm \Sigma^{R}_{t+1,h}
&= \mathbb E_{\mu}\left[\mathbb E_{\mu}\left[U_h \mid \{U_{t,h}:h \in [m]\},\{\wt X_\ell:\ell \in [\wt m]\}\right]^{\otimes 2}\right]
\quad \mbox{and} \quad 
\bm M^{R}_{t+1,h}= \bm \Sigma^{R}_{t+1,h}\bm D_h, 
\end{align}
where $U_{t,h}=\bm M^L_{t,h}U_h+(\bm\Sigma^L_{t,h})^{1/2}Z_{t,L,h}$, $V_{t,h}=\bm M^R_{t,h}V_h+(\bm\Sigma^R_{t,h})^{1/2}Z_{t,R,h}$, and $\wt X_\ell=\bm L_\ell \wt U_\ell+\wt Z_\ell.$
Here $(U_1,\ldots,U_m,\wt U_1,\ldots,\wt U_{\wt m})\sim\mu$, $V_h\sim\nu_h$, and the noise vectors
$\{Z_{t,L,h}: h \in [m]\}$,
$\{Z_{t,R,h}: h \in [m]\}$, and
$\{\wt Z_\ell: \ell\in[\wt m]\}$
have independent standard normal entries and are mutually independent of the latent variables for each $t \ge 1$. Then, under Assumption~\ref{assu:reg} and the conditions of Theorem~\ref{thm:consistency_prior_main}, for any pseudo-Lipschitz functions $\varphi:\R^r \rightarrow \R $ and $\{\vartheta_h:\R^{r_h} \rightarrow \R\}$, as $n,p_h \to \infty$, for all $t \ge 1$, we almost surely have
\begin{align}
&\lim_{n \rightarrow \infty}\frac1n\sum_{i=1}^n\varphi\!\left((\wh{\bm F}_{t,1})_{i*},\ldots,(\wh{\bm F}_{t,m})_{i*},(\wt{\bm X}_1)_{i*},\ldots,(\wt{\bm X}_{\wt m})_{i*},(\bm U_1)_{i*},\ldots,(\bm U_m)_{i*},(\wt{\bm U}_1)_{i*},\ldots,(\wt{\bm U}_{\wt m})_{i*}\right)\notag\\
&\hskip 7em =\E_{\mu}\!\left[\varphi\!\left(U_{t,1},\ldots,U_{t,m},\wt X_1,\ldots,\wt X_{\wt m},U_1,\ldots,U_m,\wt U_1,\ldots,\wt U_{\wt m}\right)\right], \notag\\
&\lim_{n \rightarrow \infty}\frac1{p_h}\sum_{j=1}^{p_h}\vartheta_h\!\left((\wh{\bm G}_{t,h})_{j*},(\bm V_h)_{j*}\right)=\E_{\nu_h}\!\left[\vartheta_h(V_{t,h},V_h)\right],\qquad \mbox{for all $h \in [m]$,}
\end{align}
\end{thm}

\begin{proof}
By Theorem~\ref{thm:consistency_prior_main}, we have $\widehat\mu\overset{w}{\to}\mu$ and $\widehat\nu_h\overset{w}{\to}\nu_h$ almost surely for
all $h\in[m]$. Moreover, Lemma~\ref{lem:first_order_nuisance} establishes almost-sure consistency of the estimated nuisance parameters.

Assumption~\ref{assu:reg} ensures that the posterior-mean denoisers
are uniformly Lipschitz over weak neighborhoods of the true priors.
Consequently, Lemmas~G.2 and G.3 of \cite{nandy2024multimodal} imply that, at every fixed iteration, the empirical-Bayes denoisers and their averaged Jacobians converge to their oracle counterparts. The empirical-Bayes AMP iterates are therefore asymptotically equivalent, in normalized Frobenius norm, to the corresponding oracle AMP iterates as described in Section~5.2.2 of \cite{nandy2024multimodal}. Therefore the stated limits follow from Theorem~5.1 of \cite{nandy2024multimodal}.
\end{proof}

The important takeaway from Theorem~\ref{thm:w_2_amp} is that the consistency of $\wh \mu$ and $\wh \nu_h$ are enough to ensure that the double use of data arising from using the same training data for prior estimation and AMP iterations do not induce any additional bias in the asymptotic behavior of the iterates. The asymptotic properties of the empirical Bayes AMP iterates \eqref{eq:denoise_cols}-\eqref{eq:rough_denoise_rows} are same as the iterates constructed with oracle knowledge of the priors. This property is driven by Assumption~\ref{assu:reg}.

Now, proceeding as in the proof of Theorem 5.2 and 5.3 of \cite{nandy2024multimodal} we can show the following theorem.
\begin{thm}
    \label{thm:bayes_optimality}
    Suppose the true priors $\mu$ and $\nu_h$ for all $h \in [m]$ satisfy the fixed point equations in (5.13) of \cite{nandy2024multimodal}. Then, if the conditions of Theorem~\ref{thm:w_2_amp} hold, we have the following.
    \begin{align}
        &\lim_{t \rightarrow \infty}\lim_{n \rightarrow \infty}\frac{1}{np_h}\mathbb E\left[\|\bdU_h\bm D_h\bdV^\top_h-\wh{\bm U}_{t,h}\bm D_h\wh{\bm V}^\top_{t,h}\|^2_F\right]\\
        &=\lim_{n \rightarrow \infty}\frac{1}{np_h}\mathbb E\left[\left\|\bdU_h\bm D_h\bdV^\top_h-\mathbb E\left[\bdU_h\bm D_h\bdV^\top_h \mid \bm X_1,\ldots,\bm X_m,\wt{\bm X}_1,\ldots,\wt{\bm X}_{\wt m}\right]\right\|^2_F\right].
    \end{align}
\end{thm}

\begin{proof}
    The proof of the above theorem follows by using Theorem~\ref{thm:w_2_amp} and the techniques outlined in Theorem~5.3 of \cite{nandy2024multimodal}.
\end{proof}

\section{Examples of priors satisfying the regularity assumption}
\label{sec:examples_regular_priors}
We now give several concrete classes of latent generating priors for which
Assumption~\ref{assu:reg} holds. The argument is stated for a generic
Gaussian observation model, since both parts of Assumption~\ref{assu:reg}
are special cases of this setup. Henceforth, we shall denote the spectral norm of a matrix $\bm A$ by $\|\bm A\|_{\mathrm{op}}$. Furthermore, the class of all $m \times m$ dimensional positive definite matrices will be denoted by $\mathbb S^{m}_+$ and for two matrices $\bm A$ and $\bm B$, $\bm A \preceq \bm B$ shall mean $\bm B-\bm A$ is positive semi-definite. 

Let $U\in\R^r$ have prior distribution $\pi \in \mathcal P$, and suppose that $Y \mid U \sim \dnorm_q(\bm AU,\bm B)$,
where $\bm A\in\R^{q\times r}$ is fixed and $\bm B\in\R^{q\times q}$ is positive
definite. Define the Bayes denoiser $\mathsf m_\pi(y):=\E_\pi[U\mid Y=y].$
Observe that the generic setting described here can be mapped to the settings in either parts of Assumption~\ref{assu:reg} through appropriate choices of $\bm A$ and $\bm B$. The
individual denoisers $u_h^G$ and $\wt u_\ell^G$ are simply coordinate blocks
of $\mathsf m_\pi$ where $\pi=\mu$. Similarly, the right denoiser $v_h^G$ corresponds to the
same construction with $U=V_h$ and $\pi=\nu_h$. Next, we use the following property of $\mathsf m_\pi$.

\begin{lem}
\label{lem:posterior_mean_derivative}
Assume that $\bm B$ is positive definite and that $\E_\pi \|U\|_2^2<\infty.$ 
Then
\begin{align}
\label{eq:nabla_m_pi}
    \nabla_y \mathsf m_\pi(y)
    =
    \Cov_\pi(U\mid Y=y) \bm A^\top \bm B^{-1}.
\end{align}
Consequently,
\[
    \|\nabla_y \mathsf m_\pi(y)\|_{\op}
    \le
    \|\Cov_\pi(U\mid Y=y)\|_{\op}\,
    \|\bm A^\top \bm B^{-1}\|_{\op}.
\]
In particular, if
\[
    \sup_{y\in\R^q}
    \|\Cov_\pi(U\mid Y=y)\|_{\op}
    \le C_\pi, \quad \mbox{uniformly over $\pi \in \mathcal P$,}
\]
then $\mathsf m_\pi$ is uniformly Lipschitz over $\mathcal P$ with Lipschitz constant at most
$C_\pi\|\bm A^\top \bm B^{-1}\|_{\op}$.
\end{lem}
The proof of the lemma is provided in Appendix~\ref{sec:proof_lem_deriv}.
From the foregoing lemma, to prove that $\mathsf m_\pi$ is Lipschitz uniformly over a local
neighborhood of the true prior $\pi$ in the topology of weak convergence, it suffices to control the posterior
covariance uniformly over such a neighborhood. Now, consider the following proposition providing examples of classes of priors $\mathcal P$ satisfying Assumption~\ref{assu:reg}.
\begin{lem}
\label{prop:examples_regular_priors}
Assume that the Gaussian noise covariance satisfies $\bm B\succeq b \bm I_q$ for
some $b>0$. Then Assumption~\ref{assu:reg} holds if the latent prior class $\mathcal P$ (or, equivalently $\mathcal P_h$ for $h \in [m]$) is one of the following.
\begin{enumerate}
    \item The class of priors supported on a common compact set,
    \[
    \mathcal P_{\mathrm{comp}}(M):=\left\{\pi:\supp(\pi)\subseteq\{u\in\mathbb R^r:\|u\|_2\leq M\}\right\},\qquad M<\infty.
    \]
    \item The compact Gaussian parameter class
    \[
    \mathcal P_{\mathrm{Gauss}}:=\left\{\mathcal N_r(m,\bm\Sigma):\|m\|_2\leq M,\ \underline\lambda\bm I_r\preceq\bm\Sigma\preceq\overline\lambda\bm I_r\right\},
    \]
    where $M<\infty$ and $0<\underline\lambda\leq\overline\lambda<\infty$.
    \item The class of uniformly strongly log-concave priors $\mathcal P_{\mathrm{logconv}}$ having
    density satisfying $d\pi(u) \propto \exp\{-V_\pi(u)\}$, where
    \[
    \nabla^2V_\pi(u)\succeq c_0\bm I_r\qquad\text{for every }u\in\mathbb R^r,\,\mbox{and}\;\pi \in \mathcal P_{\mathrm{logconv}},
    \]
    with a common constant $c_0>0$.
    \item The class of finite homoscedastic Gaussian mixtures
    \[
    \pi=\sum_{a=1}^{J}w_a\mathcal N_r(m_a,\bm\Sigma),
    \]
    where $J\leq J_{\max}$, $w_a\geq w_{\min}>0$, $\sum_{a=1}^{J}w_a=1$,
    $\|m_a\|_2\leq M$, and
    \[
    \underline\lambda\bm I_r\preceq\bm\Sigma\preceq\overline\lambda\bm I_r.
    \]
\end{enumerate}
In each case, the posterior-mean denoiser $m_\pi(y):=\mathbb E_\pi[U\mid Y=y]$ is globally Lipschitz in $y$, uniformly over $\pi\in\mathcal P$. Consequently, it is uniformly Lipschitz over a relative weak neighborhood of every $\pi_0\in\mathcal P$.
\end{lem}
In particular, the assumption holds for the discrete priors used by the
non-parametric maximum likelihood estimator, as well as for Gaussian,
compactly supported, strongly log-concave, and regular finite-mixture priors. In particular, the Gaussian mixture classes used to model the latent priors in the experiments of Sections~\ref{sec:numerical_exp} and~\ref{sec:real_data} satisfy Assumption~\ref{assu:reg}.

\subsection{Proof of Lemma~\ref{lem:posterior_mean_derivative}}
\label{sec:proof_lem_deriv}
For $y\in\R^q$ and $u\in\R^r$, define the Gaussian likelihood kernel
\[
L_y(u):=\exp\left\{-\frac12 (y-\bm Au)^\top \bm B^{-1}(y-\bm Au)\right\}.
\]
The normalizing constant of the posterior distribution is
\[
D(y):=\int L_y(u)\,d\pi(u),
\]
and the numerator of the posterior mean is
\[
N(y):=\int u L_y(u)\,d\pi(u).
\]
Thus $\mathsf m_\pi(y)=N(y)/D(y).$
First observe that
\[
\nabla_y \log L_y(u)=-\bm B^{-1}(y-\bm Au)=\bm B^{-1}(\bm Au-y).
\]
Therefore,
\[
\nabla_y L_y(u)=L_y(u)\bm B^{-1}(\bm Au-y).
\]
Since $N(y)$ is vector-valued, its derivative with respect to $y$ is an $r\times q$ matrix. Differentiating $N(y)$ gives
\[
\nabla_y N(y)=\int u(\bm Au-y)^\top \bm B^{-1} L_y(u)\,d\pi(u).
\]
Similarly,
\[
\nabla_y D(y)=\int \bm B^{-1}(\bm Au-y)L_y(u)\,d\pi(u).
\]
Equivalently,
\[
\{\nabla_y D(y)\}^\top=\int (\bm Au-y)^\top \bm B^{-1}L_y(u)\,d\pi(u).
\]
By the quotient rule,
\[
\nabla_y \mathsf m_\pi(y) =\frac{\nabla_y N(y)}{D(y)}
    -\frac{N(y)}{D(y)}\frac{\{\nabla_y D(y)\}^\top}{D(y)}.
\]
Using the posterior distribution induced by $\pi$ and the observation $Y=y$,
this becomes
\[
\begin{aligned}
\nabla_y \mathsf m_\pi(y)&=\E_\pi\left[U(\bm AU-y)^\top\mid Y=y\right]\bm B^{-1}-\E_\pi[U\mid Y=y]\,
    \E_\pi\left[(\bm AU-y)^\top\mid Y=y\right]\bm B^{-1}.
\end{aligned}
\]
Expanding the two terms,
\[
\E_\pi\left[U(\bm AU-y)^\top\mid Y=y\right]=\E_\pi[UU^\top\mid Y=y]\bm A^\top-\E_\pi[U\mid Y=y]y^\top,
\]
and
\[
\E_\pi[U\mid Y=y]\;\E_\pi\left[(\bm AU-y)^\top\mid Y=y\right]=\E_\pi[U\mid Y=y]\E_\pi[U\mid Y=y]^\top \bm A^\top-\E_\pi[U\mid Y=y]y^\top.
\]
After some algebraic manipulation, we get
\[
\begin{aligned}
\nabla_y \mathsf m_\pi(y)&=\left\{\E_\pi[UU^\top\mid Y=y]-\E_\pi[U\mid Y=y]\E_\pi[U\mid Y=y]^\top\right\}
\bm A^\top \bm B^{-1} \\
&=\Cov_\pi(U\mid Y=y)\bm A^\top \bm B^{-1}.
\end{aligned}
\]
Taking operator norms yields
\[
\|\nabla_y \mathsf m_\pi(y)\|_{\op}\le\|\Cov_\pi(U\mid Y=y)\|_{\op}
\|\bm A^\top \bm B^{-1}\|_{\op}.
\]
If the posterior covariance is uniformly bounded by $C_\pi$, then the
mean-value theorem gives, for any $y,y'\in\R^q$,
\[
    \|\mathsf m_\pi(y)-\mathsf m_\pi(y')\|_2
    \le
    C_\pi\|\bm A^\top \bm B^{-1}\|_{\op}\|y-y'\|_2.
\]
This proves the claim.

\subsection{Proof of Lemma~\ref{prop:examples_regular_priors}}
\label{sec:proof_examples_reg_prior}
Let us begin with the compactly supported prior class $\mathcal P_{\mathrm{comp}}(M)$ such that
\[
\mathbb P_{U \sim\pi}\left[\|U\|_2 \le M\right]=1, \quad \mbox{for all $\pi \in \mathcal P_{\mathrm{comp}}(M)$.}
\]
Then, for every $y$,
\[
    \|\Cov_\pi(U\mid Y=y)\|_{\op}\le
    \E_\pi\left[\|U\|_2^2\mid Y=y\right]
    \le M^2.
\]
Therefore, by \eqref{eq:nabla_m_pi},
\[
    \|\nabla_y \mathsf m_\pi(y)\|_{\op}
    \le
    M^2\|\bm A^\top \bm B^{-1}\|_{\op}
    \le
    \frac{M^2\|\bm A\|_{\op}}{b}, \quad \mbox{for all $\pi \in \mathcal P_{\mathrm{comp}}(M)$.}
\]
Thus $\mathsf m_\pi$ is globally Lipschitz in $y$, uniformly over all priors $\pi \in \mathcal P_{\mathrm{comp}}(M)$ supported on the same ball. This proves the claim for compactly supported
priors.

Next, suppose that $\pi=\mathcal N_r(m,\bm\Sigma)\in\mathcal P_{\mathrm{Gauss}}$. The posterior covariance in this setting is given  by
\[
\bm\Sigma_{U\mid Y}=\bm\Sigma-\bm\Sigma\bm A^\top\left(\bm A\bm\Sigma\bm A^\top+\bm B\right)^{-1}\bm A\bm\Sigma,
\]
and hence
\[
\bm 0\preceq\bm\Sigma_{U\mid Y}\preceq\bm\Sigma\preceq\overline\lambda\bm I_r.
\]
Thus, the posterior covariance is uniformly bounded by $\overline\lambda$. By Lemma~\ref{lem:posterior_mean_derivative}, this implies
\[
\|\nabla_y\mathsf m_\pi(y)\|_{\op}\le M_\Sigma\|\bm A^\top \bm B^{-1}\|_{\op}\le\frac{\overline\lambda\|\bm A\|_{\op}}{b},
\]
uniformly over $\pi \in \mathcal P_{\mathrm{Gauss}}$.
Therefore, the posterior mean in this setting is globally Lipschitz.

If $\pi \in \mathcal P_{\mathrm{logconv}}$ has density
\[
    d\pi(u)~\propto~ \exp\{-V(u)\}\,du
\]
with respect to Lebesgue measure on $\R^r$, where $V$ is twice continuously
differentiable and satisfies $\nabla^2 V(u)\succeq \mathrm c_0 \bm I_r$ for all $u\in\R^r$ and $\pi\in \mathcal P_{\mathrm{logconv}}$, for some $\mathrm c_0>0$. The posterior density of $U$ given $Y=y$ is proportional
to
\[
    \exp\left\{-V(u)-\frac12(y-\bm Au)^\top \bm B^{-1}(y-\bm Au)\right\}.
\]
The Hessian of the negative log-posterior is given by
\[
\nabla^2 V(u)+\bm A^\top \bm B^{-1}\bm A \succeq \mathrm c_0 \bm I_r.
\]
By the Brascamp-Lieb covariance inequality \citep{BrascampLieb1976},
\[
\left\|\Cov_{\pi}(U\mid Y=y)\right\|_\op \le \mathbb E_\pi\left[\left\|\left(\nabla^2 V(U)+\bm A^\top \bm B^{-1}\bm A\right)^{-1}\right\|_\op \,\Big|\, Y=y\right]\le\mathrm c^{-1}_0,
\]
uniformly in $y$. Hence
\[
\|\nabla_y\mathsf m_\pi(y)\|_{\op}\le\mathrm c^{-1}_0\|\bm A^\top \bm B^{-1}\|_{\op}\le\frac{\|\bm A\|_{\op}}{\mathrm c_0 b}, \quad \mbox{for all $\pi$.}
\]
Thus, strongly log-concave priors with curvature uniformly bounded below by $\mathrm c_0>0$ satisfy the required Lipschitz condition. 

Finally, consider the homoscedastic Gaussian mixture class where $\pi=\sum_{a=1}^{J}w_a\dnorm_r(m_a,\bm\Sigma)$ and $w_a>0$ for all $a \in [J]$. Conditional on the cluster label $C=a$, the posterior covariance is $\bm C:=\left(\bm\Sigma^{-1}+\bm A^\top\bm B^{-1}\bm A\right)^{-1},$
which does not depend on $a$ or $y$ and satisfies $\bm C\preceq\bm\Sigma\preceq\overline\lambda\bm I_r$. The corresponding posterior mean is
\[
\eta_a(y)=\bm C\left(\bm\Sigma^{-1}m_a+\bm A^\top\bm B^{-1}y\right).
\]
Consequently, $\eta_a(y)-\eta_b(y)=\bm C\bm\Sigma^{-1}(m_a-m_b),$ which is independent of $y$. Define
\[
L_0:=\sup_{\substack{\underline\lambda\bm I_r\preceq\bm\Sigma\preceq\overline\lambda\bm I_r}}\left\|\left(\bm\Sigma^{-1}+\bm A^\top\bm B^{-1}\bm A\right)^{-1}\bm\Sigma^{-1}\right\|_{\op}<\infty.
\]
Then $\max_{a,b}\|\eta_a(y)-\eta_b(y)\|_2\leq 2ML_0.$
By the conditional covariance decomposition,
\[
\Cov_\pi(U\mid Y=y)=\E_\pi\left[\Cov(U\mid Y=y)\mid Y=y\right]+
\Cov_\pi\left(\E_\pi[U\mid Y=y]\mid Y=y\right).
\]
The first term equals $\bm C$, while the second term is the covariance of the discrete random vector taking the value $\eta_a(y)$ with probability $q_a(y)$. Therefore,
\[
\Cov_\pi(U\mid Y=y)=\bm C+\sum_{a=1}^Jq_a(y)\bigl(\eta_a(y)-\overline\eta(y)\bigr)\bigl(\eta_a(y)-\overline\eta(y)\bigr)^\top,
\]
where $\overline\eta(y):=\sum_{a=1}^J q_a(y)\eta_a(y)=\mathbb E_\pi[U\mid Y=y].$ Since $\overline\eta(y)$ is a convex combination of
$\eta_1(y),\ldots,\eta_J(y)$, for every $a\in[J]$,
\[
\|\eta_a(y)-\overline\eta(y)\|_2\leq\max_{b\in[J]}\|\eta_a(y)-\eta_b(y)\|_2\leq 2ML_0.
\]
Consequently,
\begin{align}
\left\|\Cov_\pi(U\mid Y=y)\right\|_{\op}&\leq
\|\bm C\|_{\op}+\sum_{a=1}^Jq_a(y)\|\eta_a(y)-\overline\eta(y)\|_2^2 \leq
\overline\lambda+4M^2L_0^2,
\end{align}
uniformly over $y\in\mathbb R^r$ and over all priors in the mixture class.

\subsection{Proof of Theorem~\ref{thm:ols_estimator}}
First, we focus on the rescaled training embeddings.
For simplicity, we define
\[
\widehat{\bm A}_h:=\bigl(\widehat{\bm\Sigma}_{T,h}^{L}\bigr)^{-1}
\widehat{\bm D}_h^{-1},\qquad \bm A_h:=\bigl(\bm\Sigma_{T,h}^{L}\bigr)^{-1}
\bm D_h^{-1},\qquad \widehat{\bm F}_h:=\widehat{\bm F}_{T,h}.
\]
Since $r_h=O(1)$ for all $h \in [m]$,  using the definitions of $\wh{\bm\Sigma}_{T,h}^{L} \in \R^{r_h \times r_h}$ from Section~\ref{sec:feature_extract} and $\bm\Sigma_{T,h}^{L}\in \R^{r_h \times r_h}$ from Theorem~\ref{thm:w_2_amp}, we can invoke Theorem~\ref{thm:w_2_amp} with the pseudo-Lipschitz function $x \mapsto x^2$ to conclude that
\[
\lim_{n \to \infty}\|\wh{\bm\Sigma}_{T,h}^{L}-\bm\Sigma_{T,h}^{L}\|_F =0, \quad \mbox{almost surely for all $h \in [m]$.}
\]
Again, since $r_h=O(1)$ for all $h \in [m]$, $\bm\Sigma_{T,h}^{L}\succ0$ and $\bm D_h\succ0$, continuity of matrix inversion and Lemma~\ref{lem:first_order_nuisance} yields
\begin{align}
\|\widehat{\bm A}_h-\bm A_h\|_{\op}
\overset{\mathrm{a.s.}}{\rightarrow}0
\qquad\text{for every }h\in[m].
\label{eq:Ahat_A_convergence}
\end{align}
Next, we show that replacing the estimated normalization matrices $\{\wh{\bm A}_h:h \in [m]\}$ by their population limits $\{\bm A_h:h \in [m]\}$ does not affect the empirical average. In that direction, for each $i\in[n]$, define the collection of rescaled embeddings and their oracle counterparts
\begin{align}
\bm Q_{n,i}&:=\left((\widehat{\bm F}_1)_{i*}\widehat{\bm A}_1,
\ldots,(\widehat{\bm F}_m)_{i*}\widehat{\bm A}_m,
(\widetilde{\bm X}_1)_{i*},
\ldots,(\widetilde{\bm X}_{\widetilde m})_{i*}
\right),\\
\bm Q_{n,i}^{\circ}&:=
\left((\widehat{\bm F}_1)_{i*}\bm A_1,
\ldots,(\widehat{\bm F}_m)_{i*}\bm A_m,
(\widetilde{\bm X}_1)_{i*},
\ldots,(\widetilde{\bm X}_{\widetilde m})_{i*}\right).
\end{align}
Then, we have
\begin{align}
\frac{1}{n}\sum_{i=1}^n
\|\bm Q_{n,i}-\bm Q_{n,i}^{\circ}\|_2^2
&=\sum_{h=1}^m\frac1n\left\|\widehat{\bm F}_h
(\widehat{\bm A}_h-\bm A_h)\right\|_F^2 \notag\\
&\leq\sum_{h=1}^m\left\{\frac1n\|\widehat{\bm F}_h\|_F^2\right\}\|\widehat{\bm A}_h-\bm A_h\|_{\op}^2.
\label{eq:normalization_replacement_frob}
\end{align}
By Theorem~\ref{thm:w_2_amp} applied to the pseudo-Lipschitz function $(f_1,\ldots,f_m,\wt x_1,\ldots,\wt x_{\wt m})\mapsto\|f_h\|_2^2$, we get
\begin{align}
\label{eq:lim_sup_bounded}
\limsup_{n \to \infty}\frac1n\|\widehat{\bm F}_h\|_F^2=
\mathbb E\left[\left\|U_h(\bm M_{T,h}^{L})^\top
+Z_{T,h}(\bm\Sigma_{T,h}^{L})^{1/2}\right\|_2^2\right]
<\infty, \quad \mbox{almost surely.}
\end{align}
Combining the above with \eqref{eq:Ahat_A_convergence}, we get
\begin{align}
\frac1n\sum_{i=1}^n
\|\bm Q_{n,i}-\bm Q_{n,i}^{\circ}\|_2^2
\overset{\mathrm{a.s.}}{\rightarrow}0.
\label{eq:Q_Qoracle_frob}
\end{align}

Since $\psi$ is pseudo-Lipschitz, using \eqref{eq:pseudo_lips} and Cauchy--Schwarz inequality, there exists $C>0$ such that
\begin{align}
&\frac1n\sum_{i=1}^n\left|\psi(\bm Q_{n,i})-\psi(\bm Q_{n,i}^{\circ})\right| \\
&\quad\leq C\left\{
\frac1n\sum_{i=1}^n\left(1+\|\bm Q_{n,i}\|_2+\|\bm Q_{n,i}^{\circ}\|_2\right)^2\right\}^{1/2}
\left\{\frac1n\sum_{i=1}^n\|\bm Q_{n,i}-\bm Q_{n,i}^{\circ}\|_2^2\right\}^{1/2}.
\label{eq:PL_replacement_bound}
\end{align}
By \eqref{eq:lim_sup_bounded} and \eqref{eq:Q_Qoracle_frob}, we have
\[
\limsup_{n \to \infty}\frac{1}{n}\sum_{i=1}^n\|\bm Q_{n,i}\|^2_2 < \infty, \quad \limsup_{n \to \infty}\frac{1}{n}\sum_{i=1}^n\|\bm Q^\circ_{n,i}\|^2_2 < \infty, \quad \mbox{almost surely.}
\]
Therefore, by \eqref{eq:Q_Qoracle_frob} and \eqref{eq:PL_replacement_bound}, we get
\[
\limsup_{n \to \infty}\frac1n\sum_{i=1}^n\left|\psi(\bm Q_{n,i})-\psi(\bm Q_{n,i}^{\circ})\right|_2=0, \quad \mbox{almost surely.}
\] 

Now define the modified function
\begin{align}
\psi^\circ(f_1,\ldots,f_m,\wt x_1,\ldots,\wt x_{\wt m}):=\psi\left(f_1\bm A_1,\ldots,f_m\bm A_m,
\wt x_1,\ldots,\wt x_{\wt m}\right).
\end{align}
As $\bm A_1,\ldots,\bm A_m$ are fixed matrices with finite operator norm, $\psi^\circ$ is also pseudo-Lipschitz. Therefore, Theorem~\ref{thm:w_2_amp} gives
\begin{align}
\frac1n\sum_{i=1}^n\psi(\bm Q_{n,i}^{\circ})
\overset{\mathrm{a.s.}}{\rightarrow}
\mathbb E\Big[\psi\big(U_{T,1}\bm A_1,\ldots,U_{T,m}\bm A_m,\widetilde X_1,\ldots,\widetilde X_{\widetilde m}\big)
\Big], \quad \mbox{as $n \to \infty$.}
\label{eq:transformed_state_evolution}
\end{align}
where $\{U_{T,h}:h \in [m]\}$ are defined in Theorem~\ref{thm:w_2_amp}.

By the state-evolution identity $\bm M_{T,h}^{L}=\bm\Sigma_{T,h}^{L}\bm D_h,$ we obtain
\begin{align}
(\bm M_{T,h}^{L})^\top\bm A_h=\bm D_h\bm\Sigma_{T,h}^{L}
(\bm\Sigma_{T,h}^{L})^{-1}\bm D_h^{-1}=\bm I_{r_h}.
\end{align}
Therefore, 
\[
U_{T,h}\bm A_h\overset{d}{=}U_h+Z^{\rsc}_{T,h}(\bm\Sigma_{T,h}^{L})^{1/2}(\bm\Sigma_{T,h}^{L})^{-1}\bm D_h^{-1}.
\]
The Gaussian term in the right of the foregoing expression has covariance $\bm D_h^{-1}(\bm\Sigma_{T,h}^{L})^{-1}\bm D_h^{-1}$. Hence, \eqref{eq:training_debiased_joint_limit} follows.

Next, we focus on the OLS projected test embeddings $\{U^{\test,\ols}_h:h \in [m]\}$. In this case, one can use Theorem~\ref{thm:w_2_amp} and retrace the arguments in Theorem~5.4 of \cite{nandy2024multimodal} to conclude that $\{\wh U^{\test,\ols}_h:h \in [m]\}$ satisfy
\[
\wh U^{\test,\ols}_h \xrightarrow{d} U^\test_h+(\bm D^{-1}_h(\bm\Sigma^L_{T,h})^{-1}\bm D^{-1}_h)^{1/2}Z^{\test,\ols}_h, \quad \mbox{where $Z^{\test,\ols}_h \overset{\mathrm{i.i.d}}{\sim} \dnorm_{r_h}(0,\bm I_{r_h})$, for all $h \in [m]$.}
\]

\end{document}